\documentclass[sigconf,nonacm]{acmart}
\AtBeginDocument{%
  \providecommand\BibTeX{{%
    \normalfont B\kern-0.5em{\scshape i\kern-0.25em b}\kern-0.8em\TeX}}}

\newcommand\vldbdoi{10.14778/3551793.3551824}
\newcommand\vldbpages{2692 - 2705}
\newcommand\vldbvolume{15}
\newcommand\vldbissue{11}
\newcommand\vldbyear{2022}
\newcommand\vldbauthors{\authors}
\newcommand\vldbtitle{\shorttitle} 
\newcommand\vldbavailabilityurl{https://github.com/UIC-InDeXLab/FCS}
\newcommand\vldbpagestyle{empty} 

\usepackage{amsthm}
\usepackage{amsmath}
\usepackage{color}
\usepackage{enumitem}
\usepackage{subcaption}
\usepackage{caption}
\usepackage{array}
\usepackage{booktabs}
\usepackage{makecell}
\usepackage{footnote}  
\usepackage{graphicx}
\usepackage{comment}
\usepackage{todonotes}
\usepackage{tikz}
\usepackage{pgfplots}
\usepackage{multicol}
\usepackage{multirow}
\pgfplotsset{every tick label/.append style={font=\small}}

\usepackage{algorithmicx}
\usepackage{algorithm}
\usepackage[noend]{algpseudocode}
\usepackage{balance}
\algrenewcommand\algorithmicrequire{\textbf{Input:}}
\algrenewcommand\algorithmicensure{\textbf{Output:}}

\newtheorem{theorem}{Theorem} 
\newtheorem{lemma}[theorem]{Lemma}
\newtheorem{definition}{Definition}

\newcommand\abol[1]{\textcolor{orange}{(Abol) #1}}

\newcommand\ian[1]{\textcolor{purple}{(Ian) #1}}

\newcommand{\stitle}[1]{\vspace{2mm}\noindent{\bf #1}:}
\newcommand\eat[1]{}
\newcommand\new[1]{#1\xspace}

\usepackage{xspace}
\newcommand{\problem}{{\sc FairCS}\xspace}
\newcommand{\cs}{{\sc CSM}\xspace}
\newcommand{\scalable}{{\sc LP-SCALE}\xspace}

\newcommand{\Gee}{\mathcal{G}}
\newcommand{\gee}{\mathbf{g}}
\newcommand{\ee}{\mathbf{e}}
\DeclareMathOperator{\EX}{\mathbb{E}}
\newcommand{\eps}{\varepsilon}
\newcommand{\pr}{Pr}

\newcommand{\sharpP}{\mbox{\#P}}
\newcommand{\NP}{\mathsf{NP}}
\newcommand{\LP}{\mathsf{LP}}
\newcommand{\IP}{\mathsf{IP}}

\usepackage{balance}

\begin{document}

\fancyhead{}
\title{Maximizing Fair Content Spread via Edge Suggestion in Social Networks}
%

\author{Ian P. Swift}
\affiliation{\institution{University of Illinois Chicago}}
\email{iswift2@uic.edu}

\author{Sana Ebrahimi}
\affiliation{\institution{University of Illinois Chicago}}
\email{sebrah7@uic.edu}

\author{Azade Nova}
\affiliation{\institution{Google Brain}}
\email{azade@google.com}

\author{Abolfazl Asudeh}
\affiliation{\institution{University of Illinois Chicago}}
\email{asudeh@uic.edu}

\renewcommand{\shortauthors}{I. Swift et al.}

\begin{abstract}
Content spread inequity is a potential unfairness issue in online social networks, disparately impacting minority groups.
In this paper, we view friendship suggestion, a common feature in social network platforms, as an opportunity to achieve an equitable spread of content.
In particular, we propose to suggest a subset of potential edges (currently not existing in the network but likely to be accepted) that maximizes content spread while achieving fairness.
Instead of re-engineering the existing systems, our proposal builds a fairness wrapper on top of the existing friendship suggestion components.

We prove the problem is $\NP$-hard and inapproximable in polynomial time unless $\mathsf{P}=\NP$.
Therefore, allowing \new{relaxation of the fairness constraint}, we propose an algorithm based on $\LP$-relaxation and randomized rounding with fixed approximation ratios on fairness and content spread. We provide multiple optimizations, further improving the performance of our algorithm in practice.
Besides, we propose a scalable algorithm 
\new{that dynamically adds subsets of nodes, chosen via iterative sampling, and solves smaller problems corresponding to these nodes}. Besides theoretical analysis, we conduct comprehensive experiments on real and synthetic data sets. 
Across different settings, our algorithms found solutions with near-zero unfairness while significantly increasing the content spread. Our scalable algorithm could process a graph with half a million nodes on a single machine, reducing the unfairness to around 0.0004 while lifting content spread by 43\%.

\end{abstract}

\newcommand{\Snote}[1]{\textcolor{brown}{(Sana: #1)}}
\maketitle

\pagestyle{\vldbpagestyle}
\begingroup\small\noindent\raggedright\textbf{PVLDB Reference Format:}\\
\vldbauthors. \vldbtitle. PVLDB, \vldbvolume(\vldbissue): \vldbpages, \vldbyear.\\
\href{https://doi.org/\vldbdoi}{doi:\vldbdoi}
\endgroup
\begingroup
\renewcommand\thefootnote{}\footnote{\noindent
This work is licensed under the Creative Commons BY-NC-ND 4.0 International License. Visit \url{https://creativecommons.org/licenses/by-nc-nd/4.0/} to view a copy of this license. For any use beyond those covered by this license, obtain permission by emailing \href{mailto:info@vldb.org}{info@vldb.org}. Copyright is held by the owner/author(s). Publication rights licensed to the VLDB Endowment. \\
\raggedright Proceedings of the VLDB Endowment, Vol. \vldbvolume, No. \vldbissue\ %
ISSN 2150-8097. \\
\href{https://doi.org/\vldbdoi}{doi:\vldbdoi} \\
}\addtocounter{footnote}{-1}\endgroup

\ifdefempty{\vldbavailabilityurl}{}{
\vspace{.3cm}
\begingroup\small\noindent\raggedright\textbf{PVLDB Artifact Availability:}\\
The source code, data, and/or other artifacts have been made available at \url{\vldbavailabilityurl}.
\endgroup
}

\section{Introduction}\label{sec:intro}
Online Social networks have become an inseparable aspect of modern human life and a prominent medium for human interactions at scale.
These networks play an invaluable role in removing the physical communication barriers and enabling the spread of information across the world in real-time. However, like other recent technologies, the benefits of the social networks have not been for free, as they introduce new social challenges and complications that did not exist before. From privacy issues to misinformation spread and even to social media addiction, we keep hearing about the concerns and challenges specific to the nature of social networks.

Machine bias has recently been a focal concern in the social aspects of computer science and (big) data technologies~\cite{fairmlbook,asudeh2020fairly}, triggering extensive effort, including data pre-processing techniques~\cite{kamiran2012data,feldman2015certifying,calmon2017optimized,salimi2019interventional},  algorithm or model modification~\cite{kamishima2012fairness,zemel2013learning,zafar2015fairness}, and output post-processing techniques~\cite{kamiran2010discrimination,hardt2016equality,woodworth2017learning} to address these concerns.
Social networks are not an exception when it comes to unfairness issues~\cite{hargreaves2020filters,olteanu2019social,hargreaves2019fairness}.

Inequity in information spread across social networks is one of the potential unfairness challenges, disparately impacting minority groups~\cite{beilinson2020fairness, SNBias1,rahmattalabi2020fair,tsang2019group,yu2017fair,FarnadiFIM}.
While similar issues have been reported for cases such as biased advertisement over social networks~\cite{speicher2018potential}, fairness in content spread becomes critical when related to health-care~\cite{yonemoto2019gatekeeper}, news and valuable information spread~\cite{starinsky2021books}, job opportunities~\cite{yaseen2016influence}, etc. 

Despite its importance, fairness in content spread over social networks has been relatively less studied. Existing works take two directions to address the problem: 
(i) fairness-aware influence maximization (IM)~\cite{FarnadiFIM,rahmattalabi2020fair,tsang2019group, yu2017fair}, which aims to select the content spread seed nodes in a way that a combination of fairness and influence spread is maximized,
(ii) \cite{Masrour_Wilson_Yan_Tan_Esfahanian_2020, laclau2021all} which consider fairness in terms of topological properties of the network e.g., promoting inter-group connections in a network in order to alleviate unfairness caused by strong alignment within groups. 

In this paper, we consider friendship suggestion, a popular feature across social network platforms, where a subset of potential edges that currently do not exist in the network is suggested to the users.
We view this as an opportunity to achieve an equitable spread of content in social networks.
That is, instead of selecting a subset of potential edges that only maximize content spread \new{\cite{chaoji-cs, yu-contentspread, yang2019marginal}}, we aim to consider fairness when suggesting the edges.
Our proposal is different from the existing work on two angles.
First, unlike the works on IM that focus on selecting the \new{fair influential} seed nodes \new{with a fixed graph}, our focus is on {\em fair edge selection} \new{with fixed content source nodes}.
Second, 
unlike works that aim to directly make the {\em topology} of the network diverse and independent from the demographic information, we consider suggesting edges to achieve \new{a new objective:} fairness in {content spread}.

Existing social networks often use sophisticated algorithms for friendship suggestions. 
Re-engineering the existing algorithms are costly and perhaps impractical.
Instead, we propose \new{{\em a fair content spread component} that works with any possible friendship suggestion method}, taking their output as the set of potential edges and selecting a subset for the suggestion to achieve fairness.

We introduce Fair Content Spread Maximization (\problem) problem, where given a set of candidate edges, our goal is to select a subset such that: it contains a fixed number of incident edges to each node ($k$ friend suggestions to each node), satisfies fairness (defined on the probability that different demographic groups receive content~\cite{becker2020fairness,BEAMAN2018147,Masrour_Wilson_Yan_Tan_Esfahanian_2020}), and maximizes content spread. \new{To the best of our knowledge, we are the first work to consider fairness in this context.} Unfortunately, not only is the problem $\NP$-hard, but also impossible to approximate in $\mathsf{P}$ time unless $\mathsf{P}=\NP$.
By allowing approximation on content spread and \new{relaxation of the fairness constraint},
we propose a non-trivial randomized approximation algorithm for the \problem problem based on $\LP$-relaxation and randomized rounding~\cite{motwani1995randomized}.
Our algorithm provides constant approximation ratios on the content spread and fairness of its output. \new{We are the first Linear Program designed for content spread with approximation guarantees on fairness.}
We propose several optimizations that further help our algorithm be efficient in practice.
Having to solve an $\LP$, our original algorithm lacks to scale to very large settings.
To resolve this {issue}, 
we design a scalable algorithm based on \new{dynamically increasing the nodes via} sampling. Instead of solving one expensive-to-solve $\LP$ that covers the entire problem space, the algorithm \new{iteratively solves problems over subsets of nodes} with reasonably small sizes.

Our experiments confirm that our algorithms not only could find solutions with {\em near-zero unfairness} but, due to the practical effectiveness of $\LP$-relaxation and randomized techniques, outperformed \new{all baselines from a large breadth of related research \cite{yu-contentspread, chaoji-cs, yang2019marginal, zhu2021minimizing} in content spread.}
We also observed our algorithm \new{achieves} comparable results to the Optimal Brute Force method on very small graphs.
Our scalable algorithm could scale to within an order of {\em half a million nodes}, on a single work station in a reasonable time, \new{while none of the baselines we used were able to compute half a million nodes in under 24 hours,} confirming the effectiveness of our \new{scalable} approach. On half a million nodes, we observed a decrease from an initial unfairness of 3.9\% to 0.04\%, while the content spread increased by 43\%. 
\new{
We were able to run our original algorithm on a largest setting of 4000 nodes, achieving a lift of 57.99\% on content spread and an {\em unfairness of less than 0.0001}. Alternatively, our scalable algorithm still achieved a lift of 45.07\% and a unfairness of 0.0132, and also had a 489x speedup over our other algorithm. At 4000 nodes, no baseline could produce a higher lift, a lower unfairness, or had a lower runtime than our scalable solution.}

\stitle{Summary of Contributions}
\begin{itemize}[leftmargin=*]
    \item (Section \ref{sec:pre}) We formally introduce the \problem problem to {choose} a subset of candidate edges {which maximizes} content spread while being fair. \item (Section \ref{sec:pre:complexity}) We prove the problem is $\NP$-hard, and is inapproximable in polynomial time unless $\mathsf{P}=\NP$.
    \item (Section \ref{sec:approx}) \new{Approximating on content spread while relaxing the fairness constraint}, we provide a randomized approximation algorithm for the problem based on $\LP$-relaxation and randomized rounding, guaranteeing an approximation ratio of $(1-\ee^{-1})$ on content spread and $\frac{2\ee}{1-\ee^{-1}}$ on fairness.
    \item (Section \ref{sec:opt}) We provide multiple optimizations to improve the performance of our algorithm in practice.
    \item (Section \ref{sec:scaling}) We provide a scalable algorithm that,
    instead of solving one very-large $\LP$,
    \new{dynamically adds sampled nodes to the graph allowing the problem to be solved in multiple computationally easier problems}.
    \item (Section \ref{sec:exp}) Besides theoretical analyses, we conduct experiments on real data sets. Our experiments verify the effectiveness of our algorithms and their efficiency\new{, showing that our LP relaxation algorithm outperforms implementations of other algorithms in all but runtime, and that our scalable algorithm produces a fair result of less than $0.05\%$ and does so faster than all implementations even in very large settings.}
\end{itemize}

\section{Preliminaries}\label{sec:pre}

\subsection{Social Networks Model}
\stitle{Network Model} A social network is a model representing relationships between users in form of a graph $G(V,E)$, where $V$ is the set of users, and $E$ is the set of edges, either directed or undirected~\cite{Kleinberg2010, wasserman_faust_1994, scott-sage-2011}. 
Directed relationships represent one way connections such as citations, while undirected relationships represent shared connections such as ongoing correspondence. 
Our work has \new{targeted} {both} directed and undirected social networks.
We assume the nodes are associated with at least one sensitive attribute, such as {\tt \small race}, {\tt \small sex}, or {\tt \small age-group}
used for defining demographic groups, or simply groups\footnote{
While our motivation in this paper is on fairness over demographic groups, our techniques are not limited to those. In particular, the groups can be defined over any attributes of interest such as {\tt occupation} or {\tt political-affiliation}.}, as we shall explain in Section \ref{sec:pre:fairness}.

\stitle{Candidate Edges (Friendship)} Suggesting a set of {\em potential relations} that currently do not exist in the social network but a user may accept is an active area of research. Different techniques and algorithms (white or black-box) are used to find such a candidate set. A popular strategy is ``{\em friends of friends}''~\cite{yu-contentspread}, where only users that are at most within a two-hop distance from a user are candidates for friend suggestion. The intuition behind this model is that users who do not share a friend are unlikely to be friends. More sophisticated methods exist: candidate edges can be determined using an implicit social graph~\cite{roth-implicit}, \new{using overlapping ego-nets \cite{epasto2015ego}, by augmenting the network and detecting resulting communities \cite{hamid2014cohesion}, or with the addition of a peripheral system by analysis of the lifestyle of potential matches \cite{wang2014friendbook}}. Our work only expects the existence of an input candidate edge set, which can be determined by any of these methods. We refer to the set of candidate edges as $E_c$ where, $\forall e \in E_c, e \notin E$.

\stitle{Content Nodes} Social networks enable users to share content at an unprecedented scale. 
Information and marketing~\cite{tang-information-propagation}, user content such as images~\cite{Cha-Flickr}, opinions~\cite{cercel-opinion}, and even happiness~\cite{Fowlera2338} are examples of the content being shared across users in a cascading manner. While often there is a clear source from which the content originated~\cite{Cha-Flickr,gruhl2004information,lerman2006social}, some existing work detects content orignality~\cite{RajapakshaOriginator, adar2005tracking}, or targeted nodes for optimizations such as influence maximization ~\cite{tang-information-propagation,kempe-infmax}.
In this paper, we assume that, by any of the prior work, 
content source nodes are selected as a subset of $V_s \subset V$. 


\subsection{Cascade Model}
\label{sec:cascade}
In social networks, content spread follows a cascade model. Each neighbor recursively shares received info within its own neighborhood, and content propagates through the network to a wide user population.
Specifically, in {\em independent cascade} (IC)~\cite{kempe-infmax} model, the content nodes are active in a timestep $t_0$. Then at timestep $t_{i+1}$, each node $v_x$ that was activated in timestep $t_i$ has a probability $P_{x,y}$ of activating each unactivated node $v_y$ for which there is an edge $e_{x,y} \in E$. 
Computing the expected propagation using IC is known to be \sharpP-hard, and accurately estimating the content spread requires a combinatorially large time~\cite{chaoji-cs,Chen-Scalable}. Thus, models such as {\em Maximum Influence Path}~\cite{Chen-Scalable}, and {\em Restricted Maximum Probability Path}~\cite{chaoji-cs} have been proposed that (i) require polynomial-time computation and (ii) provide an approximation of IC.

{\em Maximum Influence Path} (MIP)~\cite{Chen-Scalable} takes into consideration the single path from content nodes to each other node for which the path has the maximum propagation probability of reaching the destination node.
Given a content node $v_j\in V_s$ and a node $v_i\in V\backslash V_s$, consider a path $q$ from $v_j$ to $v_i$, such that $\forall e\in q,~e\in E\cup E_p$, where $E_p\subseteq E_c$.
Given the probability $P_e$ on every edge $e$, the probability of content spread from a source to a destination along a path $q$ is $Pr(q) = \prod_{e \in q}P_e$.
Let $Q_{i}$ be the set of paths from a content node to the node $v_i\in V\backslash V_s$.
Following MIP, the probability of content reaching $v_i$ is equal to the probability that the node was reached along the path in $Q_{i}$ with the highest probability. That is, $\forall{v_i \in V\backslash V_s}$:
\begin{align} \small
    F_{MIP}(v_i, E_p) = \max_{q \in Q_i} Pr(q)
\end{align}

{\em Restricted Maximum Probability Path} (RMPP)~\cite{chaoji-cs} extends MIP by using the included candidate set, such that it considers only paths that include at most one edge from that set.
Under RMPP, $Q_{i}$ is the set of paths from a content node to the node $v_i\in V\backslash V_s$, such that $\forall q\in Q_i$ there exists at most one edge $e\in E_p$. $F_{RMPP}(v_i, E_p)$ is the expected content received by a node $v_i$ considering $E_p$ as the suggested edges\footnote{Unless mentioned otherwise, we simplify the notation $F_{RMPP}(.)$ to $F(.)$.}.
Following prior work~\cite{chaoji-cs,Chen-Scalable}, we use RMPP to design our randomized algorithm in Section \ref{sec:approx}.

\begin{table}[!t]
    \caption{Table of Notations}
    \centering \small
    \begin{tabular}{|c|c|}
        \hline
         {\bf Symbol}&{\bf Description}  \\ \hline
         $G(V,E)$& The social network with nodes $V$ and edges $E$\\ \hline
         $E_c$ & Set of candidate edges $\notin E$ \\ \hline
         $V_s$ & Content source nodes \\ \hline
         $E_p$ & Suggested candidate edges $\subseteq E_c$\\ \hline
         $k$ & Number of edges per node to be suggested\\ \hline
         $\Gee$ & The set of groups\\ \hline
         $\gee_i$ & The $i$-th group in $\Gee$ \\ \hline
         $\mathcal{N}(v_i)$ & Candidate edges $e_j\in E_c$ incident to the node $v_i\in V$ \\ \hline
         $p$ & Edge activation probability \\ \hline
         $\eps$& Content Spread Disparity (Unfairness) -- Equation~\ref{eq:fairness}\\ \hline
    \end{tabular}
    \label{tab:notations}
\end{table}

\stitle{Edge Suggestion to Maximize Content Spread}
Considering content nodes and a cascade model for information diffusion, a more sophisticated approach to friend suggestion is maximizing content spread in social networks~\cite{chaoji-cs}.
Even though candidate set $E_c$ are those friendships that are likely to be accepted by a user~\cite{yu-contentspread, roth-implicit, Feng-Personality}, there may be an overwhelming number of candidates for each user. Thus, a subset of $k$ edges is to be selected from the candidate set. We call the chosen set of edges, that \new{is} a subset of $E_c$ and has at most $k$ incident edges to each node, the set of {\em suggested edges} $E_p \subseteq E_c$. In~\cite{chaoji-cs}, the content spread maximization problem is:

\begin{definition}[Content Spread Maximization (\cs) Problem]\label{def:spread}
\small
Given a social network $G(V,E)$ 
and a set of candidate edges $E_c$,
let $\mathcal{E}_p \subseteq 2^{E_c}$ be the collection of subset of edges from $E_c$ such that each set $E_p\in \mathcal{E}_p$ contains at most $k$ incident edges on each node.
Find the set $E_p\in \mathcal{E}_p$ that maximizes the content spread.

\begin{equation}\small
    \hspace{24mm}\max_{E_p \in \mathcal{E}_p}~ \sum_{v_i \in V}F(v_i,E_p) 
\end{equation}
\end{definition}

\subsection{Fairness Model}\label{sec:pre:fairness}
Following the literature on fairness in social networks~\cite{becker2020fairness,BEAMAN2018147,Masrour_Wilson_Yan_Tan_Esfahanian_2020}, and other similar settings~\cite{asudeh2019designing,rahmattalabi2020exploring,accinelli2021impact,hannak2017bias,asudeh2020maximizing}, our fairness model is based on the notion of {\bf equity}.
In a social network, users can have {sensitive attributes} such as {\tt \small gender} and {\tt \small race} with values such as 
\{{\tt \small female, male}\} and \{{\tt \small black, white}\} that should be considered for fairness.
We identify demographic groups as the intersection of domain values of the sensitive attributes. For example, $\{${\tt \small white-male, white-female, black-male, black-female}$\}$ can be the groups defined as the intersection of {\tt \small race} and {\tt \small gender}. 
We use the notation $\Gee=\{\gee_1, \ldots, \gee_c\}$ to denote the set of groups.


Existing work highlights unfairness in social networks as 
different groups receive content with different probabilities~\cite{BEAMAN2018147}. 
We consider the notion of group fairness in social networks to assure that users from different groups receive content in an equitable manner.
We define fairness in content spread as Definition~\ref{def:fairness}.
%

\small
The content spread in a social network is fair, if 
for each pair of groups $\gee_i$ and $\gee_j$ in $\Gee$, the expected total content received by group $\gee_i$ proportional to its size is equal to the one by $\gee_j$ proportional to its size.
That is, $\forall \gee_i, \gee_j\in \Gee$:
\begin{equation*}
\frac{1}{|\gee_i|}\sum_{v_\ell \in \gee_i}F(v_\ell, E_p) = \frac{1}{|\gee_j|}\sum_{v_\ell \in \gee_j}F(v_\ell,E_p)
\end{equation*}
%
%
%
We introduce, the problem to maximize fair content spread.


\small
Given a social network $G(V,E)$, 
a set of candidate edges $E_c$, and the demographic groups $\Gee=\{\gee_1,\cdots,\gee_c\}$, 
let $\mathcal{E}_p \subseteq 2^{E_c}$ be the collection of subset of edges from $E_c$ such that each set $E_p\in \mathcal{E}_p$ contains at most $k$ incident edges on each node and $E_p$ satisfies fairness in content spread (Fair Content Spread).
Find the set $E_p\in \mathcal{E}_p$ that maximizes the content spread.

\subsection{Problem Complexity and Inapproximability}\label{sec:pre:complexity}

\begin{lemma}\label{lem:complexity}
The Fair Content Spread Maximization (\problem) Problem is $\NP$-hard.\footnote{The proofs of Lemma~\ref{lem:complexity} and Lemma~\ref{lem:3} are included in the appendix.}
\end{lemma}

A difference between the content spread maximization (\cs) and the \problem problem is that {for any given problem instance}, \cs always has a valid solution, but \problem problem may not have a valid solution for \new{a particular problem instance}. In proof of Theorem~\ref{th:inapprox} we show it is $\NP$-complete to determine if an instance of \problem problem has a valid solution, even when there are only two groups $\gee_1$, $\gee_2$. It consequently proves the inapproximability result for \problem problem as in Theorem~\ref{th:inapprox}. Following this negative result, none of the existing approximation algorithms for the \cs problem (including the greedy approach) works for \problem.
\begin{theorem}\label{th:inapprox}
There exists no polynomial approximation algorithm for the \problem problem, even when there are only two groups $\gee_1$ and $\gee_2$, unless $\mathsf{P}=\NP$.
\end{theorem}

\begin{figure}[!t]
    \centering
    \includegraphics[width=.32\textwidth]{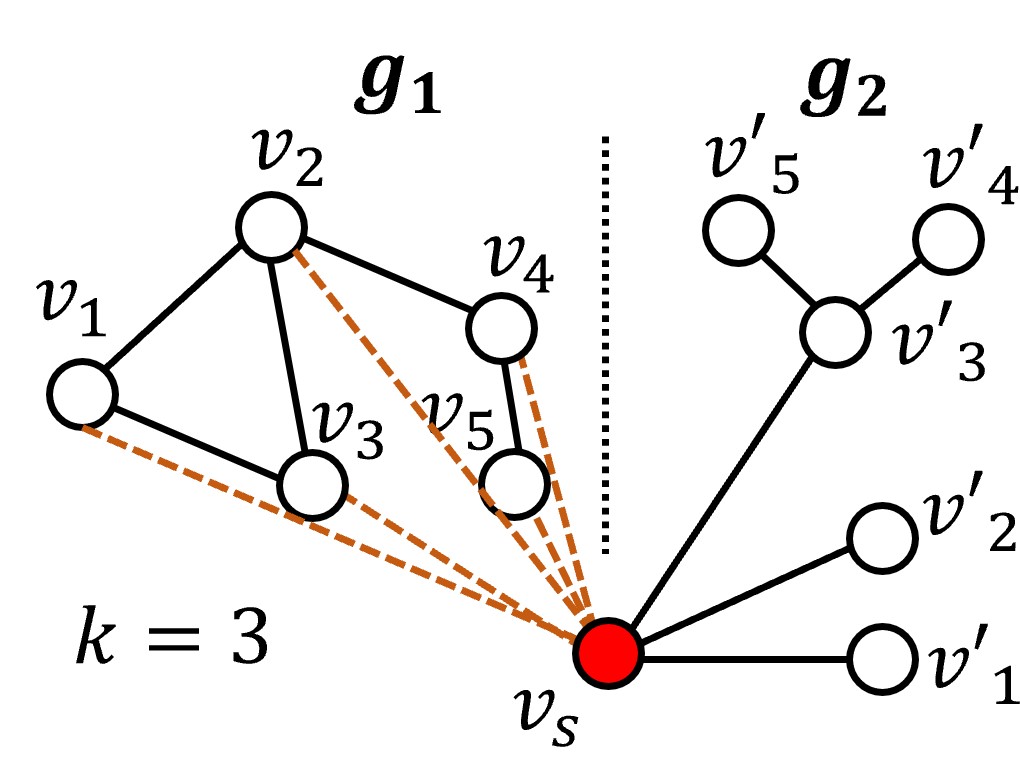}
    \caption{\small (Proof of Theorem~\ref{th:inapprox}) Illustration of a vertex-cover instance (left graph, $k=3$) to an instance of \problem problem. Candidate edges $E_c$ are specified in dashed brown.}
    \label{fig:proof-1}
\end{figure}

\begin{proof}
Here, we prove that an approximation algorithm for the \problem problem 
provides an exact solution for (the decision version of) the Vertex-Cover problem (VC). As a result, since VC is $\NP$-complete~\cite{clrs}, the polynomial approximation algorithm for \problem does not exist, unless $\mathsf{P}=\NP$.

Consider an instance of the VC problem, where given the graph $G(V,E)$, a value $k$, and $|V|=n$, the objective is to determine if there exists a subset of $k$ vertices that \new{are at least one endpoint of every edge} in $E$. For example, in Figure~\ref{fig:proof-1}, let the (induced sub-)graph with vertices $v_1$ to $v_5$ and the edges between them (the left graph) \new{be the input to a VC problem}. 
Given $k=3$, the answer to the VC is yes, since there are 3 vertices \new{(e.g. $\{v_1, v_2,v_5\}$)
that includes at least one endpoint of every edge of the graph}.

We transform a given instance of the VC to an instance of \problem problem over the graph $G''(V'',E'')$ as following:
\begin{itemize}[leftmargin=*]
    \item Add $V$ to $V''$ and $E$ to $E''$. Also, for each vertex $v\in V$ add a vertex $v'$ to $V''$.
    \item Let $\Gee=\{\gee_1,\gee_2\}$, where every vertex $v\in V$ belongs to $\gee_1$ and every vertex $v'$ belongs to $\gee_2$.
    \item Add a single content node $v_s$ to $V_s$. Add $k$ edges between $v'_1,\cdots,v'_{k}$ and $v_s$ (in Figure~\ref{fig:proof-1}, connected $v'_1, v'_2, v'_3$ to $v_s$). Add $(n-k)$ edges between $v'_{k+1},\cdots,v'_n$ and $v'_k$ (in Figure~\ref{fig:proof-1}, connected $v'_4, v'_5$ to $v'_3$). Add $n$ potential edges between $v_1,\cdots,v_n$ and $v_s$ to $E_c$. 
\end{itemize}
Now, as shown below, using $k''=k$ as the max incident edges suggested to each node, an approximation solution for the \problem problem on $G''$, provides an exact solution for VC on $G$.

\new{
To see why, consider the {universe} of valid solutions for this instance of the \problem problem.
A solution is valid if it selects at most k incident edges to each node and satisfies fairness constraints {\em but is not necessarily optimal} (i.e., it may not have the maximum content spread).
As we shall show in Equations~\ref{eq:proof:1} and \ref{eq:proof:2},
in order to satisfy the fairness constraints in the constructed \problem problem, 
any valid solution (optimal or approximate) {\em must} select 
$k$ edges between the nodes of a vertex-cover (of size $k$) and $v_s$.} 
\new{
The expected content received by $v'_1,\cdots,v'_n$ is:}
\begin{align}\label{eq:proof:1} \small
\nonumber    \EX\big[CS_{\gee_2}\big] &= \sum_{v'_i}F(v'_i, E_p) = \sum_{i=1}^kF(v'_i, E_p) + \sum_{i=k+1}^nF(v'_i, E_p) \\
    &= k\,p+ (n-k)\,p^2 
\end{align}
where $p$ is the edge activation probability.
To satisfy the fairness constraint, the expected content received by the two groups should be equal.
That is, $\EX\big[CS_{\gee_1}\big]=\EX\big[CS_{\gee_2}\big]$.
This happens when only the edges from a vertex-cover of size $k$ to $v_s$ are selected. In this situation, the expected content received by $v_1,\cdots,v_n$ is:

\begin{align}\label{eq:proof:2} \small
\nonumber    \EX\big[CS_{\gee_1}\big] &= \sum_{v_i}F(v_i, E_p) = \sum_{(v_i,v_s)\in E_p}F(v_i, E_p) + \sum_{(v_i,v_s)\notin E_p}F(v_i, E_p) \\
    &= k\,p+ (n-k)\,p^2 
\end{align}

\new{
So far, we proved that every VC of size $k$ corresponds to a valid solution.
Next, we prove by contradiction that every valid solution also corresponds to a VC of size $k$.
Suppose there exists a valid solution that does not correspond to a VC of size $k$. In this situation, there should exist at least one node in $\gee_1$ (one of the end-nodes of an uncovered edge) with a distance of at least 3 to $v_s$.
This reduces the expected content received by $\gee_1$ by at least $p^2-p^3$.
To resolve this reduction, the valid solution requires to select more than $k$ edges connected to $v_s$, which has contradiction with the requirement that at most $k$ incident edges on each node should be selected.
The one-to-one mapping between the VC of size $k$ and the valid solutions complete our proof: a polynomial approximation algorithm, finding a valid solution for \problem, solves (the decision version of) VC; hence, such an algorithm does not exist unless $\mathsf{P}=\NP$.
}
\end{proof}

Next, we provide an integer programming formulation for solving the problem exactly. This formulation, though not practical, helps us in designing our solution in \S~\ref{sec:approx:lpr}.

\stitle{Warm-up Integer Programming}
The optimal solution to \problem can be obtained by formulating it as an integer program ($\IP$). 

{\small
\begin{align*}
    \mbox{max}\mbox{imize}\hspace{5mm}&\sum_{v_i\in V}{x_i} \\
    \mbox{sub}\mbox{ject to}\hspace{5mm}&
    \forall v_i\in V,~ x_i \leq \max_{e_i\in E_P}(y_j p_{ij}) \\
    &\forall v_i\in V,~ \sum_{e_j \in \mathcal{N}(v_i)}y_j \leq k \\ 
    &\forall\gee_a, \gee_b\in \Gee,~ \frac{1}{|\gee_a|}\sum_{\ell \in \gee_a}{x_\ell} =\frac{1}{|\gee_b|}\sum_{\ell \in \gee_b}{x_\ell} \\
    &\forall v_i\in V,~  0 \leq x_i \leq 1;~
    \forall e_i\in E_c,~  y_j \in \{0,1\}  
\end{align*}}
In the above formulation, $x_i\in[0,1]$ is the probability that a node $v_i$ receives content, and the goal is to maximize this for all nodes. $y_j$ is $1$ if the edge $e_j\in E_c$ is added to $E_p$, 0 otherwise.
$p_{ij}$ is a pre-computed probability for each edge $e_j$ and node $v_i$, indicating 
the likelihood that $v_i$ receives content, if only $e_j$ is added to $E_p$.
The first set of constraints \new{computes} $x_i$ following the RMPP cascade model (\S~\ref{sec:cascade}).
The second constraints enforce that selected edges incident to each node $v_i$ should not be more than $k$ ($\mathcal{N}(v_i)$ is the candidate edges $e_j\in E_c$ incident to the node $v_i$).
Finally, the weighted sum of nodes activation is the overall content spread, and the weighted sum of node activation values between groups size enforce fairness.

Besides the fact the $\IP$ is $\NP$-complete, the max function in the $\IP$ constraints is not immediately solvable by any of the $\IP$ solvers and needs to be converted into a compatible mathematical format, which turns out to make the solution double exponential~\cite{caramia2008multi}.

\section{The Randomized Approximation Algorithm}\label{sec:approx}
Theorem~\ref{th:inapprox} proves the inapproximability of the \problem problem.
To find an efficient solution for the problem, we
allow \new{approximation on content spread and relaxation of the fairness constraint}.
We define a $(1+\eps)$-approximation ($\eps\geq 0$) on fairness, as following: 

\begin{align}\label{eq:fairness} \small
\forall \gee_i, \gee_j\in \Gee ~:~
(1+\eps)^{-1} \leq \frac{\frac{1}{|\gee_i|}\sum_{v_\ell \in \gee_i}F(v_\ell, E_p)}{\frac{1}{|\gee_j|}\sum_{v_\ell \in \gee_j}F(v_\ell, E_p)} \leq (1+\eps)
\end{align}
Note that $\eps=0$ refers to the maximum fairness, and it decreases as $\eps$ increases. 
Conversely, the disparity of content spread is $\eps$. In the rest of the paper and in our experiments, we refer to $\eps$ as {\em disparity} or unfairness.

Our algorithm is $\LP$-relaxation followed by randomized rounding.
The first step towards developing the algorithm is to formulate the problem as $\IP$. 
At a high level, the algorithm has three steps (Algorithm~\ref{algo:LP}): step 1 computes the values and necessary inputs for the $\LP$ solver; in step 2 it forms an $\LP$-relaxation and solves it; finally, in step 3 it randomly rounds the output of the solver into an approximation solution for the \problem problem.
\begin{algorithm}[t]
\small
    \caption{ $\LP$-relaxation (LP-APPROX)}\label{algo:LP}
    \hspace*{\algorithmicindent} \textbf{Input} : $G(V, E), V_s, E_c, k, \mathcal{G};~$
    \textbf{Output} : suggested edges $E_p$
    \begin{algorithmic}[1]
        \State $\mathcal{S}' \leftarrow$ {\sc ShortestDistances}$(G(V, E), E_c)$ {\tt \small // Algorithm \ref{algo:shortestpaths} in Appendix}
        \Statex {\tt \small // $S_{ir} = \{e_j~|~(e_j,d_{ij})\in S'_i\in\mathcal{S}',~d_{ij}\leq r\}$}
        \State $\langle\mathbf{x}^*,\mathbf{y}^*\rangle \leftarrow$ Solve the LP in Figure~\ref{fig:LPRelax}
        \State {\bf return} {\sc ROUND} $(\mathbf{y}^*, G(V, E), k, V_s, \Gee)$ {\tt \small // Algorithm~\ref{algo:round}}
    \end{algorithmic}
\end{algorithm}

\begin{figure}[!tb]
    \framebox[\columnwidth]{
    \parbox{\columnwidth}{
    {\small
    \begin{align*}
        ~\mbox{max}&\mbox{imize}~ \sum_{r=1}^{r_m} \delta_r \sum_{v_i\in V} {x_{ir}}\\
        \mbox{sub}&\mbox{ject to} &~\\
        &\forall v_i\in V,r\in[1,r_m],~x_{ir} \leq \sum_{e_j \in S_{ir}} y_j\\
        &\forall  v_i\in V,~  \sum_{e_j \in \mathcal{N}(v_i)}y_j \leq k\\
        &\forall\gee_a, \gee_b\in \Gee,~ \frac{1}{|\gee_a|}\sum_{r=1}^{r_m} \delta_r\sum_{v_\ell \in \gee_a}{x_{\ell r}} = \frac{1}{|\gee_b|}\sum_{r=1}^{r_m} \delta_r\sum_{v_\ell \in \gee_b}{ x_{\ell r}}\\
        &\forall v_i\in V,r\in[1,r_m],~ 0 \leq x_{ir} \leq 1\\
        &\forall  e_j\in E_c,~  0 \leq y_j \leq 1
    \end{align*}}}
    }
    \caption{$\LP$-relaxation (relaxed $x_{ir} \in \{0,1\}$ and $y_j \in \{0,1\}$) used for solving \problem problem.}
    \label{fig:LPRelax}
\end{figure}

\subsection{$\LP$-Relaxation}\label{sec:approx:lpr}

Instead of the $\IP$ formulation discussed in \ref{sec:pre:complexity}, we consider an alternative formulation inspired from~\cite{GUNEY2019589} for the $\LP$-relaxation.
A major change in the $\IP$ formulation is that we replace the variable $x_i$, representing the probability of content receiving at node $v_i$, by the weighted sum of a collection of {\em binary} variables $x_{ir}$, which is $1$ if there exists a path of length at most $r$ from a content node to $v_i$, and $0$ otherwise.
Let $p$ be the probability that an edge gets activated during the cascade process. Then, the probability that content reaches a node $v_i$ using a path of length $r$ is $p^r$.
Given a node $v_i$, let $\ell$ be its shortest distance from a content node. Furthermore, let $r_m$ be the maximum distance for which all nodes $V\backslash V_s$ are reachable from a content node.
First, $\forall r\geq \ell$, $x_{ir}=1$, while $\forall r< \ell$, $x_{ir}=0$.
We define the values $\delta_r$, such that $x_i = \sum_{r=1}^{r_m} \delta_r x_{ir}$.
We shall explain the derivation of $\delta_r$ values in \ref{sec:approx:pre}. Using the above transformation, the objective function can be rewritten as 
$\sum_{v_i\in V} x_i = \sum_{r=1}^{r_m} \delta_r \sum_{v_i\in V} {x_{ir}}$.

Next, for every node $v_i$ and a distance $r$, we precompute the set $S_{ir}$, as the set of edges that if added to $E_p$, the distance to $v_i$ from a content node will be at most $r$.
Introducing the new variables, the $\LP$-relaxation formulation is shown in Figure~\ref{fig:LPRelax}.
Recall that in the $\IP$ formulation $y_j$ is a binary attribute that is $1$ when the edge $e_j\in E_c$ gets selected. Similarly, $x_{ir}$ is $1$ only if there exists a path of length $r$ to $v_i$ from a content node. That is, if at least one of the edges in $S_{ir}$ has been selected. This is enforced by the first set of constraints in the $\LP$ formulation.
The second constraints enforce $k$ incident edges to each node, while the third one enforces fairness. 

After constructing the $\LP$ as in Figure~\ref{fig:LPRelax}, the algorithm finds its optimal values $\mathbf{x}^*=\{x^*_{ir}~|~\forall v_i\in V, r\in [1,r_m]\}$ and $\mathbf{y}^* = \{y^*_j~|~e_j\in E_c\}$ and \new{passes} those to step 3 of Algorithm~\ref{algo:LP} for rounding, as we shall explain next in \ref{sec:approx:round}.




\subsection{Randomized Rounding}\label{sec:approx:round}
Next step after solving the $\LP$ in Figure~\ref{fig:LPRelax} is to round the continuous values for the edge selection variables $\mathbf{y}^*=\{y^*_j~|~e_j\in E_c\}$ into binary values $\mathbf{y}^+=\{y^+_j~|~e_j\in E_c\}$.
An edge $e_j\in E_c$ is selected for suggestion if $y^+_j=1$.
We follow a randomized rounding for rounding $\mathbf{y}^*$ to $\mathbf{y}^+$, using $\mathbf{y}^*$ as the probability to select edges for suggestion. That is, the probability to suggest edge $e_j\in E_p$ is $P(y^+_j=1) = y^*_j$.
To do so, Algorithm~\ref{algo:round} draws a number uniformly at random from the range $[0,1]$. If the generated sample is less than $y^*_j$,  $e_j$ is added to $E_p$ i.e. $y^+_j=1$. 
As we shall study in \ref{sec:aapprox:analysis}, following this rounding approach, 
the expected number of edges selected per each node is $k$, and the algorithm guarantees constant approximation ratios both on fairness and the content spread.

The bottle-neck in the time complexity of Algorithm~\ref{algo:LP} is solving the $\LP$, i.e., the second step.
On the other hand, the rounding of the variables is fast (rounding the values of $y^*$ to $y^+$ is in $O(|E_c|)$).
Since the rounding is randomized, 
it is reasonable to repeat the rounding step many times to boost the success probability.
Let $iter_m$ be the number of times the rounding is repeated and let 
$E_{p_1} \cdots, E_{p_{iter_m}}$ be the results obtained over different repetition of randomized rounding.
From these, the algorithm considers the ones that have fairness close to the most-fair one found, and returns the one that maximizes the content spread (Lines 5-10 of Algorithm~\ref{algo:round}).

\begin{algorithm}[t]
\small
    \caption{ Randomized Rounding ({\sc ROUND})}\label{algo:round}
    \hspace*{\algorithmicindent} \textbf{Input} : $\mathbf{y}^*, G(V, E), V_s, \Gee;~$
    \textbf{Output} : selected edges $E_p$
    \begin{algorithmic}[1]
        \State $sol\gets \{\}$
            \For {$i\gets 1$ to $iter_m~$} 
            \For {$e_j\in E_c$}
                \State {\bf if} {\sc Rand}$(0,1)\leq y^*_j$ {\bf then} $y^+_j=1$, i.e., add $e_j$ to $E_{p_i}$
            \EndFor
                \State $CS\gets \sum_{v_j \in V}F(v_j,E_{p_i})$ {\tt // content spread}
            \State $\mathcal{F}\gets \max_{\gee_i,\gee_j \in \Gee} \Big(\frac{\frac{1}{|\gee_i|}\sum_{v_\ell \in \gee_i}F(v_\ell, E_{p_i})}{\frac{1}{|\gee_j|}\sum_{v_\ell \in \gee_j}F(v_\ell, E_{p_i})}\Big)$ {\tt // fairness}
            \State add $\langle E_{p_i}, \mathcal{F}, CS\rangle$ to $sol$
        \EndFor
        \State $\mathcal{F}_{min}\gets \min_{\mathcal{F}}(sol)$
        \State candidates $\leftarrow\{\langle E_{p_i}, \mathcal{F}, CS\rangle\in sol~|~\mathcal{F} - \mathcal{F}_{min}\leq \epsilon\}$
        \State {\bf return} argmax$_{CS}$ (candidates)
    \end{algorithmic}
\end{algorithm}

\subsection{Constructing $\LP$ Inputs}\label{sec:approx:pre}

For every node $v_i$ and a distance $r$, we need to precompute the set $S_{ir}$, of edges in $E_c$, that makes $v_i$ at most $r$ hops away from a content node.
Again, under RMMP, a path from a content node to $v_i$ can contain at most one edge from $E_c$.
Hence, to compute the sets $\mathcal{S}=\{S_{ir}~|~v_i\in V\backslash V_s, r\leq r_m\}$, we compute the shortest distance from any content node to every node in the graph with edge $e_j\in E_c$.
That is, $\forall_{e_j \in E_c}$, we compute all-pair shortest paths from content nodes in graph $G(V,E\cup\{e_j\})$. Suppose $d_{ij}$ is the shortest path length from content nodes to a node $v_i$ in graph $G(V,E\cup\{e_j\})$.
Then $e_j$ is added to all sets $S_{ir}$ where $r\geq d_{ij}$, because $e_j$ enables a path with length $\geq d_{ij}$ to $v_i$. 
Instead of repeating $e_j$ in all sets $S_{ir}, r\geq d_{ij}$, one can save space by defining the intermediate sets $\mathcal{S}' = \{S'_{ir}~|~v_i\in V\backslash V_s, d_{ij} = r\}$, where $S'_{ir}$ contains the edges in $E_p$ that enable a shortest path of length $r$ to $v_i$.
Then $S_{ir} = \cup_{r'\leq r} S'_{ir'}$.

Lastly, in order to form the $\LP$, we need to compute the values of $\delta_r$, such that 
$\sum_{r=1}^{r_m} \delta_r x_{ir}$ is equal to $x_i$, the probability that a node $v_i$ receives content. Let $d$ be the length of the shortest path from a content node to a node $v_i$, and let $p$ be the activation probability for each edge.
Then, the probability that it receives content (using RMPP) is $x_i=p^d$.
We want to compute the values of $\delta_r$, in a way that for every shortest-path length $d$, $\sum_{r=1}^{r_m} \delta_r x_{ir} = p^d$.
For cases where $d=r_m$, $x_{ir_m}=1$ and $x_{ir}=0, \forall r<r_m$.
As a result: $\sum_{r=1}^{r_m} \delta_r x_{ir} = \delta_{r_m}$. Hence, $\delta_{r_m} = p^{r_m}$.


When $d=r_m-1$, $\sum_{r=1}^{r_m} \delta_r x_{ir} = \delta_{r_{m-1}} + \delta_{r_{m}} = p^{r_m-1}$. Therefore, $\delta_{r_{m-1}} = p^{r_m-1} - p^{r_m}$.

Similarly, $\forall r<r_m$, $\sum_{r=1}^{r_m} \delta_r x_{ir}$ can be written as 
\begin{align*} \small
   &\sum_{r=1}^{r_m} \delta_r x_{ir} = \sum_{r=d}^{r_m} \delta_r = \delta_d + \sum_{r=d+1}^{r_m} \delta_r = p^d \\
   &\sum_{r=d+1}^{r_m} \delta_r = p^{d+1}
   \Rightarrow \delta_d = p^d - p^{d+1}
\end{align*}

In summary,

\begin{align} \small
    \delta_r = 
    \begin{cases}
        p^{r_m} & r=r_m \\
        p^r - p^{r+1} & r<r_m
    \end{cases}
\end{align}

\subsection{Approximation Analysis}\label{sec:aapprox:analysis}

\begin{lemma}\label{lem:3}
The expected number of selected edges by Algorithm~\ref{algo:LP} per each node $v_i\in V$ is $k$.
\end{lemma}


\begin{theorem}
If \problem problem has at least one valid solution, then the $\LP$-relaxation algorithm (Algorithm~\ref{algo:LP}) for solving it satisfies the approximation ratio of $(1-\ee^{-1})$ on the content spread and $\frac{2\ee}{1-\ee^{-1}}$ on fairness.
\end{theorem}

\begin{proof}

Let $\{x^*_{ir}~|~\forall v_i\in V, r\in [1,r_m]\}$ and $\{y^*_j~|~e_i\in E_c\}$ be the optimal values of $\LP$, shown in Figure~\ref{fig:LPRelax}.
Besides, let $\{x^+_{ir}~|~\forall v_i\in V, r\in [1,r_m]\}$ and $\{y^+_j~|~e_i\in E_c\}$
be the approximation values for the \problem problem.
Let us first compute the approximation ratio on the content spread.

\begin{multline} \small
    \EX\Big[x_{ir}^+\Big] = \pr(x_{ir}^+=1) = 1- \pr(x_{ir}^+=0)
    \\ = 1 - \prod_{e_j \in S_{ir}}(1 - y_j^*)
    \geq 1 - \prod_{e_j \in S_{ir}}e^{-y_j^*} = 1 - \ee^{-\sum_{e_j \in S_{ir}}y_j^*}\\ 
    \geq 1 - \ee^{-x_{ij}^*}\geq (1 - \ee^{-1})x_{ij}^*
\end{multline}
Let $\mathsf{APR}$ be the content spread of the approximation algorithm output and let $\mathsf{OPT}$ be the optimal content spread.

\begin{align} \small
    \nonumber \EX[&\mathsf{APR}]= \EX\Big[ \sum_{r=1}^{r_m} \delta_r \sum_{v_i\in V} {x^+_{ir}}\Big] = \sum_{r=1}^{r_m} \delta_r \sum_{v_i\in V} { \EX\big[x^+_{ir}\big]} \\
    &\geq \sum_{r=1}^{r_m} \delta_r \sum_{v_i\in V} {(1 - \ee^{-1})x_{ij}^*} = (1 - \ee^{-1})\sum_{r=1}^{r_m} \delta_r \sum_{v_i\in V} {x_{ij}^*}
\end{align}
Since the optimal solution for $\LP$ is an upper-bound on the $\IP$ solution (optimal content spread),

\begin{align}\label{eq:lower} \small
    &\sum_{r=1}^{r_m} \delta_r \sum_{v_i\in V} {x_{ij}^*} \geq \mathsf{OPT} \\
    \nonumber\Rightarrow&~ \EX[\mathsf{APR}] \geq (1 - \ee^{-1}) \sum_{r=1}^{r_m} \delta_r \sum_{v_i\in V} {x_{ij}^*} \geq (1-\ee^{-1})~ \mathsf{OPT}
    \\ \Rightarrow&~ \frac{\EX[\mathsf{APR}]}{\mathsf{OPT}} \geq (1-\ee^{-1})
\end{align}


    

Next, we prove the approximation ratio on fairness.
Inspired from \cite{asudeh2020maximizing}, we first provide an upper bound on $\EX[x_{ir}^+]$ that bounds the value of fairness across different groups.

Looking at Figure~\ref{fig:LPRelax}, if $\sum_{e_j \in S_{ir}} y_j\geq 1$, then $x_{ir}^*=1$. Hence, (case I) $\EX[x_{ir}^+]\leq 1 = x_{ir}^*$. \\
For cases where $\sum_{e_j \in S_{ir}} y_j< 1$, we know $x_{ir}^* = \sum_{e_j \in S_{ir}} y_j$.
Now, if $\exists~ e_j\in S_{ir}$, such that $y_j^*\leq \frac{1}{2\ee}$, then $x_{ir}^*\geq \frac{1}{2\ee}$. Therefore, (case II) $\EX[x_{ir}^+]\leq 1\leq 2\ee x_{ir}^*$.\\
The only case left is when $\forall e_j\in S_{ir}$, $y_j^*\leq \frac{1}{2\ee}$.
Since for this case, $y_j^*\in[0,1]$,
using the second-order Taylor series expansion of $\ee^{-2y_j^*}$, we know:

\begin{align*}\small
\ee^{-2y_j^*} \leq 1 - 2y_j^* + \frac{\ee}{2}(2y_j^*)^2 \leq 1 - y_j^*
\end{align*}

Therefore (case III),

\begin{multline}\small
    \EX\Big[x_{ir}^+\Big] = \pr(x_{ir}^+=1) = 1- \pr(x_{ir}^+=0)
    \\ = 1 - \prod_{e_j \in S_{ir}}(1 - y_j^*)
    \leq 1 - \prod_{e_j \in S_{ir}}e^{-2y_j^*} = 1 - \ee^{-\sum_{e_j \in S_{ir}}2y_j^*}\\ 
    \leq 1 - \ee^{-2 x_{ir}^*} \leq 1 - (1 - 2 x_{ir}^*) =  2 x_{ir}^*
\end{multline}

From cases I, II, and III, we can conclude that $\EX\Big[x_{ir}^+\Big]\leq 2\ee x_{ir}^*$.
As a result,

\begin{align}\label{eq:upper}
    \nonumber \EX&\Big[ \sum_{r=1}^{r_m} \delta_r \sum_{v_i\in V} {x^+_{ir}}\Big] = \sum_{r=1}^{r_m} \delta_r \sum_{v_i\in V} {\EX\big[x^+_{ir}\big]} \\
    &\leq \sum_{r=1}^{r_m} \delta_r \sum_{v_i\in V} {2\ee \,x_{ij}^*} = 2\ee\sum_{r=1}^{r_m} \delta_r \sum_{v_i\in V} {x_{ij}^*}
\end{align}

We now use the lower-bound found in Equation~\ref{eq:lower} and the upper-bound in Equation~\ref{eq:upper} to bound the unfairness (Equation~\ref{eq:fairness}) between any pair of groups $\gee_a,\gee_b\in \Gee$.

From the fairness constraint in the $\LP$ formulation (the third constraints Figure~\ref{fig:LPRelax}), we know that

\begin{align*}\small
    \frac{1}{|\gee_a|}\sum_{r=1}^{r_m} \delta_r\sum_{v_\ell \in \gee_a}{x_{\ell r}^*} = \frac{1}{|\gee_b|}\sum_{r=1}^{r_m} \delta_r\sum_{v_\ell \in \gee_b}{x_{\ell r}^*}
\end{align*}
Using the upper-bound and the lower-bound in Equations~\ref{eq:lower} and ~\ref{eq:upper}, for every group $\gee_a\in \Gee$, we know:

\begin{align*}\small
    &\frac{1}{|\gee_a|}\sum_{r=1}^{r_m} \delta_r\sum_{v_\ell \in \gee_a}{x_{\ell r}^+} \geq (1-\ee^{-1})\frac{1}{|\gee_a|}\sum_{r=1}^{r_m} \delta_r\sum_{v_\ell \in \gee_a}{x_{\ell r}^*} 
    \\
    \mbox{and }&\frac{1}{|\gee_a|}\sum_{r=1}^{r_m} \delta_r\sum_{v_\ell \in \gee_a}{x_{\ell r}^+} 
    \leq 2\ee \frac{1}{|\gee_a|}\sum_{r=1}^{r_m} \delta_r\sum_{v_\ell \in \gee_a}{x_{\ell r}^*}
\end{align*}
As a result, for all $\gee_a,\gee_b\in \Gee$,
\begin{align}
    \frac{1-\ee^{-1}}{2\ee} \leq 
    \frac{\frac{1}{|\gee_a|}\sum_{r=1}^{r_m} \delta_r\sum_{v_\ell \in \gee_a}{x_{\ell r}^+}}{\frac{1}{|\gee_b|}\sum_{r=1}^{r_m} \delta_r\sum_{v_\ell \in \gee_b}{x_{\ell r}^+}} \leq \frac{2\ee}{1-\ee^{-1}}
\end{align}
\end{proof}
\section{Practical Optimizations}\label{sec:opt}

\stitle{LP-Advanced Algorithm (Extension Beyond RMPP)}
Our approximation algorithm is designed using RMPP as the cascade model. 
In RMPP, the content spread to a node is the \emph{maximum} probability of a path from a content node to this node containing at most one edge from $E_c$.
This means, if by including edge $e_j'$ the content received by each node is less than or equal to the content received by each node by including $e_j$, then if $e_j$ is in the chosen set $E_p$, including $e_j'$ does not change the maximum probability of these paths. Under $RMPP$, the edge $e_j'$ will not be included if $e_j$ is included. Alternatively, under the MIP and IC models, the content spread to a node is determined by the path with the highest probability (for MIP) or the probability over all paths (for IC) such that the paths contain \emph{any number} of edges from $E_c$. After adding $e_j$, a path that includes both $e_j$ and $e_j'$ may have an additional effect under MIP and IC when having no effect under RMPP. As discussed in Section \ref{sec:approx}, Algorithm~\ref{algo:LP}, using RMPP, finds a solution relatively efficiently with theoretical guarantees on accuracy. In order to find edges that would be included in MIP and IC but not RMPP, we propose to \emph{rerun} Algorithm~\ref{algo:LP} multiple times, allowing it to first calculate a result $E_p$, then redefine $E \leftarrow E \cup E_p$ and $E_c \leftarrow E_c \setminus E_p$. By moving $E_p$ from $E_c$ to $E$, in the next iteration $e_j'$ can provide addition content to nodes, even under RMPP. To mantain $k$ expected incident edges in this modification, we only find $k'_i$ edges for a node $v_i$, determined by the number of incident edges to $v_i$ in $E_p$ from all previous iterations
This process is shown in Algorithm~\ref{algo:lp-iter}.

\new{Additional optimizations to limit the variables being solved in an Linear Program and to dynamically update the graph are included in the appendix.}

\begin{algorithm}[t]
\small
    \caption{{\sc LP-Advanced}}\label{algo:lp-iter}
    \hspace*{\algorithmicindent} \textbf{Input} : $G(V, E), V_s, E_c, k, \Gee, cutoff;~$
    \textbf{Output} : suggested edges $E_p$
    \begin{algorithmic}[1]
    \State $sol \gets \{\}$; 
    $\forall_{v_i \in V} ~ k'_i \gets k$; 
    Copy $E_c$ to $E_c^*$
    \Repeat
        \State{$E_p \gets$ {\sc LP-Approx}$(G(V ,E), V_s, E_c^*, k', \Gee)$} {\tt \small // Using MIP for rounding} 
        \State $sol \leftarrow sol \cup E_p$;
        $E \gets E \cup E_p$
        \State{$\forall_{v_i \in V} ~ k'_i \gets \max(k - |\{e_{u,v} \in sol~|~u == i~||~v == i\}|, 0)$}
        \State{$E_c^* \gets \{e_{u,v} \in E_c^* \setminus E_p~|~k'_u > 0~\&\&~k'_v > 0\}$}
    \Until{$|E_p| == 0$ OR a small ratio of $E_c$ left}\\
    \Return{$sol$}
    \end{algorithmic}
    
\end{algorithm}
\section{Scalability}\label{sec:scaling}

Though polynomial, solving $\LP$ problems \new{is} time-taking, and usually impractical for large settings with tens of thousands or more variables.
As we demonstrate in our experiments in Section \ref{sec:exp}, our $\LP$-relaxation algorithm is not scalable for such settings.
We introduce a different approach for large settings and, instead of solving one inefficient $\LP$, use a heuristic to break down the problem into {\em multiple instances of $\LP$ with reasonably small sizes}.

Our extension is based on the observation that 
{\em the contribution of a path to the content spread exponentially reduces as its length increases,} and the nodes that are far from the content nodes have a minimal impact in the content spread.
Solving the problem for the high-impact nodes can be viewed as an approximation heuristic for solving the full problem.
On the other hand, only \new{considering the impact of a small portion of the graph on the overall problem,}
can significantly reduce the size of $\LP$, remarkably increasing the algorithms efficiency.
We propose an iterative approach traversing the graph that, at a high level, starts from a local neighborhood around the content nodes, forms a reasonably small subgraph in the neighborhood, and solves it using the $\LP$-relaxation algorithm. Then we gradually add nodes, branching away for the content nodes, and solving \new{the overall problem by making changes to just those nodes} using the same $\LP$-relaxation technique.

\stitle{Subgraphs} 
Solving the \problem problem using Algorithm~\ref{algo:LP} or Algorithm~\ref{algo:lp-iter}, for a candidate edge set $E_c$, and each vertex whose expected content is improved by an edge in $E_c$, $V^* \subset V$, requires solving a $\LP$ of size $n = \Omega(|V^*| + |E_c|)$.
Instead of solving an $\LP$ with $n$ variables, taking $O(n^{2.37})$ time using \cite{cohen2021solving}, we solve a series of $b$ problems of size $\bar{n} << n$, with a runtime of $O(b \times \bar{n}^{2.37})$. 

We design a method to \emph{incrementally} add nodes to the graph such that processing the graph in stages provides a meaningful result while limiting computation time. We start by dividing the graph into $b$ subsets of nodes, each containing $N \over b$ nodes. In each iteration $i$, we consider the set of nodes $V'$ defined as the union of all the subsets $1..i$. 
With each iteration our graph grows, until approaching the full graph, allowing us to isolate portions of the graph and converges on an overall solution. In iteration $i$, we consider \emph{new} candidate edges $E'_{c}$, incident to both a node from iteration $i$ as well as another node which can be from iteration $i$ or any other iteration in $V'$. Since some nodes are processed earlier, they may have reached $k$ edges in an earlier iteration and can not add additional candidate edges, limiting the processing which occurs in later iterations.
Since each candidate edge that is considered
is only processed once,
the total sum of edges processed is less than the number of candidate edges, $|E_c| \geq \sum |E'_{c}|$. It follows that each iteration the candidate set $E'_{c}$ has less than or equal to $|E_c| \over b$ candidate edges on average. The nodes receiving an improvement in expected content $V'^* \subset V'$ is localized to the ones reachable by candidate edges from the nodes in the new iteration and thus $V'^*$ is also much smaller than both $V'$ and $V$. Hence, at each iteration we process a smaller \emph{subgraph} limited by $V'$ and the set of edges available in the current iteration $E'_c$, equal in size to $\bar{n} = \Omega(|V'^*| + |E'_{c}|)$, which is then repeated over a total of $b$ iterations.

When processing a subgraph, a decision is made to add or not add an edge $e'_j \in E'_c$ to $E_p$. If this decision were made on a full graph, the choice to add or not add an edge would be dependent on 1) the impact of the edge on all nodes in $V$, as well as 2) all other candidate edges in $E_c$ it could choose instead. In order for the heuristic to be effective, $V'$ must somewhat mimic the structure of $V$, and the method of producing subgraphs must provide a variety of options, with the options believed to be more impactful being considered first.
\stitle{Forest Fire Sampling} In order to generate the subsets described in the previous section, we need a way to find portions of the graph which are representative of the graph as a whole, and provide choices with a higher impact earlier. 
To do so, starting from the content nodes $V_s$, we perform sampling using a traversal based method. 
There are multiple intuitions as to why this is a good choice. First, the content received by each node $v_i \in V \backslash V_s$ is based on its path from a content node $v_j \in V_s$. By having the graph that always includes the nodes in $V_s$, we can find new paths $v_j \rightarrow v_i$ which have a shorter distance than the initial distance to a node $v_i$. Second, the impact of an edge $e_{u,v}$ is based on the nodes reachable by paths from $v$, where the paths continue to improve content, so analyzing connected components allowed us to make more meaningful decisions. A traversal based approach supplies such connected components. 

We use the Forest Fire algorithm~\cite{Leskovec:GraphsOT}, the de-facto graph sampling method that maintains the underlying properties of the graph.
This method of sampling randomly generates a subset of a node's neighbors in the graph. The subgraphs being sampled would include both a large number of neighbors for each node, as well as a significant depth, meaning a new candidate edge $e_{u,v} \in E'_c$ which ended at a node in the current subgraph $v \in V'$ would have information about many of the \emph{additional} nodes whose content spread was updated by including edge $e_{u,v}$. Since the branching factor is large in a given sample, a node $u$ might have several candidate edges, $e_{u,v_1}, e_{u,v_2},...$ such that it could make opinionated choices about which edges to select. The forest fire method will produce more prominent nodes closer to the source earlier than less prominent nodes further away. This means that we will be able to make the decisions which are more likely to have a higher impact on content spread earlier on, while making less meaningful decisions later.

\begin{algorithm}[!tb]
\small
    \caption{ Scalable Algorithm (\scalable)}\label{algo:ff-lp}
    \hspace*{\algorithmicindent} \textbf{Input} : $G(V, E), V_s, E_c, k, \Gee, $ Iter; 
    \textbf{Output} : suggested edges $E_p$
    \begin{algorithmic}[1]
        \State excluded$\leftarrow \{\}$; $E_p \leftarrow \{\}$; 
         $\forall_{v_i \in V}~:~ k'_i \leftarrow k$
        \While{$(G^{'}(V^{'}, E^{'}) \leftarrow $ForestFireSample$(G(V, E), V_s, $ Iter$)) \not = G$}
            \State add $E_p$ to $G^{'}(V^{'}, E^{'})$
            \State $E_c^{'} \leftarrow E_c$ on $G^{'} - $ excluded
            \State add $E_c^{'}$ to excluced
            \State $\mathcal{S}' \leftarrow$ {\sc ShortestDistances}$(G^{'}(V^{'}, E^{'}), E_c^{'})$ {\tt \small // Algorithm \ref{algo:shortestpaths} in Appendix}
            \State $\langle\mathbf{x}^*,\mathbf{y}^*\rangle \leftarrow$ Solve the LP in Figure~\ref{fig:LPRelax}
            \State add {\sc ROUND} $(\mathbf{y}^*, G^{'}(V^{'}, E^{'}), \textbf{k}', V_s, \Gee)$ to $E_p^{'}$ {\tt \small // Algorithm~\ref{algo:round}}
            \State {\bf for} $v_i \in V$ {\bf do} $k_i' \leftarrow \max (k - |\{e_{u,v} \in E_p | u = i\}|, 0)$
        \EndWhile
        \State \bf{return} $E_p$
        
    \end{algorithmic}
\end{algorithm}

\subsection{\scalable Algorithm}
We now propose algorithm \scalable (Algorithm~\ref{algo:ff-lp}) for solving the \problem at scale. The algorithm uses a Forest Fire sampling to achieve nodes close to the content nodes, and then finds the subproblem on that induced subgraph consisting of the candidate edges on that graph $E'_c$ and the nodes improved by those edges $V'^*$ as described above. An instance of $\LP$-Advanced is solved for the first edges to be added to $E_p$ based on the current understanding of the graph $G'(V', E')$. 
For every node $v_i\in V$, let $k_i'$ be the number of candidate edges remained to be suggested to it. Initially, $k_i'=k$ since no edge has been added to $E_p$. 
After adding the new candidate edges to $E_p$, values of $k_i'$ in vector $\mathbf{k}'$ also get updated accordingly.
\new{The new shortest distances to all nodes from source nodes and the corresponding content spread disparity is updated dynamically without recalculating the entire graph.}
The graph is expanded finding nodes that are either less connected or further from the content source nodes $V_s$, meaning the redefined candidate set $E'_c$ contains edges that are generally less important than the previous iteration, but more important than later iterations. The algorithm then expands upon $E_p$ with a subset of the edges in $E'_c$, using the increased information of the new nodes in $V'$. This process repeats over $b$ iterations, adding $N \over b$ nodes and an average $|E_c| \over b$ each iteration. 
Keeping the size of each $\LP$ instances relatively small in each iteration, our algorithm scales to large setting with low disparity and high content spread.

\section{Experiments}\label{sec:exp}

The experiments were conducted using a single work station with a Core i9 Intel X-series 3.5 GHz processor and 128 GB of DDR4 memory, running Linux Ubuntu. The algorithms were implemented in Python 3.
For evaluation purposes, we generated samples from two data sets: a collection of Antelope Valley synthetic networks to model obesity prevention intervention, as well as a real data set from the Pokec social network.

\stitle{Antelope Valley Synthetic Networks Data Set}
We use the Antelope Valley Synthetic Networks from Farnadi et al.'s work on Fair Influence Maximization \cite{FarnadiFIM}. We used twenty of the synthetic networks. Each of these networks initially contained 500 nodes and between 1576-1697 edges. Each node in the network has a sensitive attribute of {gender}  with two classes. 



\stitle{Pokec Social Network Data Set}
We used the Pokec Social Network data set from the SNAP's network data sets~\cite{snapnets}. 
The Pokec data set represents real data on the Pokec social network, and it contains 1.63 million nodes and 30.6 million edges. Each node in the network contains a sensitive attribute {\small\tt gender}. 



\stitle{Evaluation Plan and Performance Measures} We evaluate our algorithms LP-Approx \new{(Algorithm~\ref{algo:LP})}, LP-Advanced \new{(Algorithm~\ref{algo:lp-iter})}, and \scalable \new{(Algorithm~\ref{algo:ff-lp}, using \cite{littleballoffur})}. The metrics we focus on are (1) the run time of the algorithm, (2) the disparity between groups, and (3) the change in content spread. 
We report the runtime in seconds, we report disparity ($\eps$ in Equation~\ref{eq:fairness}) as the \% difference between the two groups' size weighted content spread, or $\eps \times 100$.


\new{
Computing the content spread using IC is known to be \sharpP-hard~\cite{chaoji-cs,Chen-Scalable} and, hence MIP~\cite{Chen-Scalable}, and later RMPP~\cite{chaoji-cs}, have been proposed alternative measures to approximate it. Even though the theoretical analyses in this paper have been carried out using RMPP, our algorithms work (almost) equally well for other cascade models. Confirming that we observed consistent results for RMPP, MIP, and IC, we report our evaluations using MIP to demonstrate practical performance of our work beyond RMPP.
}
\new{
The change in Content Spread is reported as the percent increase in MIP Content Spread:}
\begin{align}\small
\sum_{v_i \in V}{F_{MIP}(v_i, E_p) - F_{MIP}(v_i, \{\}) \over F_{MIP}(v_i, \{\})}\times 100
\end{align}


\stitle{\new{Baselines}}
\new{
\begin{itemize}[leftmargin=*]
    \item {\em Brute Force}. Checks all valid combinations of edges from the candidate set, choosing the combination which maximizes content spread, constrained to k edges per node, and perfect fairness between demographic groups.
    \item {\em Continuous Greedy} \cite{chaoji-cs}. Each candidate edge is assigned a weight starting at 0. Over 500 iterations, 10 sampled graphs are generated with the edges included with the probability of their weights. Each edge is then added to the sampled graphs. A value based on it's lift on the sample is assigned to the edge. A matching of the highest valued edges is made, and then those edges have their weights increased. After all iterations the edges with the highest weights are chosen as the solution. Optimizations made based on our dynamic algorithms included in the appendix.
    \item {\em IRFA} \cite{yang2019marginal}. 
    Greedily chooses edges based on their Influence Ranking (IR) or effect on down stream nodes, and the improvement in probability of that influence ranking based on adding the edge. It then performs a fast adjustment (FA) to recalculate probabilities and influence rankings. The input to this algorithm should be a Directed Acyclic Graph. As a preprocessing step, we used \cite{sun2017breaking} to remove cycles from the graph.
    \item {\em SpGreedy} \cite{zhu2021minimizing}. 
    Measures the relationship between demographic groups via the opinion model properties {\em polarizaiton} (groups being radically separated) and {\em disagreement} (members of different groups neighboring each other). The fairness of the network is then attempted to be increased by greedily choosing edges which minimize these properties. The opinions of one demographic group are given as $1$ and for the other are $-1$.
    \item {\em ACR-FoF} \cite{yu-contentspread}. The algebraic connectivity of including an edge and the number of shared friends are combined to score each edge. The highest scored edges are greedily chosen to k per node. This algorithm require an undirected graph, so for this baseline, in our experiments we ignored the edge directions.   
\end{itemize}}

\stitle{\new{Candidate Edge Selection}}
%
\new{\begin{itemize}[leftmargin=*]
\item {\em Friend of Friend (FoF)} - All friends of friends (not already friends) are considered as possible recommendations.
\item {\em Intersecting Group Count (IGC)} \cite{roth-implicit}.
For each node $v_i \in V$, a rating for a potential candidate edge is given to each other node equal to the number of groups that node shares with friends of $v_i$. Each node on the graph's neighborhood is considered as a group. If a node is not in the same group as any friend it does not have a rating. The 1/3 highest rated nodes are included in the candidate set of suggestions for $v_i$. 
\end{itemize}} 

\subsection{Experiment Results}

\stitle{\new{The Cost of Fairness}}
\new{Algorithm~\ref{algo:lp-iter} uses $\LP$-Relaxation and randomized rounding for solving \problem. Alternatively, if the $\LP$ did not contain the fairness constraint (Definition~\ref{def:fairness}), it would generate an approximate solution to \cs (Definition~\ref{def:spread}). Fairness is a desirable property of an algorithm, but it comes with "The Cost of Fairness (CoF)". The CoF of \problem~is demonstrated in Figure~\ref{fig:cof}.
The left-axis and solid lines in the figure show the average lift, while the disparity is shown in the right-axis with dashed lines.
Algorithm~\ref{algo:lp-iter} is run both with and without the fairness constraint, with $k=5$, $p=.5$, $|V|=500$, an initial disparity between 30\% and 35\%, and with 20 trials across each of a varying $|V_s|$. The resulting average lift and average disparity are recorded. The algorithm without the fairness constraint achieves 1.25x on lift. The algorithm with the fairness constraint reduces disparity by over 98\% in all cases while the algorithm without it reduces it by less than 60\%. The "Cost of Fairness" is demonstrated as a loss in overall lift for the result of near zero disparity between demographic groups.}



\stitle{Demonstration of Near Optimality}
\new{Algorithm~\ref{algo:LP}, $\LP$-Approx, is demonstrated to be nearly optimal in both overall content spread and disparity, by comparison with the optimal Brute Force baseline. 7 graphs with $|V|=40$ are sampled from the Antelope Valley dataset, such that Brute Force can achieve 0 disparity through edge selection}
, using
1 content node, with $k = 2$, \new{$p$ = 0.5,} and an initial disparity of ~$95.6-98.9\%$.
The $\LP$-Approx algorithm was run ten times on each of the graphs and the average and standard deviation were reported \new{and compared against the results of Brute Force}. 

\new{Figure~\ref{fig:optimal-disparity} demonstrates the near optimal disparity results.} \new{As shown, for all but two problem instances,} the $\LP$-Approx was able to achieve optimal disparity ($\eps=0$).
\new{Even in the worst case, the average disparity across experiments never exceeded} $0.006$.

\begin{figure*}[!tb] 
    \begin{minipage}[t]{0.23\linewidth}
        \centering
        \includegraphics[width=\textwidth]{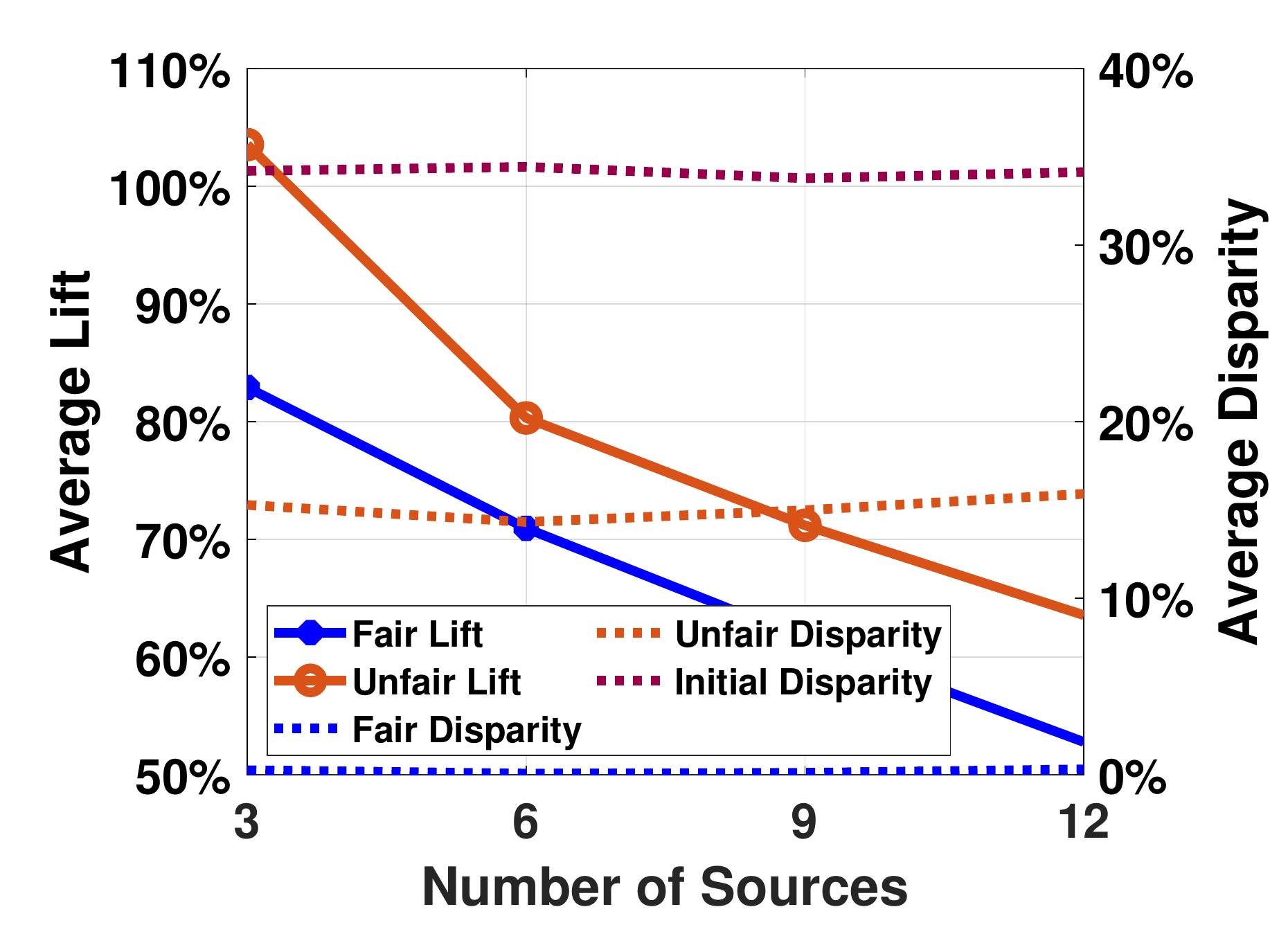}
    \vspace{-8mm} \caption{\small\new{Alg.~\ref{algo:lp-iter}: w/ and w/o the fairness constraint, comparing lift and disparity}}
        \label{fig:cof}
    \end{minipage}
    \hfill
    \begin{minipage}[t]{0.23\linewidth}
        \centering
        \includegraphics[width=\textwidth]{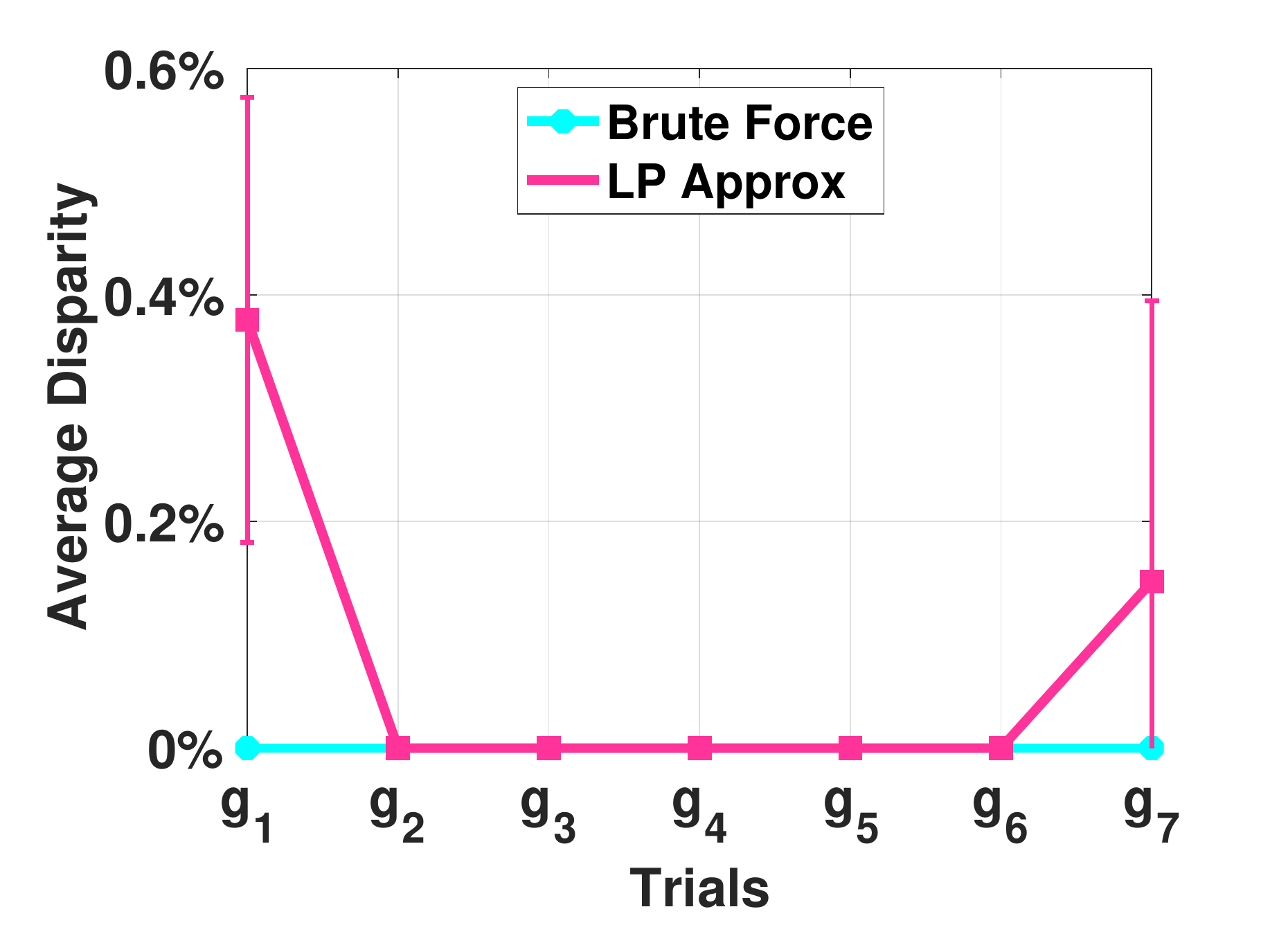}
        \vspace{-8mm}\caption{\small demonstration of near optimal disparity of $\LP$-Approx}
        \label{fig:optimal-disparity}
    \end{minipage}
    \hfill
    \begin{minipage}[t]{0.23\linewidth}
        \centering
        \includegraphics[width=\textwidth]{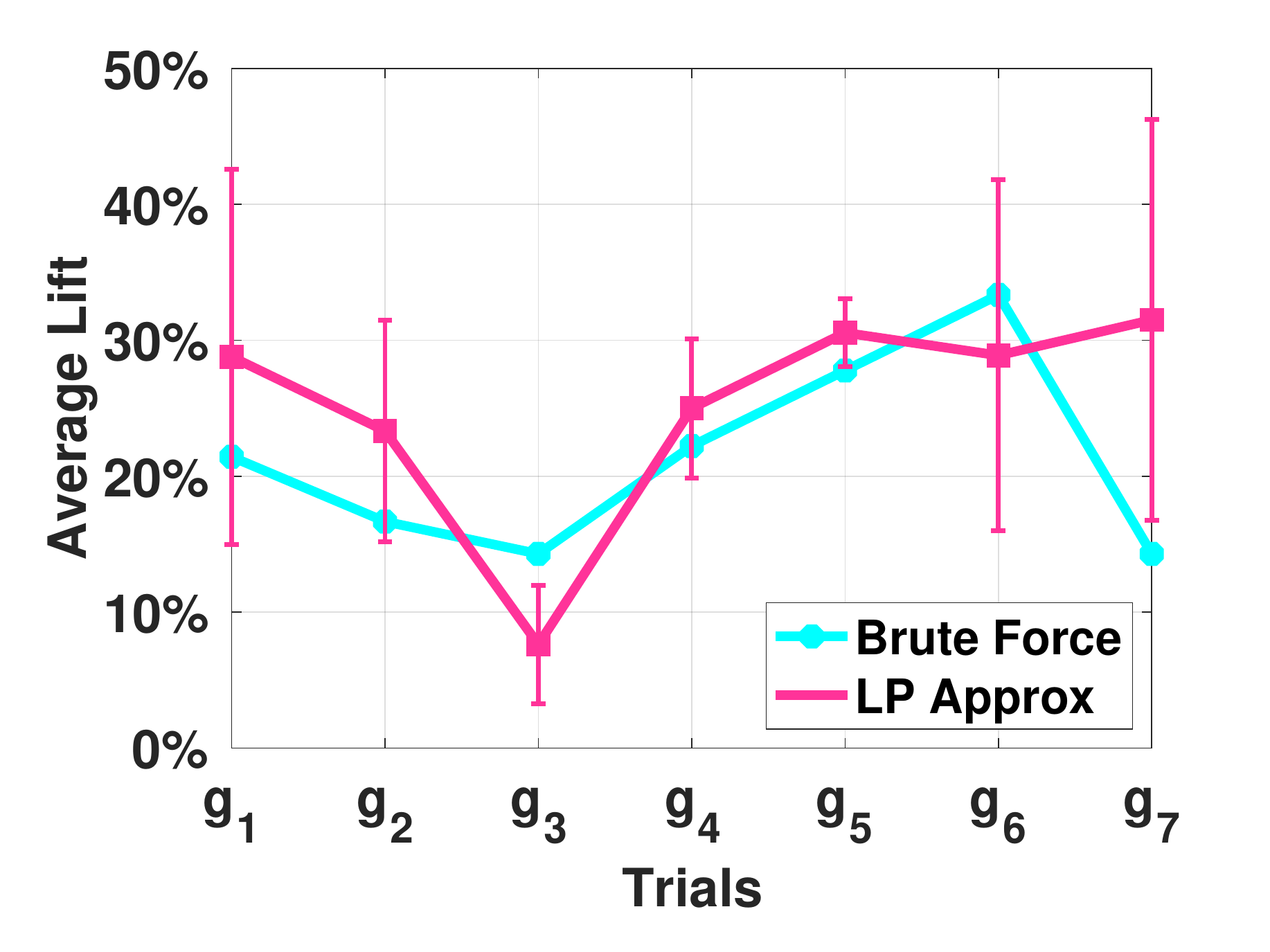}
        \vspace{-8mm}\caption{\small demonstration of near optimal lift of $\LP$-Approx}
        \label{fig:optimal-lift}
    \end{minipage}
    \hfill
    \begin{minipage}[t]{0.23\linewidth}
         \centering
         \includegraphics[width=\textwidth]{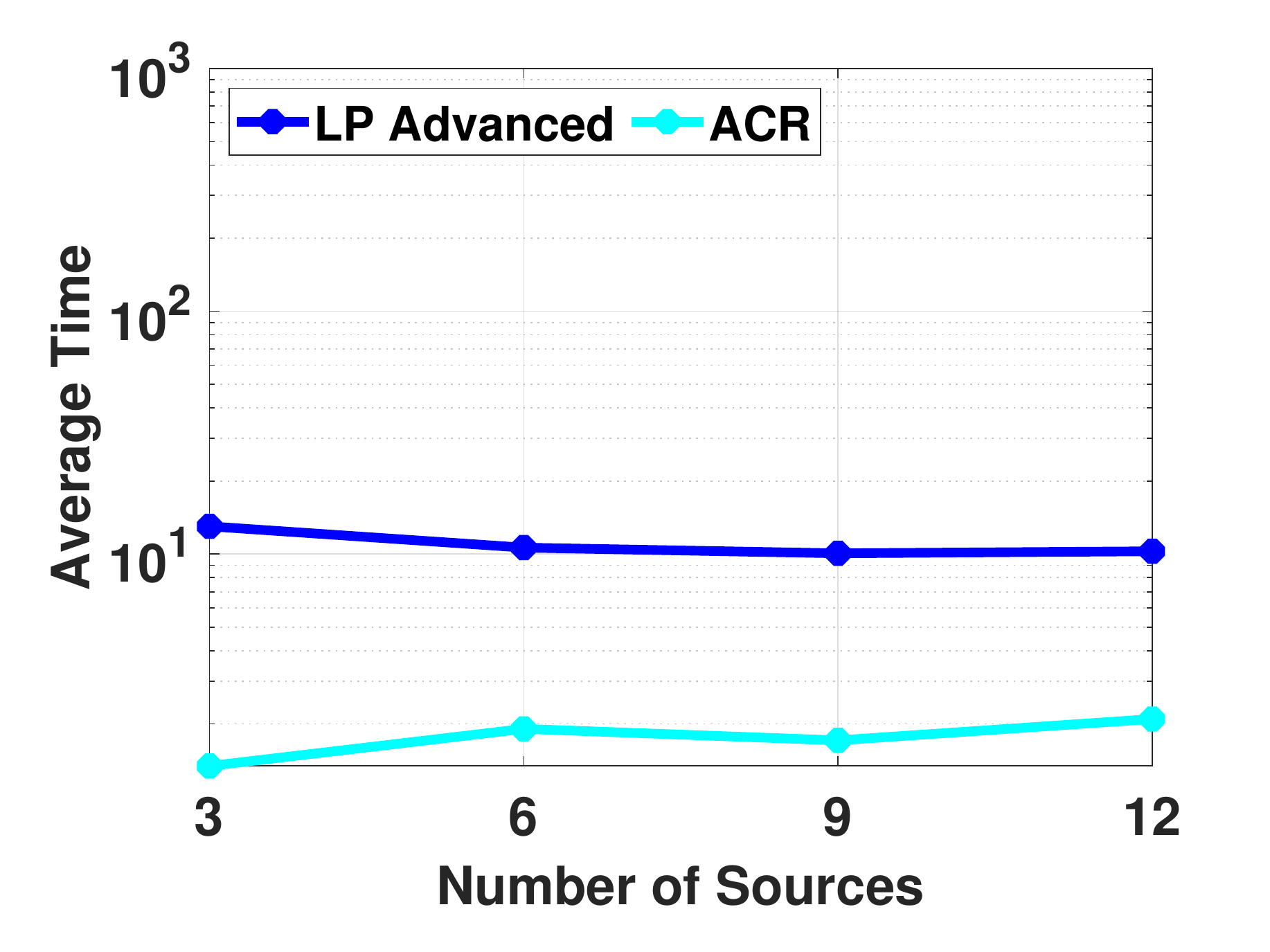}
         \vspace{-8mm}\caption{\small\new{Runtime with varying sources at 500 nodes}}
         \label{fig:500-varying-runtime}
    \end{minipage}
    \vspace{-4mm}
\end{figure*}
\begin{figure*}[!tb] 
    \begin{minipage}[t]{0.23\linewidth}
        \centering
        \includegraphics[width=\textwidth]{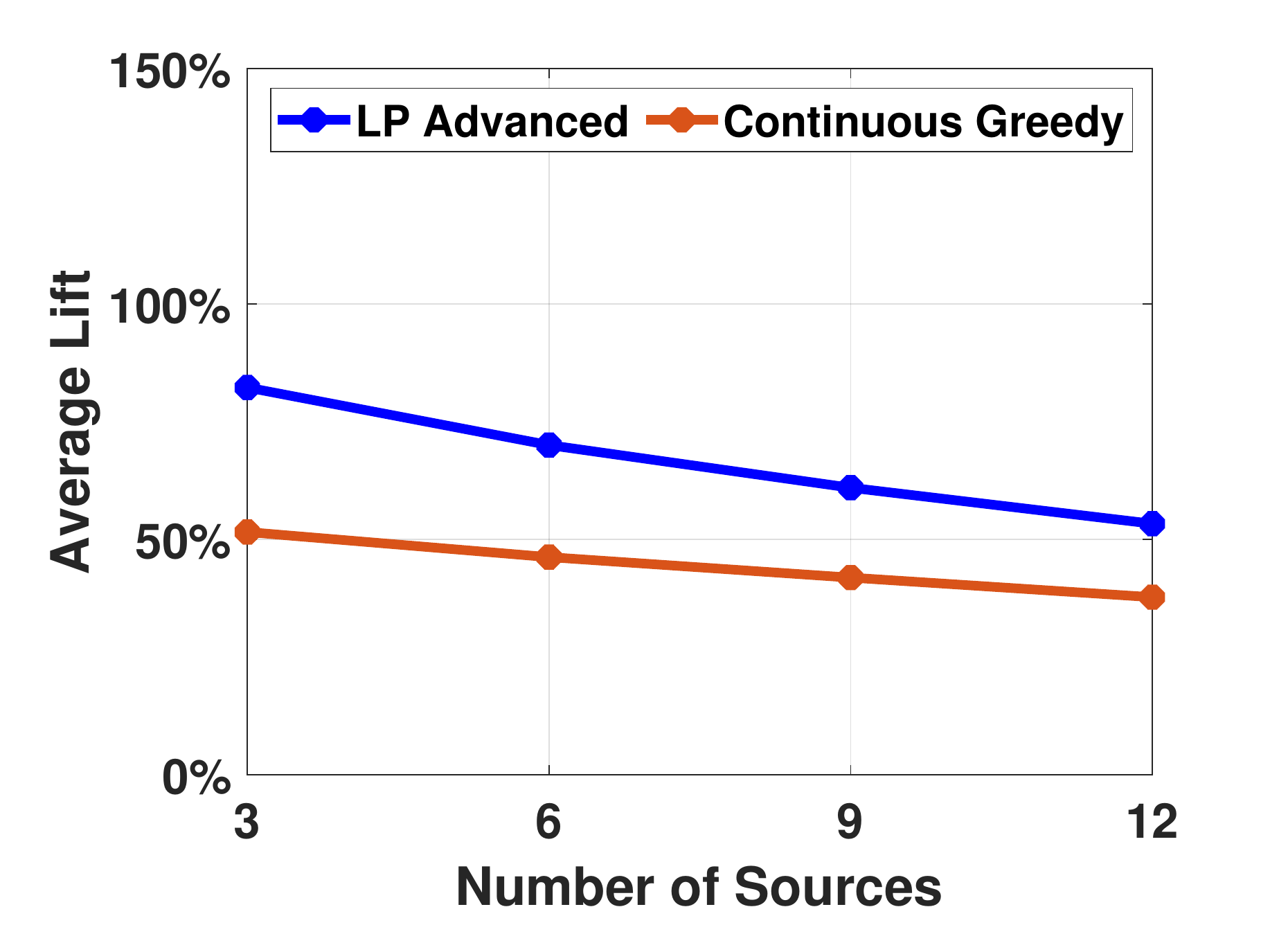}
    \vspace{-8mm} \caption{\small\new{Lift with varying sources at 500 nodes}}
        \label{fig:500-varying-lift}
    \end{minipage}
    \hfill
    \begin{minipage}[t]{0.23\linewidth}
        \centering
        \includegraphics[width=\textwidth]{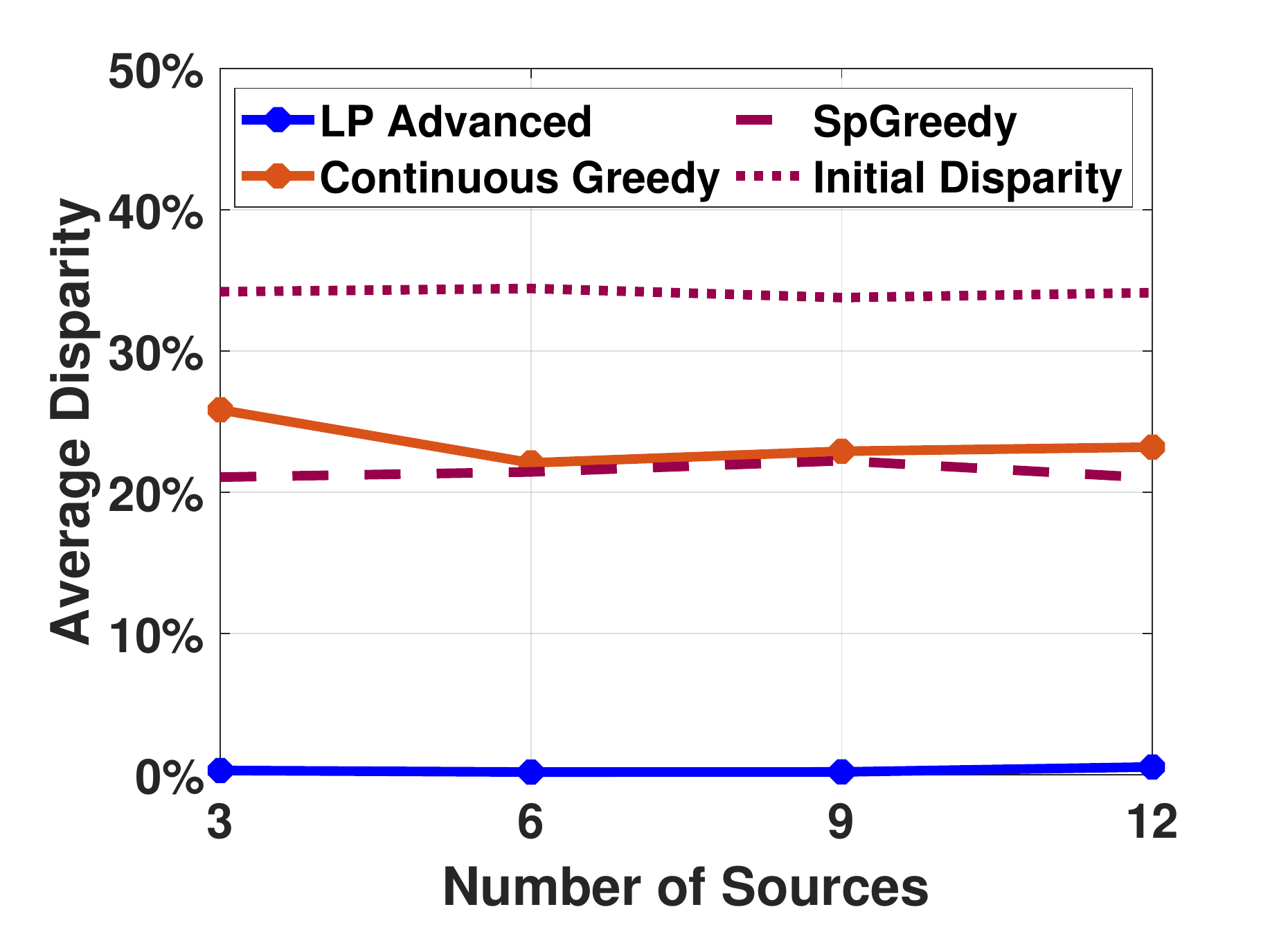}
        \vspace{-8mm}\caption{\small\new{Disparity with varying sources at 500 nodes}}
        \label{fig:500-varying-disparity}
    \end{minipage}
    \hfill
    \begin{minipage}[t]{0.23\linewidth}
        \centering
        \includegraphics[width=\textwidth]{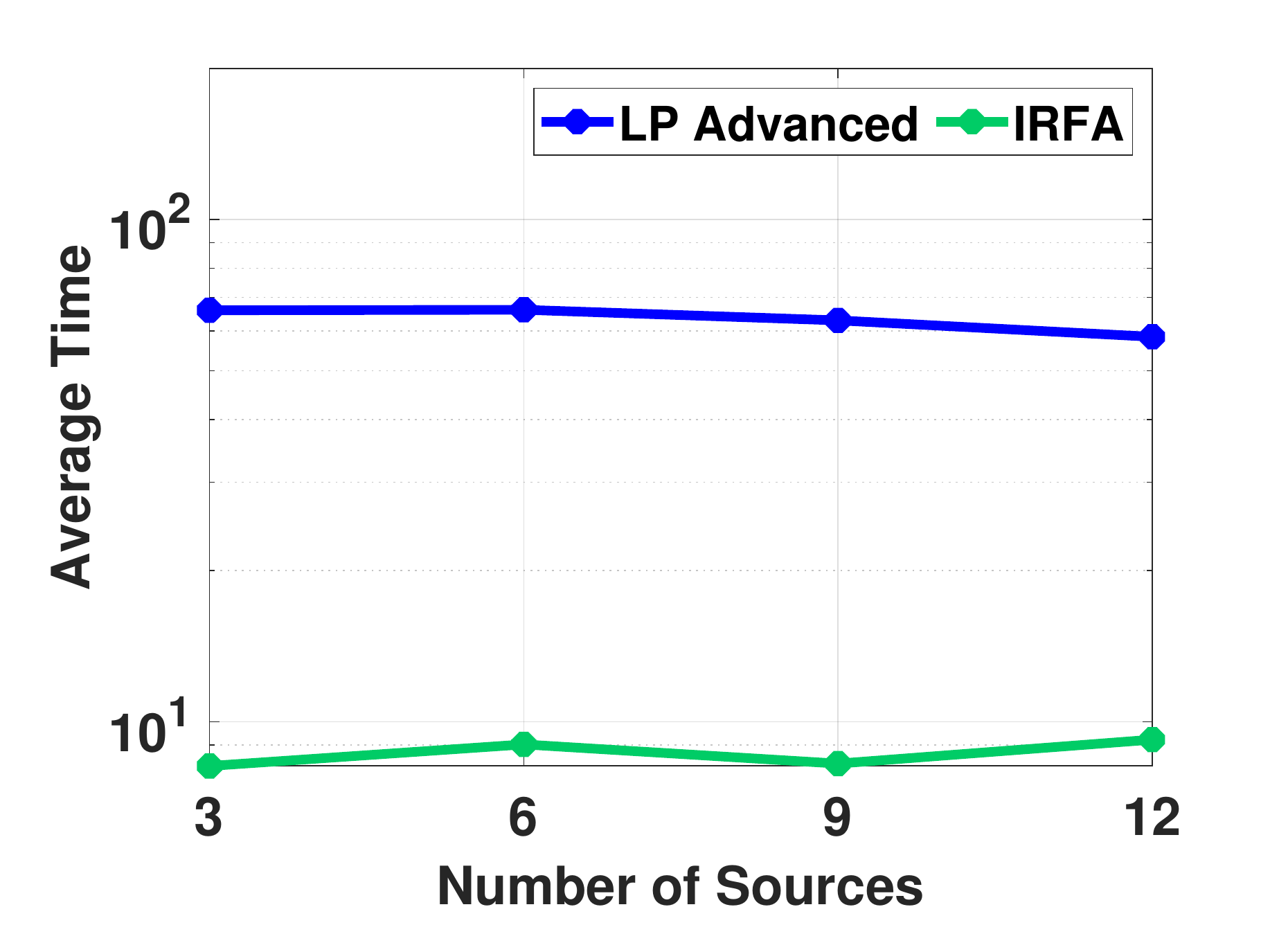}
        \vspace{-8mm}\caption{\small\new{Runtime with varying sources at 1000 nodes}}
        \label{fig:1000-varying-time}
    \end{minipage}
    \hfill
    \begin{minipage}[t]{0.23\linewidth}
         \centering
         \includegraphics[width=\textwidth]{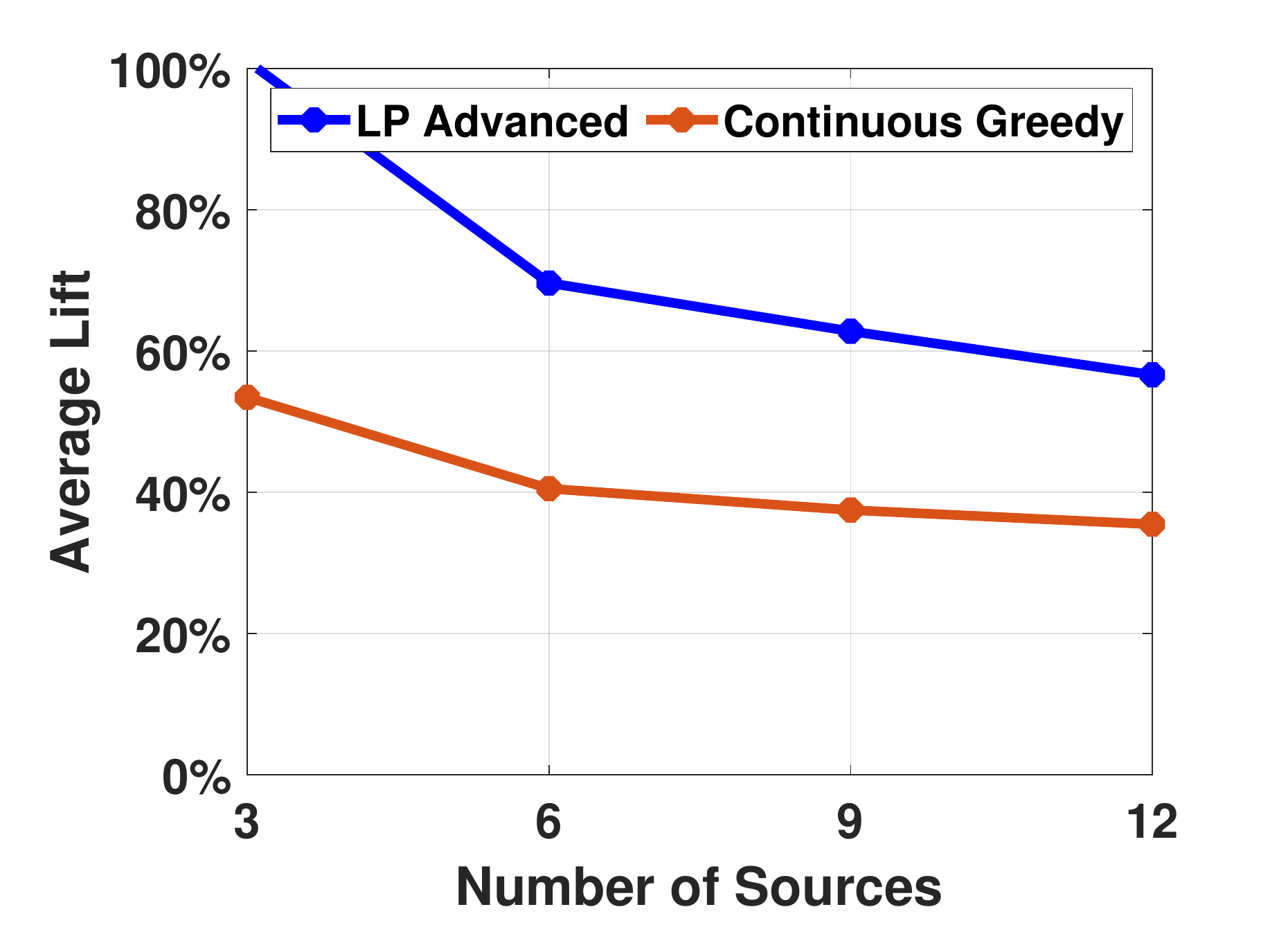}
         \vspace{-8mm} \caption{\small\new{Lift with varying sources at 1000 nodes}}
         \label{fig:1000-varying-lift}
    \end{minipage}
    \vspace{-4mm}
\end{figure*}
\begin{figure*}[!tb] 
    \begin{minipage}[t]{0.23\linewidth}
         \centering
         \includegraphics[width=\textwidth]{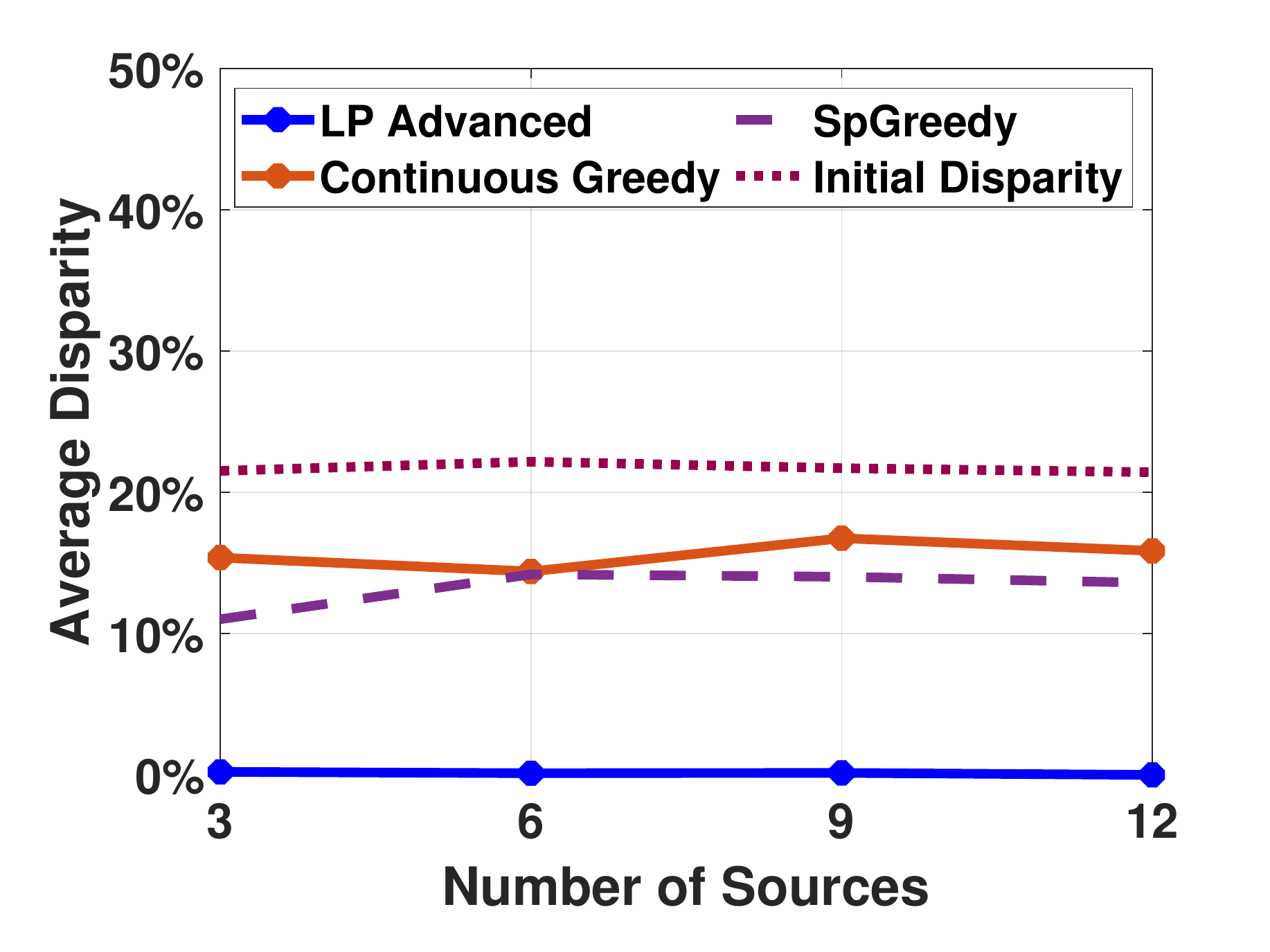}
         \vspace{-8mm} \caption{\small\new{Disparity with varying sources at 1000 nodes}}
         \label{fig:1000-varying-dispairty}
    \end{minipage}
    \hfill
    \begin{minipage}[t]{0.23\linewidth}
        \centering
        \includegraphics[width=\textwidth]{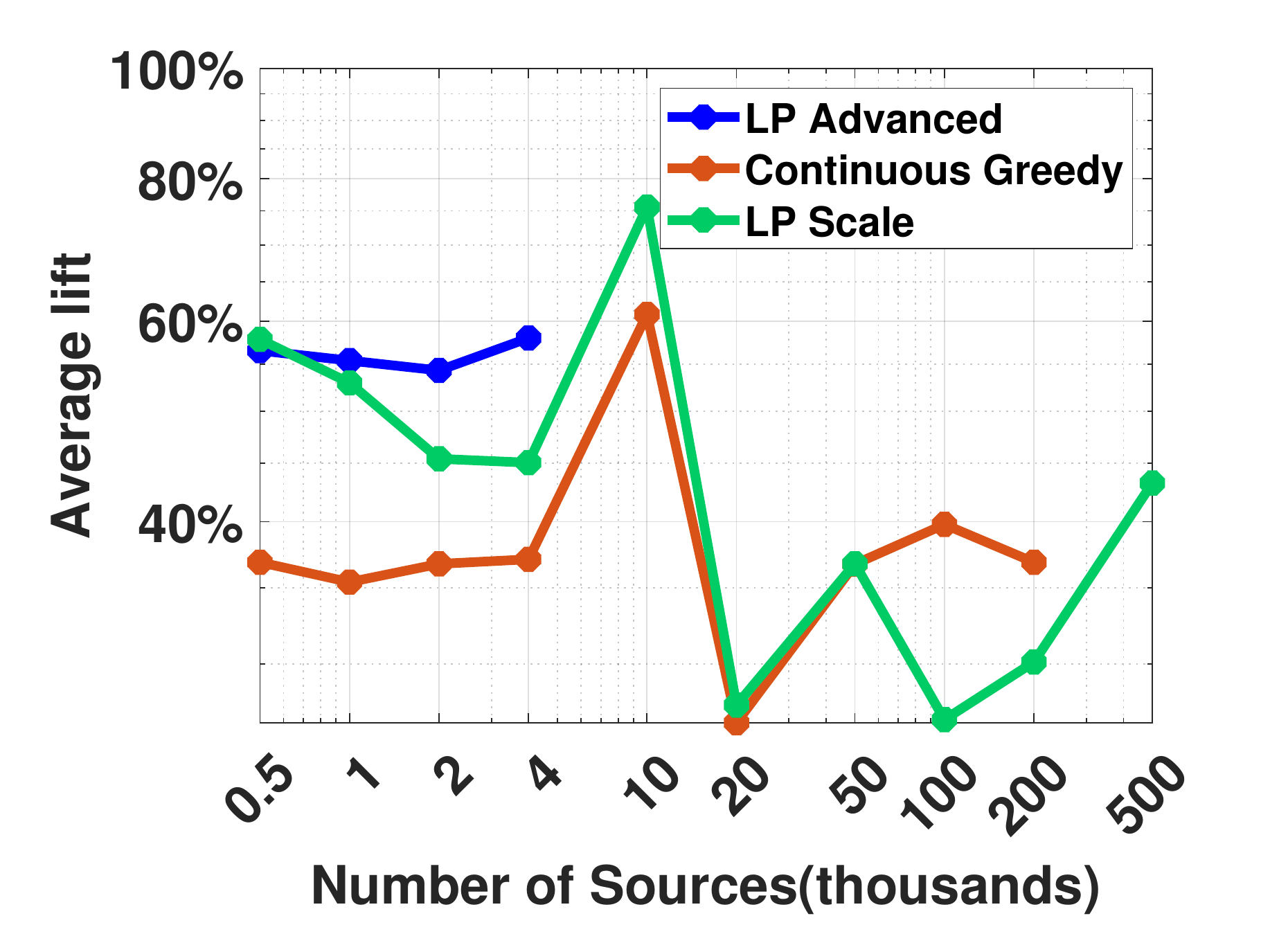}
        \vspace{-8mm}\caption{\small\new{Lift at scale}}
        \label{fig:scaling-lift}
    \end{minipage}
    \hfill
    \begin{minipage}[t]{0.23\linewidth}
        \centering
        \includegraphics[width=\textwidth]{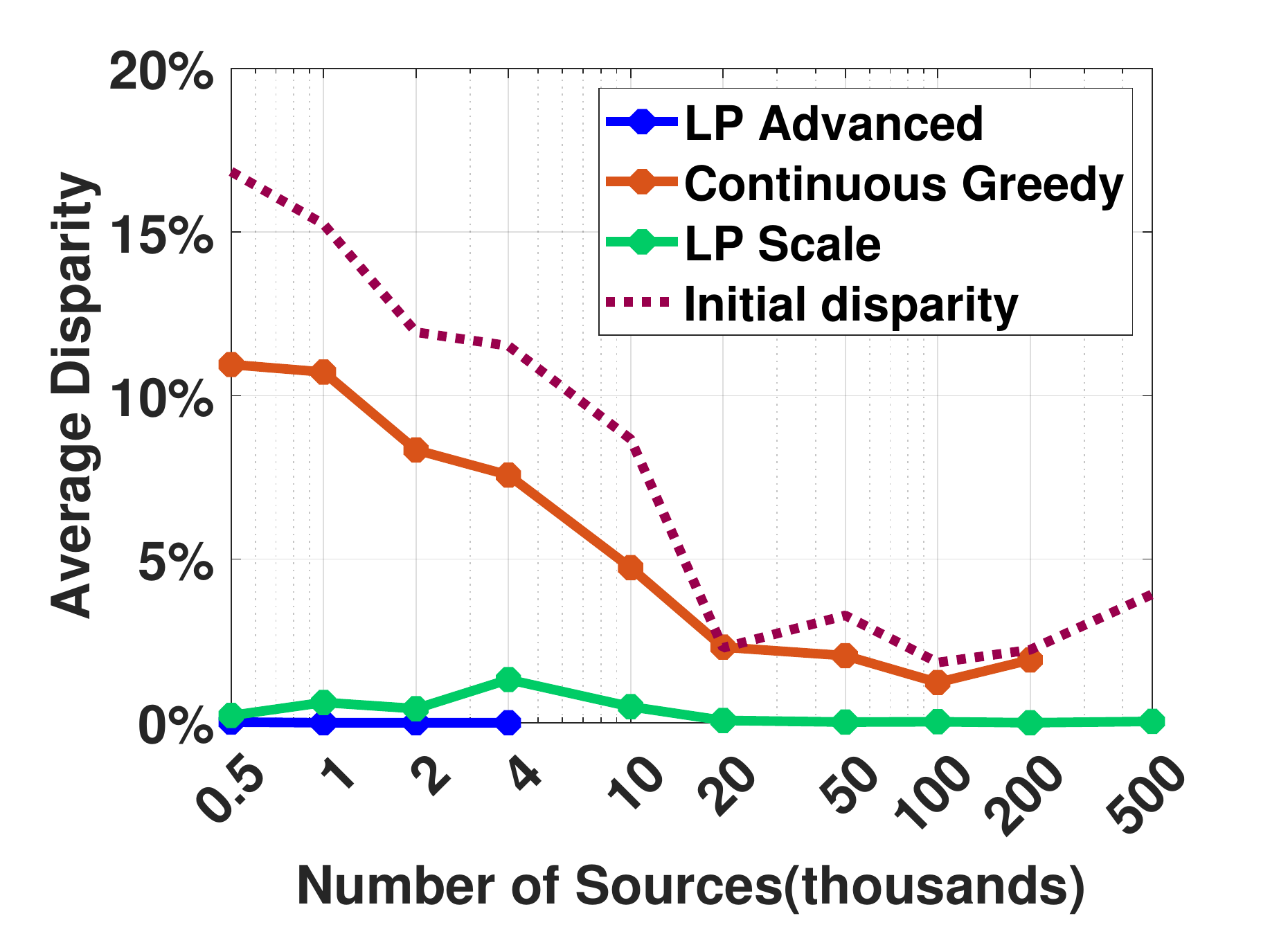}
        \vspace{-8mm}\caption{\small\new{Disparity at scale}}
        \label{fig:scaling-disparity}
    \end{minipage}
    \hfill
    \begin{minipage}[t]{0.23\linewidth}
        \centering
        \includegraphics[width=\textwidth]{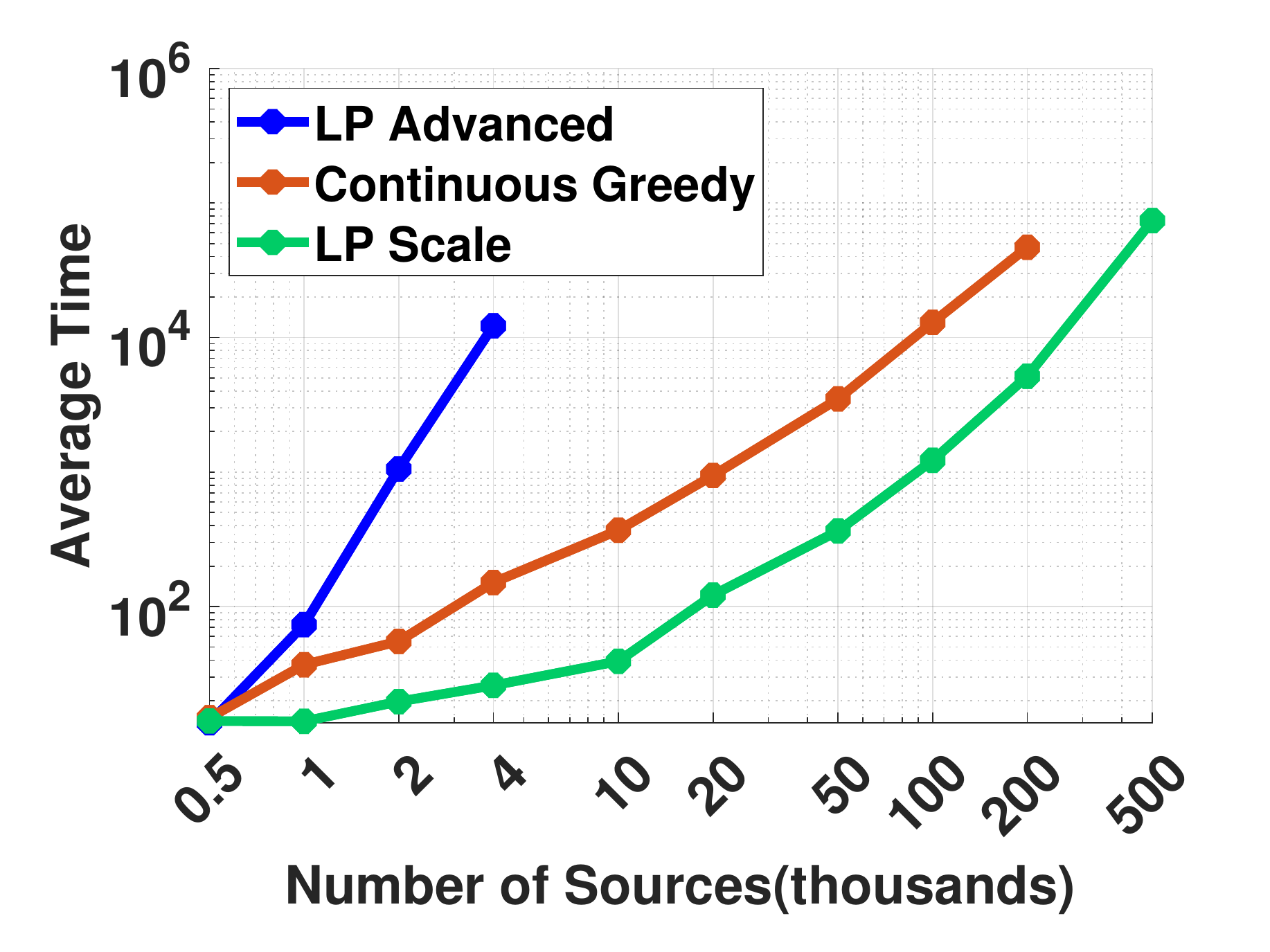}
        \vspace{-9mm}\caption{\small\new{Runtime at scale}}
        \label{fig:scaling-runtime}
    \end{minipage}
\end{figure*}



\new{Figure~\ref{fig:optimal-lift} demonstrates the near optimal lift results.
For the previously mentioned disparity results, $\LP$-Approx achieved optimal lift} in every case except two. It should be noted that in some cases $\LP$-Approx \new{strictly dominates} the optimal solution in both disparity and lift. This is because occasionally in rounding $> k$ edges are selected, \new{as only} the expected number of edges for each node is $k$.



\stitle{Comparison with \new{Various} Baseline \new{Methods}}
\new{Algorithm~\ref{algo:lp-iter} ($\LP$-Advanced), is demonstrated to achieve with both methods of candidate edge selection a higher content spread and lower disparity than all baselines. These results are demonstrated on both the $500$ node Antelope Valley datasets, and on $1000$ node samples of the Pokec dataset, with sources selected to be between $30-35\%$ and $20-25\%$ respectively. Experiments were done varying $k$, $p$, and $|V_s|$. For each setting, 20 trials were run, with $k=3$, $p=0.5$ and $|V_s|=3$ when not being varied (varying $k$ and varying $p$ are included in the appendix.)}

\new{In Table \ref{table:varying-sources-500}, baseline methods and $\LP$-Advanced are compared on 500 nodes. The bold elements contain the best values for each property on each setting. $\LP$-Advanced had the highest lift and lowest disparity compared against all baselines in all instances. It however was not the fastest, where ACR-FoF had the lowest runtime (Figure ~\ref{fig:500-varying-runtime} demonstrates the comparison between $\LP$-Advanced and ACR-FoF w.r.t. runtime). Even though we are 10x slower than ACR-FoF, we had 3-4x greater content spread lift, and while they reduced disparity by at most 33\% we reduced it by 99\%!
The baseline with the next highest lift was Continuous Greedy (Figure ~\ref{fig:500-varying-lift} demonstrates the comparison between $\LP$-Advanced and Continuous Greedy w.r.t. lift), and the next lowest disparity was closely between Continuous Greedy and SpGreedy (Figure ~\ref{fig:500-varying-disparity} demonstrates comparison between $\LP$-Advanced and both these methods w.r.t. disparity). Additionally, at 500 nodes, $\LP$-Advanced had a shorter runtime than both Continuous Greedy and SpGreedy.}

\new{For 1000 node graphs, $\LP$-Advanced is compared against baselines in Table \ref{table:varying-sources-1000}. Again, the best results are shown in bold. $\LP$-Advanced still achieved the highest lift and lowest disparity, and in these instances IRFA had the lowest runtime (Figure ~\ref{fig:1000-varying-time} demonstrates the comparison between $\LP$-Advanced and IRFA w.r.t. runtime)
While we are 10x slower than IRFA we produced 3x greater lift in content spread, and they achieved only a 25\% improvement in disparity compared to our 95\% improvement.
Continuous Greedy produced the next highest lift. (Figure ~\ref{fig:1000-varying-lift} demonstrates the comparison between $\LP$-Advanced and Continuous Greedy w.r.t. lift), and the next lowest disparity was between SpGreedy and Continous Greedy (Figure ~\ref{fig:1000-varying-dispairty} demonstrates comparison between $\LP$-Advanced and both these methods w.r.t. disparity). Additionally, at 1000 nodes, $\LP$-Advanced runs faster than only SpGreedy.}

\stitle{\new{Scalability}} \new{For these experiments, Pokec was sampled at 500, 1000, 2000, and 4000, 10k, 20k, 50k, 200k, and 500k nodes, with $k = 3$, and $p = 0.5$ for each setting, and $|V_s| = 5$ for sizes up to and including 10k nodes, while $|V_s| = 10$ for 20k nodes and greater. For up to and including 4k nodes, results are averaged over 5 trials. Due to runtime constraints, for 10k nodes and higher, results are from one trial. FoF candidate edges were used. The baseline methods were compared with both $\LP$-Advanced as well as \scalable. If an algorithm took greater than 24 hours, results were not reported. Additionally, IRFA is not reported over 20k due to constraints with generating the necessary DAG from the initial graph.}

\new{The results are shown in Table ~\ref{tab:scaling} (best results in bold). This table shows that while $\LP$-Advanced is able to be run, it produces both the highest lift and lowest disparity (except once our scalable algorithm produced higher lift). $\LP$-Advanced produced 1.5x as much content spread as the next closest baseline, Continuous Greedy. SpGreedy achieved the next best reduction in disparity of 50\%, and $\LP$-Advanced {\em reduced the disparity by at least 99.9\%}.
In all other cases, the next best algorithm is \scalable, producing only marginally less lift and slightly higher disparity. For 500 and 1000 nodes, IRFA had a lower runtime than \scalable. However, \scalable had 2x the lift of IRFA and while IRFA reduced disparity by 25\%, \scalable reduced disparity by 96\%.
For 2000 nodes and higher, \scalable always resulted in the shortest runtime. Additionally, when $\LP$-Advanced cannot be run, \scalable always produces the lowest disparity among all methods. \scalable produced the highest lift at 10k, 20k, and 50k nodes, followed by Continuous Greedy. Then for two settings, 100k nodes and 200k nodes, Continuous Greedy produced a higher lift than \scalable.
On the same settings Continuous Greedy only reduced the disparity by at most 25\%, whereas \scalable reduced the disparity by at least 98\%. Since all other methods failed to run within runtime constraints at 500k nodes, \scalable is the defacto best option. $\LP$-Advanced, \scalable, and Continuous Greedy are compared w.r.t. lift, disparity, runtime and in Figures~\ref{fig:scaling-lift}, \ref{fig:scaling-disparity}, and \ref{fig:scaling-runtime}~respectively.} 

\begin{table*}[]
\caption{\small\new{Comparing Lift (\%), Disparity (\%), and Runtime (sec), on 500 node Antelope Valley data sets, of LP-Advanced and all baselines, varying $|V_s|$, using 20 trials per setting. Both FoF and IGC were considered for candidate edge selection. $k = 3$, $p=0.5$}}
\resizebox{\textwidth}{!}{ %
\begin{tabular}{ | @{}c@{} | @{}c @{}| @{}c @{}| @{}c@{}c@{}c@{} | @{}c@{}c@{}c@{} | @{}c@{}c@{}c@{} | @{}c@{}c@{}c@{} | @{}c@{}c@{}c@{} |}\hline
\multirow{2}{*}
{\begin{tabular}[c]{ @ {}c @ {}}Candidate\\ Edges\end{tabular}} &
\multirow{2}{*}{$|V_s|$}
&  \multirow{2}{*}{\begin{tabular}[c]{ @ {} c @ {}}Initial\\ Disparity\end{tabular}}
&  \multicolumn{3}{@{}c@{} |}{LP - Advanced (ours)}
& \multicolumn{3}{@{}c@{} |}{Continuous Greedy}
&  \multicolumn{3}{@{}c@{} |}{IRFA}
& \multicolumn{3}{@{}c @{}|}{SpGreedy}
& \multicolumn{3}{@{}c@{} |}{ACR - FoF}
\\ \cline{4 - 18}
& & & \multicolumn{1}{@{}c@{} |}{Lift} & \multicolumn{1}{@{}c@{} |}{Disparity  } & Time
& \multicolumn{1}{@{}c@{} |}{Lift} & \multicolumn{1}{@{}c@{} |}{Disparity} & Time
& \multicolumn{1}{@{}c@{} |}{Lift} & \multicolumn{1}{@{}c@{} |}{Disparity} & Time
& \multicolumn{1}{@{}c@{} |}{Lift} & \multicolumn{1}{@{}c@{} |}{Disparity} & Time
& \multicolumn{1}{@{}c@{} |}{Lift} & \multicolumn{1}{@{}c@{}|}{Disparity} & Time
\\ \hline
\multirow{4}{*}{FoF}
& 3 & 34.2$\pm$0.3 & \multicolumn{1}{@{}c@{} |}{\textbf{82.3$\pm$14.4}} & \multicolumn{1}{@{}c@{} |}{\textbf{0.3$\pm$0.5}} & 13.0$\pm$13.9  & \multicolumn{1}{@{}c@{} |}{51.6$\pm$8.2} & \multicolumn{1}{@{}c@{} |}{25.8$\pm$5.6} & 33.0$\pm$8.9  & \multicolumn{1}{@{}c@{} |}
{25$\pm$9.1} & \multicolumn{1}{@{}c@{} |}{29.2$\pm$8.6} & 4.3$\pm$1.6  & \multicolumn{1}{@{}c@{} |}{20.1$\pm$9.0} & \multicolumn{1}{@{}c@{} |}{21.1$\pm$6.1} & 41.1$\pm$10.0  & \multicolumn{1}{@{}c@{} |}{27.1$\pm$9.0} & \multicolumn{1}{@{}c@{} |}{25.0$
\pm$7.1} & \textbf{1.3$\pm$0.7}
\\ \cline{2 - 18}
& 6 & 34.4$\pm$0.5 & \multicolumn{1}{@{}c@{} |}{\textbf{70.0$\pm$11.3}} & \multicolumn{1}{@{}c@{} |}{\textbf{0.2$\pm$0.8}} & 10.6$\pm$2.6  & \multicolumn{1}{@{}c@{} |}{46.2$\pm$7.3} & \multicolumn{1}{@{}c@{} |}{22.1$\pm$6.9} & 26.6$\pm$4.0  & \multicolumn{1}{@{}c@{} |
}{22.7$\pm$5.8} & \multicolumn{1}{@{}c@{} |}{26.1$\pm$6.9} & 3.6$\pm$1.6  & \multicolumn{1}{@{}c@{} |}{17.8$\pm$5.4} & \multicolumn{1}{@{}c@{} |}{21.5$\pm$5.0} & 43.3$\pm$6.9  & \multicolumn{1}{@{}c@{} |}{24.4$\pm$4.3} & \multicolumn{1}{@{}c@{} |}{25.0$\pm$5.4} & \textbf{1.9$\pm$0.8}
\\ \cline{2 - 18}
& 9 & 33.8$\pm$0.6 & \multicolumn{1}{@{}c@{} |}{\textbf{61.0$\pm$6.2}} & \multicolumn{1}{@{}c@{} |}{\textbf{0.2$\pm$0.4}} & 10.1$\pm$3.5  & \multicolumn{1}{@{}c@{} |}{41.9$\pm$5.1} & \multicolumn{1}{@{}c@{} |}{22.9$\pm$6.2} & 25.5$\pm$6.9  & \multicolumn{1}{@{}c@{} |}{
21.0$\pm$5.3} & \multicolumn{1}{@{}c@{} |}{27.5$\pm$5.0} & 4.0$\pm$1.4  & \multicolumn{1}{@{}c@{} |}{16.5$\pm$3.0} & \multicolumn{1}{@{}c@{} |}{22.2$\pm$3.7} & 41.4$\pm$7.7  & \multicolumn{1}{@{}c@{} |}{21.2$\pm$4.4} & \multicolumn{1}{@{}c@{} |}{27.4
$\pm$4.7} & \textbf{1.7$\pm$0.8}
\\ \cline{2 - 18}
& 12 & 34.1$\pm$0.5 & \multicolumn{1}{@{}c@{} |}{\textbf{53.4$\pm$7.9}} & \multicolumn{1}{@{}c@{} |}{\textbf{0.6$\pm$1.7}} & 10.3$\pm$3.4  & \multicolumn{1}{@{}c@{} |}{37.8$\pm$5.7} & \multicolumn{1}{@{}c@{} |}{23.2$\pm$5.6} & 23.3$\pm$6.5  & \multicolumn{1}{@{}c@{} |}
{20.1$\pm$5.0} & \multicolumn{1}{@{}c@{} |}{26.8$\pm$3.9} & 3.9$\pm$1.8  & \multicolumn{1}{@{}c@{} |}{16.7$\pm$3.0} & \multicolumn{1}{@{}c@{} |}{21.0$\pm$4.4} & 45.4$\pm$9.5  & \multicolumn{1}{@{}c@{} |}{20.2$\pm$4.3} & \multicolumn{1}{@{}c@{} |}{25.0$\pm$4.5} & \textbf{2.1$\pm$1.1}
\\ \hline
\multirow{4}{*}{IGC}
& 3 & 34.2$\pm$0.3 & \multicolumn{1}{@{}c@{} |}{\textbf{94.0$\pm$20.5}} & \multicolumn{1}{@{}c@{} |}{\textbf{0.6$\pm$1.5}} & 18.7$\pm$4.8  & \multicolumn{1}{@{}c@{} |}{61.9$\pm$13.5} & \multicolumn{1}{@{}c@{} |}{21.0$\pm$10.1} & 35.1$\pm$10.1  & \multicolumn{1}{@{}c@{}
 |}{28.0$\pm$9.9} & \multicolumn{1}{@{}c@{} |}{27.5$\pm$9.2} & 3.0$\pm$1.2  & \multicolumn{1}{@{}c@{} |}{23.6$\pm$11.6} & \multicolumn{1}{@{}c@{} |}{20.0$\pm$7.1} & 54.7$\pm$16.6  & \multicolumn{1}{@{}c@{} |}{20.0$\pm$6.3} & \multicolumn{1}{@{}c@{} |}
{26.3$\pm$7.2} & \textbf{2.5$\pm$1.0}
\\ \cline{2-18}
& 6 & 34.4$\pm$0.5 & \multicolumn{1}{@{}c@{} |}{\textbf{77.0$\pm$13.5}} & \multicolumn{1}{@{}c@{} |}{\textbf{0.1$\pm$0.2}} & 11.7$\pm$2.9  & \multicolumn{1}{@{}c@{} |}{51.2$\pm$8.4} & \multicolumn{1}{@{}c@{} |}{16.7$\pm$8.0} & 24.1$\pm$6.5  & \multicolumn{1}{@{}c@{} |
}{21.0$\pm$5.9} & \multicolumn{1}{@{}c@{} |}{27.5$\pm$6.1} & 3.0$\pm$0.7  & \multicolumn{1}{@{}c@{} |}{19.9$\pm$5.9} & \multicolumn{1}{@{}c@{} |}{19.8$\pm$5.9} & 51.2$\pm$9.9  & \multicolumn{1}{@{}c@{} |}{18.2$\pm$5.3} & \multicolumn{1}{@{}c@{} |}{27.9$\pm$5.8} & \textbf{1.9$\pm$0.6}
\\ \cline{2-18}
& 9 & 33.8$\pm$0.6 & \multicolumn{1}{@{}c@{} |}{\textbf{65.8$\pm$9.5}} & \multicolumn{1}{@{}c@{} |}{\textbf{0.1$\pm$0.3}} & 17.4$\pm$15.8  & \multicolumn{1}{@{}c@{} |}{46.7$\pm$7.6} & \multicolumn{1}{@{}c@{} |}{20.0$\pm$7.4} & 24.1$\pm$5.8  & \multicolumn{1}{@{}c@{} |
}{21.7$\pm$5.1} & \multicolumn{1}{@{}c@{} |}{27.4$\pm$4.7} & 2.8$\pm$0.7  & \multicolumn{1}{@{}c@{} |}{19.1$\pm$4.9} & \multicolumn{1}{@{}c@{} |}{21.7$\pm$5.4} & 48.0$\pm$11.0  & \multicolumn{1}{@{}c@{} |}{19.8$\pm$4.2} & \multicolumn{1}{@{}c@{} |}{26.9$\pm$6.4} & \textbf{2.2$\pm$1.0}
\\ \cline{2-18}
& 12 & 34.1$\pm$0.5 & \multicolumn{1}{@{}c@{} |}{\textbf{56.0$\pm$7.7}} & \multicolumn{1}{@{}c@{} |}{\textbf{0.1$\pm$0.1}} & 13.4$\pm$5.0  & \multicolumn{1}{@{}c@{} |}{40.9$\pm$5.4} & \multicolumn{1}{@{}c@{} |}{20.2$\pm$4.2} & 26.7$\pm$8.0  & \multicolumn{1}{@{}c@{} |}{
20.8$\pm$4.1} & \multicolumn{1}{@{}c@{} |}{26.0$\pm$6.3} & 3.2$\pm$1.0  & \multicolumn{1}{@{}c@{} |}{16.7$\pm$5.3} & \multicolumn{1}{@{}c@{} |}{22.2$\pm$5.1} & 53.3$\pm$8.7  & \multicolumn{1}{@{}c@{} |}{18.19$\pm$5.2} & \multicolumn{1}{@{}c@{} |}{25.9$\pm$4.0} & \textbf{2.2$\pm$1.1}
\\ \hline
\end{tabular}}
\label{table:varying-sources-500}\
\end{table*}
\begin{table*}[]
\caption{\small\new{Comparing Lift (\%), Disparity (\%), and Runtime (sec), on 1000 node Pokec samples, of LP-Advanced and all baselines, varying $|V_s|$. 20 trials are used at each setting. Both Friend of Friend and Intersecting Group Count are reported. $k = 3$, $p=0.5$}}
\resizebox{\textwidth}{!}{ %
\begin{tabular}{ | @{}c@{} | @{}c @{}| @{}c@{} | @{}c@{}c@{}c@{} | @{}c@{}c@{}c@{} | @{}c@{}c@{}c@{} | @{}c@{}c@{}c@{} | @{}c@{}c@{}c@{} |}\hline
\multirow{2}{*}
{\begin{tabular}[c]{ @ {}c @ {}}Candidate\\ Edges\end{tabular}} &
\multirow{2}{*}{$|V_s|$}
&  \multirow{2}{*}{\begin{tabular}[c]{ @ {} c @ {}}Initial\\ Disparity\end{tabular}}
&  \multicolumn{3}{@{}c@{} |}{LP - Advanced (ours)}
& \multicolumn{3}{@{}c@{} |}{Continuous Greedy}
&  \multicolumn{3}{@{}c@{} |}{IRFA}
& \multicolumn{3}{@{}c@{} |}{SpGreedy}
& \multicolumn{3}{@{}c@{} |}{ACR - FoF}
\\ \cline{4 - 18}
& & & \multicolumn{1}{@{}c@{} |}{Lift} & \multicolumn{1}{@{}c@{} |}{Disparity  } & Time
& \multicolumn{1}{@{}c@{} |}{Lift} & \multicolumn{1}{@{}c@{} |}{Disparity} & Time
& \multicolumn{1}{@{}c@{} |}{Lift} & \multicolumn{1}{@{}c@{} |}{Disparity} & Time
& \multicolumn{1}{@{}c@{} |}{Lift} & \multicolumn{1}{@{}c@{} |}{Disparity} & Time
& \multicolumn{1}{@{}c@{} |}{Lift} & \multicolumn{1}{@{}c@{} |}{Disparity} & Time
\\ \hline
\multirow{4}{*}{FoF}
& 3 & 21.5$\pm$1.2 & \multicolumn{1}{@{}c@{} |}{\textbf{101$\pm$40.8}} & \multicolumn{1}{@{}c@{} |}{\textbf{0.2$\pm$0.6}} & 66.0$\pm$18.0  & \multicolumn{1}{@{}c@{} |}{53.5$\pm$16.2} & \multicolumn{1}{@{}c@{} |}{15.4$\pm$4.1} & 51.1$\pm$11.1  & \multicolumn{1}
{@{}c@{} |}{30.9$\pm$16.1} & \multicolumn{1}{@{}c@{} |}{16.4$\pm$4.2} & \textbf{8.2$\pm$2.6}  & \multicolumn{1}{@{}c@{} |}{30.8$\pm$12.8} & \multicolumn{1}{@{}c@{} |}{11.0$\pm$2.9} & 290$\pm$94.5  & \multicolumn{1}{@{}c@{} |}{36.8$\pm$29.5} & \multicolumn{1}
{@{}c@{} |}{16.2$\pm$3.9} & 19.2$\pm$14.4
\\ \cline{2 - 18}

& 6 & 22.2$\pm$1.4 & \multicolumn{1}{@{}c@{} |}{\textbf{69.6$\pm$21.5}} & \multicolumn{1}{@{}c@{} |}{\textbf{0.1$\pm$0.5}} & 66.2$\pm$19.3  & \multicolumn{1}{@{}c@{} |}{40.5$\pm$9.5} & \multicolumn{1}{@{}c@{} |}{14.4$\pm$3.4} & 43.7$\pm$9.1  & \multicolumn{1}{@{}c@{} |}
{24.8$\pm$13.0} & \multicolumn{1}{@{}c@{} |}{18.3$\pm$3.2} & \textbf{9.0$\pm$3.9}  & \multicolumn{1}{@{}c@{} |}{20.1$\pm$8.2} & \multicolumn{1}{@{}c@{} |}{14.2$\pm$2.5} & 281$\pm$82.0  & \multicolumn{1}{@{}c@{} |}{19.8$\pm$12.1} & \multicolumn{1}{@{}c@{} |}{18.3$\pm$3.9} & 19.0$\pm$11.7
\\ \cline{2 - 18}

& 9 & 21.7$\pm$1.3 & \multicolumn{1}{@{}c@{} |}{\textbf{62.8$\pm$13.6}} & \multicolumn{1}{@{}c@{} |}{\textbf{0.2$\pm$0.3}} & 63.0$\pm$11.2  & \multicolumn{1}{@{}c@{} |}{37.5$\pm$6.1} & \multicolumn{1}{@{}c@{} |}{16.8$\pm$3.2} & 42.7$\pm$11.5  & \multicolumn{1}{@{}c@{}
 |}{20.7$\pm$5.5} & \multicolumn{1}{@{}c@{} |}{18.8$\pm$3.3} & \textbf{8.3$\pm$3.7}  & \multicolumn{1}{@{}c@{} |}{18.2$\pm$3.8} & \multicolumn{1}{@{}c@{} |}{14.0$\pm$3.0} & 257$\pm$46.6  & \multicolumn{1}{@{}c@{} |}{19.7$\pm$6.6} & \multicolumn{1}{@{}c@{} |}
{18.0$\pm$3.0} & 13.7$\pm$4.5
\\ \cline{2 - 18}

& 12 & 21.4$\pm$1.1 & \multicolumn{1}{@{}c@{} |}{\textbf{56.6$\pm$9.6}} & \multicolumn{1}{@{}c@{} |}{\textbf{0.0$\pm$0.0}} & 58.5$\pm$11.9  & \multicolumn{1}{@{}c@{} |}{35.5$\pm$5.4} & \multicolumn{1}{@{}c@{} |}{15.9$\pm$2.7} & 40.1$\pm$7.6  & \multicolumn{1}{@{}c@{} |
}{20.1$\pm$3.6} & \multicolumn{1}{@{}c@{} |}{18.3$\pm$2.7} & \textbf{9.2$\pm$3.9} & \multicolumn{1}{@{}c@{} |}{16.9$\pm$3.4} & \multicolumn{1}{@{}c@{} |}{13.6$\pm$1.9} & 275$\pm$95.3  & \multicolumn{1}{@{}c@{} |}{17.1$\pm$6.4} & \multicolumn{1}{@{}c@{} |}{18.0$\pm$2.5} & 17.1$\pm$15.5
\\ \hline
\multirow{4}{*}{IGC}
& 3 & 21.5$\pm$1.2 & \multicolumn{1}{@{}c@{} |}{\textbf{109$\pm$59.8}} & \multicolumn{1}{@{}c@{} |}{\textbf{1.4$\pm$2.6}} & 124$\pm$69.1  & \multicolumn{1}{@{}c@{} |}{67.0$\pm$43.0} & \multicolumn{1}{@{}c@{} |}{13.2$\pm$4.3} & 63.5$\pm$19.2  & \multicolumn{1}{
@{}c@{} |}{31.8$\pm$18.5} & \multicolumn{1}{@{}c@{} |}{16.4$\pm$4.1} & \textbf{8.8$\pm$3.7}  & \multicolumn{1}{@{}c@{} |}{38.6$\pm$37.8} & \multicolumn{1}{@{}c@{} |}{11.8$\pm$4.9} & 345$\pm$118  & \multicolumn{1}{@{}c@{} |}{23.7$\pm$16.3} & \multicolumn{1}{@{}c@{}
|}{17.6$\pm$3.4} & 45.1$\pm$37.0
\\ \cline{2-18}
& 6 & 22.2$\pm$1.4 & \multicolumn{1}{@{}c@{} |}{\textbf{70.1$\pm$30.2}} & \multicolumn{1}{@{}c@{} |}{\textbf{1.2$\pm$2.2}} & 102$\pm$52.3  & \multicolumn{1}{@{}c@{} |}{45.9$\pm$20.5} & \multicolumn{1}{@{}c@{} |}{13.7$\pm$3.8} & 54.1$\pm$15.8  & \multicolumn{1
}{@{}c@{} |}{26.5$\pm$16.1} & \multicolumn{1}{@{}c@{} |}{18.5$\pm$3.2} & \textbf{8.3$\pm$3.1}  & \multicolumn{1}{@{}c@{} |}{21.8$\pm$14.4} & \multicolumn{1}{@{}c@{} |}{14.6$\pm$3.9} & 336$\pm$108  & \multicolumn{1}{@{}c@{} |}{16.7$\pm$10.4} & \multicolumn{1}
{@{}c@{} |}{19.0$\pm$3.1} & 36.7$\pm$32.0
\\ \cline{2-18}
& 9 & 21.7$\pm$1.2 & \multicolumn{1}{@{}c@{} |}{\textbf{62.5$\pm$16.9}} & \multicolumn{1}{@{}c@{} |}{\textbf{0.8$\pm$1.8}} & 89.6$\pm$29.9  & \multicolumn{1}{@{}c@{} |}{41.7$\pm$11.7} & \multicolumn{1}{@{}c@{} |}{14.9$\pm$2.6} & 51.1$\pm$16.0  & \multicolumn{1}{
@{}c@{} |}{21.5$\pm$6.6} & \multicolumn{1}{@{}c@{} |}{18.1$\pm$3.4} & \textbf{9.5$\pm$4.0}  & \multicolumn{1}{@{}c@{} |}{20.2$\pm$8.1} & \multicolumn{1}{@{}c@{} |}{14.7$\pm$3.0} & 306$\pm$77.0  & \multicolumn{1}{@{}c@{} |}{15.4$\pm$7.6} & \multicolumn{1}{@{}c@{} |}{19.4$\pm$3.2} & 26.3$\pm$16.7
\\ \cline{2-18}
& 12 & 21.4$\pm$1.1 & \multicolumn{1}{@{}c@{} |}{\textbf{56.6$\pm$14.8}} & \multicolumn{1}{@{}c@{} |}{\textbf{0.8$\pm$1.4}} & 84.3$\pm$35.7  & \multicolumn{1}{@{}c@{} |}{37.7$\pm$9.5} & \multicolumn{1}{@{}c@{} |}{14.7$\pm$2.5} & 47.8$\pm$17.6  & \multicolumn{1}{@{}c@{}
 |}{21.0$\pm$5.2} & \multicolumn{1}{@{}c@{} |}{18.0$\pm$2.6} & \textbf{10.3$\pm$3.8}  & \multicolumn{1}{@{}c@{} |}{19.7$\pm$8.4} & \multicolumn{1}{@{}c@{} |}{14.6$\pm$2.1} & 312$\pm$114  & \multicolumn{1}{@{}c@{} |}{14.8$\pm$7.4} & \multicolumn{1}{@{}c@{} |
}{19.4$\pm$2.7} & 32.8$\pm$40.1
\\ \hline

\end{tabular}}
\label{table:varying-sources-1000}
\end{table*}

\begin{table*}[]
\caption{\small\new{Scaling results comparing  Lift (in \%), Disparity (in \%), and Runtime (in seconds), while varying graph size from 500 nodes to 500k nodes. $k$ = 3, $p$ = 0.5 and Candidate Edges are FoF for all settings. 5 trials were used up to 400k nodes, then only 1 trial is used. $|V_s| = 5$ for 500-10k nodes, and $|V_s| = 10$ for 20k+ nodes. If an algorithm took longer than 24 hours, results were not reported, and IRFA cannot be run greater than 20k due to preprocessing constraints.}}
\resizebox{\textwidth}{!}{%
\begin{tabular}{|@{}c@{}|@{}c@{}|@{}c@{}|@{}c@{}|@{}c@{}|@{}c@{}c@{}c@{}|@{}c@{}c@{}c@{}|@{}c@{}c@{}c@{}|@{}c@{}c@{}c@{}|@{}c@{}c@{}c@{}|}
\hline
\multirow{2}{*}{Nodes} & \multirow{2}{*}{\begin{tabular}[c]{@{}c@{}}Initial\\ Disparity\end{tabular}} & \multicolumn{3}{c|}{LP-Advanded (ours)}                                                           & \multicolumn{3}{@{}c@{}|}{LP-SCALE (ours)}                                                        & \multicolumn{3}{@{}c@{}|}{Continous Greedy}                                                   & \multicolumn{3}{@{}c@{}|}{IRFA}                                                               & \multicolumn{3}{@{}c@{}|}{SpGreedy}                                                            & \multicolumn{3}{@{}c@{}|}{ACR-FoF}                                                             \\ \cline{3-20}
                       &                                                                              & \multicolumn{1}{@{}c@{}|}{Lift}          & \multicolumn{1}{@{}c@{}|}{Disparity}   & Runtime            & \multicolumn{1}{@{}c@{}|}{Lift}          & \multicolumn{1}{@{}c@{}|}{Disparity}   & Runtime      & \multicolumn{1}{@{}c@{}|}{Lift}          & \multicolumn{1}{@{}c@{}|}{Disparity}    & Runtime        & \multicolumn{1}{@{}c@{}|}{Lift}         & \multicolumn{1}{@{}c@{}|}{Disparity}     & Runtme         & \multicolumn{1}{@{}c@{}|}{Lift}         & \multicolumn{1}{@{}c@{}|}{Disparity}   & Runtime           & \multicolumn{1}{@{}c@{}|}{Lift}         & \multicolumn{1}{@{}c@{}|}{Disparity}    & Runtime          \\ \hline
500                    & 16.6$\pm$0.1                                                                 & \multicolumn{1}{@{}c@{}|}{56.5$\pm$15.2} & \multicolumn{1}{@{}c@{}|}{\textbf{0.0$\pm$0.0}} & 13.7$\pm$4.4        & \multicolumn{1}{@{}c@{}|}{\textbf{57.8$\pm$14.2}} & \multicolumn{1}{@{}c@{}|}{0.2$\pm$0.2} & 14.1$\pm$5.5  & \multicolumn{1}{@{}c@{}|}{36.9$\pm$11.7} & \multicolumn{1}{@{}c@{}|}{11.0$\pm$4.6} & 14.9$\pm$2.3   & \multicolumn{1}{@{}c@{}|}{19.1$\pm$2.8} & \multicolumn{1}{@{}c@{}|}{14.0$\pm$7.8}  & \textbf{3.1$\pm$0.8}    & \multicolumn{1}{@{}c@{}|}{20.8$\pm$7.4}  & \multicolumn{1}{@{}c@{}|}{9.0$\pm$5.0} & 83.9$\pm$16.1     & \multicolumn{1}{@{}c@{}|}{16.7$\pm$8.7}  & \multicolumn{1}{@{}c@{}|}{14.3$\pm$5.6} & 6.7$\pm$3.6      \\ \hline
1000                   & 15.2$\pm$0.1                                                                 & \multicolumn{1}{@{}c@{}|}{\textbf{55.4$\pm$13.7}}  & \multicolumn{1}{@{}c@{}|}{\textbf{0.0$\pm$0.0}}    & 73.4$\pm$13.8       & \multicolumn{1}{@{}c@{}|}{53.0$\pm$12.9} & \multicolumn{1}{@{}c@{}|}{0.6$\pm$0.4} & 14.0$\pm$4.3 & \multicolumn{1}{@{}c@{}|}{35.4$\pm$8.1}   & \multicolumn{1}{@{}c@{}|}{10.7$\pm$3.4} & 37.0$\pm$9.6   & \multicolumn{1}{@{}c@{}|}{20.2$\pm$5.5} & \multicolumn{1}{@{}c@{}|}{12.6$\pm$5.3} & \textbf{13.2$\pm$2.5}   & \multicolumn{1}{@{}c@{}|}{15.6$\pm$5.4} & \multicolumn{1}{@{}c@{}|}{9.8$\pm$3.7} & 332$\pm$70.0    & \multicolumn{1}{@{}c@{}|}{11.1$\pm$4.4} & \multicolumn{1}{@{}c@{}|}{13.7$\pm$4.8}  & 40.4$\pm$18.9    \\ \hline
2000                   & 11.9$\pm$0.0                                                                 & \multicolumn{1}{@{}c@{}|}{\textbf{54.3$\pm$12.5}} & \multicolumn{1}{@{}c@{}|}{\textbf{0.0$\pm$0.0}}       & 1054$\pm$346      & \multicolumn{1}{@{}c@{}|}{45.4$\pm$9.3}  & \multicolumn{1}{@{}c@{}|}{0.4$\pm$0.4} & \textbf{19.7$\pm$6.8} & \multicolumn{1}{@{}c@{}|}{36.7$\pm$8.1}  & \multicolumn{1}{@{}c@{}|}{8.3$\pm$2.0}  & 55.2$\pm$9.2   & \multicolumn{1}{@{}c@{}|}{16.2$\pm$6.8} & \multicolumn{1}{@{}c@{}|}{10.6$\pm$2.6}  & 74.7$\pm$13.1  & \multicolumn{1}{@{}c@{}|}{14.2$\pm$5.8} & \multicolumn{1}{@{}c@{}|}{7.2$\pm$3.4} & 1597$\pm$369 & \multicolumn{1}{@{}c@{}|}{10.9$\pm$5.4}   & \multicolumn{1}{@{}c@{}|}{10.9$\pm$3.01} & 285$\pm$141 \\ \hline
4000                   & 11.5$\pm$0.0                                                                 & \multicolumn{1}{@{}c@{}|}{\textbf{58.0$\pm$13.8}} & \multicolumn{1}{@{}c@{}|}{\textbf{0.0$\pm$0.0}}       & 1.2e4$\pm$3776 & \multicolumn{1}{@{}c@{}|}{45.1$\pm$7.7}  & \multicolumn{1}{@{}c@{}|}{1.3$\pm$0.7} & \textbf{26.0$\pm$8.9} & \multicolumn{1}{@{}c@{}|}{37.1$\pm$8.7}  & \multicolumn{1}{@{}c@{}|}{7.6$\pm$1.0}  & 151$\pm$34.1 & \multicolumn{1}{@{}c@{}|}{15.2$\pm$8.7}  & \multicolumn{1}{@{}c@{}|}{9.3$\pm$1.3}   & 337$\pm$44.2 & \multicolumn{1}{@{}c@{}|}{12.34$\pm$4.2} & \multicolumn{1}{@{}c@{}|}{6.8$\pm$1.6}  & 7391$\pm$1267 & \multicolumn{1}{@{}c@{}|}{10.8$\pm$4.5} & \multicolumn{1}{@{}c@{}|}{11.1$\pm$1.4} & 1355$\pm$490    \\ \hline
10k                    & 8.67                                                                         & \multicolumn{1}{@{}c@{}|}{-}             & \multicolumn{1}{@{}c@{}|}{-}           & -                  & \multicolumn{1}{@{}c@{}|}{\textbf{75.6}}         & \multicolumn{1}{@{}c@{}|}{\textbf{0.5}}        & \textbf{39.2}        & \multicolumn{1}{@{}c@{}|}{60.8}         & \multicolumn{1}{@{}c@{}|}{4.7}         & 370         & \multicolumn{1}{@{}c@{}|}{17.4}         & \multicolumn{1}{@{}c@{}|}{8.4}          & 442         & \multicolumn{1}{@{}c@{}|}{16.4}        & \multicolumn{1}{@{}c@{}|}{3.4}        & 4.9e4          & \multicolumn{1}{@{}c@{}|}{16.9}        & \multicolumn{1}{@{}c@{}|}{7.0}         & 4754          \\ \hline
20k                    & 2.3                                                                         & \multicolumn{1}{@{}c@{}|}{-}             & \multicolumn{1}{@{}c@{}|}{-}           & -                  & \multicolumn{1}{@{}c@{}|}{\textbf{27.6}}         & \multicolumn{1}{@{}c@{}|}{\textbf{0.1}}        & \textbf{122}       & \multicolumn{1}{@{}c@{}|}{26.6}         & \multicolumn{1}{@{}c@{}|}{2.3}         & 943         & \multicolumn{1}{@{}c@{}|}{6.1}         & \multicolumn{1}{@{}c@{}|}{3.4}          & 1.2e4       & \multicolumn{1}{@{}c@{}|}{-}            & \multicolumn{1}{@{}c@{}|}{-}           & -                 & \multicolumn{1}{@{}c@{}|}{-}            & \multicolumn{1}{@{}c@{}|}{-}            & -                \\ \hline
50k                    & 3.3                                                                         & \multicolumn{1}{@{}c@{}|}{-}             & \multicolumn{1}{@{}c@{}|}{-}           & -                  & \multicolumn{1}{@{}c@{}|}{\textbf{36.7}}         & \multicolumn{1}{@{}c@{}|}{\textbf{0.0}}        & \textbf{365}       & \multicolumn{1}{@{}c@{}|}{36.7}         & \multicolumn{1}{@{}c@{}|}{2.1}         & 3498        & \multicolumn{1}{@{}c@{}|}{-}            & \multicolumn{1}{@{}c@{}|}{-}             & -              & \multicolumn{1}{@{}c@{}|}{-}            & \multicolumn{1}{@{}c@{}|}{-}           & -                 & \multicolumn{1}{@{}c@{}|}{-}            & \multicolumn{1}{@{}c@{}|}{-}            & -                \\ \hline
100k                   & 1.8                                                                         & \multicolumn{1}{@{}c@{}|}{-}             & \multicolumn{1}{@{}c@{}|}{-}           & -                  & \multicolumn{1}{@{}c@{}|}{26.8}         & \multicolumn{1}{@{}c@{}|}{\textbf{0.0}}        & \textbf{1222}      & \multicolumn{1}{@{}c@{}|}{\textbf{39.8}}         & \multicolumn{1}{@{}c@{}|}{1.2}         & 1.3e4       & \multicolumn{1}{@{}c@{}|}{-}            & \multicolumn{1}{@{}c@{}|}{-}             & -              & \multicolumn{1}{@{}c@{}|}{-}            & \multicolumn{1}{@{}c@{}|}{-}           & -                 & \multicolumn{1}{@{}c@{}|}{-}            & \multicolumn{1}{@{}c@{}|}{-}            & -                \\ \hline
200k                   & 2.2                                                                         & \multicolumn{1}{@{}c@{}|}{-}             & \multicolumn{1}{@{}c@{}|}{-}           & -                  & \multicolumn{1}{@{}c@{}|}{30.1}         & \multicolumn{1}{@{}c@{}|}{\textbf{0}}           & \textbf{5160}       & \multicolumn{1}{@{}c@{}|}{\textbf{37}}         & \multicolumn{1}{@{}c@{}|}{1.9}         & 4.7e4       & \multicolumn{1}{@{}c@{}|}{-}            & \multicolumn{1}{@{}c@{}|}{-}             & -              & \multicolumn{1}{@{}c@{}|}{-}            & \multicolumn{1}{@{}c@{}|}{-}           & -                 & \multicolumn{1}{@{}c@{}|}{-}            & \multicolumn{1}{@{}c@{}|}{-}            & -                \\ \hline
500k                   & 3.9                                                                         & \multicolumn{1}{@{}c@{}|}{-}             & \multicolumn{1}{@{}c@{}|}{-}           & -                  & \multicolumn{1}{@{}c@{}|}{\textbf{43.3}}         & \multicolumn{1}{@{}c@{}|}{\textbf{0.0}}        & \textbf{7.4e4}     & \multicolumn{1}{@{}c@{}|}{-}             & \multicolumn{1}{@{}c@{}|}{-}            & -              & \multicolumn{1}{@{}c@{}|}{-}            & \multicolumn{1}{@{}c@{}|}{-}             & -              & \multicolumn{1}{@{}c@{}|}{-}            & \multicolumn{1}{@{}c@{}|}{-}           & -                 & \multicolumn{1}{@{}c@{}|}{-}            & \multicolumn{1}{@{}c@{}|}{-}            & -                \\ \hline
\end{tabular}}
\label{tab:scaling}
\end{table*}

\new{Additional experiment results on configuring LP-SCALE are included in the appendix.}

\section{Related Work}\label{sec:related}


Propagation in networks has been extensively studied in the literature \cite{10.1145/502512.502525, 10.1145/775047.775057}. 
There are two main approaches on propagation optimization, 1. (influence maximization) the orthogonal problem of selecting nodes to maximize propagation, and also, 2. (content spread) targeting the topology of the network and selecting edges to recommend to users to maximize propagation.

\stitle{Influence Maximization} The problem of finding a set of most influential (seed) nodes in a graph is called Influence Maximization (IM) \cite{10.1145/502512.502525,li2018influence, banerjee2020survey}.
The first discrete optimization definition of Influence Maximization
was proposed in \cite{kempe-infmax}, where relying on sub-modularity of the Influence function, the author proposed a greedy algorithm for the problem.
Furthermore, practical optimizations have been proposed \cite{Leskovec:GraphsOT,Chen-Scalable} to address the scalability issue of the greedy algorithm.

\stitle{Content Spread}
Unlike the IM research that finds a set of seed {\em nodes}, the objective of the Content Spread Maximization problem is to find a set of {\em potential edges} that once added, maximize the content spread \cite{yu-contentspread, Silva2010AGF,chaoji-cs, yang2019marginal}. \new{The candidate set of edges can be considered an input into this problem for many methods, allowing it to wrap various friendship suggestion methods \cite{roth-implicit, epasto2015ego, hamid2014cohesion, wang2014friendbook}.}
With the introduction of the RMPP cascade model, \cite{chaoji-cs} proposes the Continuous Greedy approximation algorithm, achieving a content spread $90$ times higher than other heuristics and greedy strategies. \new{Content spread algorithm have also been derived to target the content spread capabilities of the entire network as opposed to a given set of sources \cite{yu-contentspread}. Recent works in content spread have optimized the process by considering only the influence ranking of a single node \cite{yang2019marginal}.}
Our paper falls under this category of Content Spread algorithms, adding fairness as a constraint to the problem.
However, as shown by Theorem~\ref{th:inapprox} and demonstrated in our experiments, existing approaches fail to extend for our problem setting.

\stitle{\new{Social Network Polarization and Opinion Formation}}
\new{Groups in the network often take opposing opinions, and measuring the structure of the network in relation to these opinions is the problem of opinion formation \cite{friedkin1990social, chen2018quantifying, gionis2013opinion}, and polarization \cite{guerra2013measure, garimella2018political, chitra2020analyzing}. A popular approach for this line of work is based on the Friedkin-Johnsen model \cite{friedkin1990social}. In this model, individuals have an innate opinion, and then based on that and the topology of the network, they establish an expressed opinion. Bindel et al. recognized that this opinion could be calculated using matrix multiplication \cite{bindel2015bad}. Several authors built upon this idea establishing calculation of polarization and other metrics using the matrix multiplication technique \cite{musco2018minimizing, matakos2017measuring, chitra2020analyzing}. 
Related to our work, several authors used the method of edge suggestion in relation to network opinions in order to alter polarization and
other metrics \cite{zhu2021minimizing, Masrour_Wilson_Yan_Tan_Esfahanian_2020, garimella2017reducing}. 
The main difference between these works and ours is that these works target conflict between groups on the social network, while ours targets {equitable} reception of content between groups, making it suited for solving \problem.} 


\stitle{Fairness}
Machine bias and fairness in data science has become timely across different research communities, especially in ML~\cite{fairmlbook,barocas2016big,dwork2012fairness,kearns2018preventing,narayanan2018translation,friedler2019comparative,kamiran2012data,feldman2015certifying,calmon2017optimized,kamishima2012fairness} and data management~\cite{islam2022through,asudeh2019assessing,nargesian2022responsible,nargesian2021tailoring,shah2020sigmod,asudeh2020fairly,zhang2021omnifair, shahbazi2022survey,asudeh2021identifying,suraj2022fair,wei2022rank,zhang2021fairrover,galhotra2022causal,stoyanovich2022responsible}.
While extensive research focused on the scalability and efficiency in this literature, recently fairness started to come into the picture in propagation optimization.  More recently, equitable access to information has raised questions. As a result, a related line of research started with proposing fairness-aware frameworks. As shown in \cite{10.1145/3308558.3313680}, the structure of the network can cause individuals from minority groups to have less access to the content across the network. 
In influence maximization, while works like \cite{10.1145/3308558.3313680} uses a maxmin formulation to impose fairness in Influence Maximization problem, other works address 
fairness by defining fairness constraints \cite{ GUNEY2019589,becker2020fairness, FarnadiFIM}, equivalent objective functions that encounter fairness notions \cite{rahmattalabi2020fair}, 
\new{assure that a certain group is influenced in addition to overall influence \cite{gershtein2018balanced},} a new multi-objective model that considers a fair influence spread among different groups \cite{tsang2019group, yu2017fair}\new{, or even a Mixed Integer linear Programming model which can satisfy a variety of fairness constraints \cite{FarnadiFIM}.}
\balance
Also, ML community has combined the adversarial networks with fairness concepts in order to design fairness aware methods in both link prediction \cite{Masrour_Wilson_Yan_Tan_Esfahanian_2020} and node selection \cite{khajehnejad2020adversarial}.
Still, to the best of our knowledge, this paper is the first to consider fairness in the Content Spread approach to propagation optimization.



\section{Conclusion and Future Work}\label{sec:conclusion}
We proposed the problem of edge (friendship) suggestion for maximizing the fair content spread.
We showed the problem is $\NP$-hard and proved its inapproximability in polynomial time.
Then allowing approximation both on fairness and on the objective value, we designed a randomized $\LP$-relaxation algorithm with fixed approximation ratios.
We proposed practical optimizations and designed a scalable algorithm that \new{dynamically adds the solutions of sub-problems with reasonably small sizes to the graph, updating the solution incrementally}.
We conducted comprehensive experiments on real and synthetic data sets to evaluate our algorithms.
Our results confirm the effectiveness of our algorithms, being able to find solutions with near-zero unfairness while significantly increasing the content spread.
Our scalable algorithm could scale up to {\em half a million nodes} in a reasonable time, reducing unfairness down to around 0.0004 while increasing the content spread by 43\%. 

\new{
We study the problem of edge suggestion without considering the recommendation acceptance probabilities. While a simple resolution may be to recommend a higher number of edges, addressing this problem requires new techniques we consider part of our future work. 
Additionally, in future works, we will explore alternative scenarios, such as determining fair content spread when all nodes produce content and the content is of a variety of topics.}
\section{Ethics Statement}

While this paper studies friendship suggestions in social networks that maximizes content spread while achieving fairness, all experiments performed herein are based on publicly available datasets in academic settings. We would like to note that this approach has not been integrated in any real world product, service, or platform. As a practical constraint, the assumption of the availability of user profile information may not be realistic due to the fact that by design none of the social networks operating today collects user demographic data. This paper is based on the fact that unfair spread of data still exists and it aims to address balanced data flow among different sets of nodes/users. Moreover, we highlight that adding such features to social networks requires rigorous regulation to avoid exploiting such features and cause ethical or societal implications. For example, even though this method is designed to bring benefits to users and maximize fairness, one may change the system setup to do the opposite. We recognize that such ethical concerns are far more complex and need policymakers to look into how platforms can detect such misusage. We hope our work could help effect a positive change in this direction.

\vspace{-2mm}
\begin{acks}
\vspace{-1mm}
This work was supported in part by the National Science Foundation 
(NSF 2107290) and the Google research scholar award.
We also thank Dr. Golnoosh Farnadi for sharing their dataset with us.
\end{acks}

\bibliographystyle{ACM-Reference-Format}
\balance
\bibliography{ref,abol,tutorialref}

\newpage
\section*{Appendix}
\appendix
\section{Proofs}
{\sc Lemma}\ref{lem:complexity}.
{\it The Fair Content Spread Maximization (\problem) Problem is $\NP$-hard.}

\begin{proof} 
To prove that \problem is $\NP$-hard, we use the reduction from the content spread problem (\cs) as in Definition~\ref{def:spread}, which is known to be $\NP$-complete~\cite{chaoji-cs}.
Consider an instance of the \cs problem (Definition~\ref{def:spread}), on the graph $G'(V',E')$, the set of candidate edges $E'_c$, the content nodes $V'_s$, and the value $k'$ on the number of incident edges selected for each node $v'_i\in V'$. Also, let the decision value for content spread be $T'$.
We generate an instance of \problem problem from the \cs instance as follows:
\begin{itemize}[leftmargin=*]
    \item Let $G''(V'',E'')$ be a duplicate of $G'(V',E')$. That is, for each node $v'\in V'$, there exists $v''\in V''$, and for each edge $(v'_i,v'_j)\in E'$, there \new{exists} an edge $(v''_i,v''_j)\in E''$.
    We define the graph $G$ as the union of $G'$ and $G''$. That is $G(V'\cup V'', E'\cup E'')$.
    \item Let $\Gee=\{\gee_1,\gee_2\}$, where every node $v'\in V'$ belongs to $\gee_1$ and every node $v''\in V''$ belongs to $\gee_2$.
    \item Set $V_s = V'_s\cup \{v''_i~|~v'_i\in V'_s\}$. Set $k=k'$ and $T=2T'$.
\end{itemize}
Now the answer to the \cs problem is yes (i.e., there \new{exists} a selection of $E'_p$ such that satisfies the content spread of $T'$) iff the answer to the \problem problem with input $G$, $V_s$, $\Gee$, $k$, and the decision value $T$ is yes. 
\end{proof}

{\sc Lemma}\ref{lem:3}.
{\it The expected number of selected edges by Algorithm~\ref{algo:LP} per each node $v_i\in V$ is $k$.}

\begin{proof}
Following the rounding step, the probability of selecting an edge $e_j\in E_c$ is $\pr(y^+_j)=y^*_j$.
The expected number of edges incident to a node $v_i$ is
\begin{align*}
    \EX\big[ \sum_{e_j \in \mathcal{N}(v_i)}y^+_j \big] = 
    \sum_{e_j \in \mathcal{N}(v_i)}\pr(y^+_j)
    = \sum_{e_j \in \mathcal{N}(v_i)}y^*_j = k
\end{align*}
\end{proof}
\section{Practical Optimizations}

\stitle{Variable Reduction} 
The number of variables in the $\LP$ formulation directly impacts its run time. 
To make our algorithm efficient we prune variables which would have no impact on the output.
%
%
Let $r_i$ be the shortest distance from content nodes to a node $v_i$ without including any edges from $E_c$
. In our $\LP$ formulation in Figure~\ref{fig:LPRelax}, for every length of path of $r\geq r_i$, there exists a path which satisfies variables $x_{ir}$. That is, $\forall r\geq r_i,~x_{ir}=1$.
All such variables can be replaced with the constant $1$, both in the objective function, as well as the fairness enforcing constraints.
Additionally, for cases where the set of edges which reduce the distance to a node $v_i$ to at most $r$ is empty
, there is no edge selection for which $x_{ir}$ is > 0.
In such cases, we replace $x_{ir}$ by a constant zero, removing the corresponding unnecessary variables.
In the friend of friend edge suggestion model~\cite{yu-contentspread}, in particular, this reduces the order of the variables to consider. Since an edge in this model can only improve the distance by one, there is only one $x_{ir}$ variable for each node for which we are unsure of the final value. Consequently the number of variables in $\LP$ reduces from $O(r_m\,|V| + |E_c|)$ down to $O(|V|+|E_c|)$.

\stitle{Computing the shortest paths}
Computing the shortest paths of an entire graph is relatively inefficient, taking $\Theta(|V| + |E|)$ time. Our algorithm needs to compute the effect of each edge $e_j$ in $E_c$ on the graph when added to the graph independently. However, when only one or more edges are added to a graph, there is a deterministic subset of nodes $V' \subset V$ for which the shortest distances improve. Algorithm~\ref{algo:shortestpaths} calculates only the changes in the distance for the nodes in $V'$. The basic idea is that the search starts at the new edge or edges, and performs a search, not following an edge if traversing over it does not improve the distance to the destination. If the union of all the outbound edges of all nodes in $V'$ is called $E_{V'}$, then this algorithm runs in $\Theta(|V'| + |E_{V'}|)$. For the clear theoretical improvement, we used this algorithm in our implementation.

\begin{algorithm}[!t]
\small
    \caption{{\sc ShortestDistances}}\label{algo:shortestpaths}
    \hspace*{\algorithmicindent} \textbf{Input}: $G(V, E)$, $V_s$ $E_c$\\
    \hspace*{\algorithmicindent} \textbf{Output}: Collection of shortest distances $\mathcal{S}'$
    \begin{algorithmic}[1]
        \State $\mathbf{d}\gets$ {\sc multi-src-shortest-distances} $(G(V,E),V_s)$
        \State {\bf for} $v_i\in V\backslash V_s$ {\bf do} $S'_i\gets \{d_i\}$
        \For{$e_j\in E_c$}
            \State let $e_j=(v_\ell,v_i)$ where $d_{\ell}\leq d_i$
            \State {\bf if} $d_i\leq d_{\ell}+1$ {\bf then continue}
            \State $d_{ij}\gets d_{\ell}+1$; add $\langle v_i, d_{ij}\rangle$ to queue
            \While{queue is not empty}
                \State remove $\langle v_i, d_{ij}\rangle$ from queue
                \State add $(e_j,d_{ij})$ to $S'_i$
                \For{$v_{i'}\in edges(v_i)$}
                    \If {$d_{ij}+1<d_{i'}$}
                        \State $d_{i'j}\gets d_{ij}+1$
                        \State add $\langle v_{i'}, d_{i'j}\rangle$ to queue
                    \EndIf
                \EndFor
            \EndWhile
        \EndFor
        \Return $\mathcal{S}' = \{S'_i~|~v_i\in V\backslash V_s\}$
    \end{algorithmic}
\end{algorithm}

\section{Extensive Experiment Results}

\begin{table*}[]
\caption{(Best results bold) Comparing Lift (in \%), Disparity (in \%), and Runtime (in seconds), on 500 node Antelope Valley data sets, of LP-Advanced and all baselines, varying $k$. 20 trials are used at each setting. Both Friend of Friend and Intersecting Group Count are reported. $p = 0.5$, $|V_s|=3$}
\label{table:varying-k-500}
\vspace{-4mm}
\resizebox{\textwidth}{!}{ %
\begin{tabular}{ | c | c | c | ccc | ccc | ccc | ccc | ccc |}\hline
\multirow{2}{*}
{\begin{tabular}[c]{ @ {}c @ {}}Candidate\\ Edges\end{tabular}} &
\multirow{2}{*}{k}
&  \multirow{2}{*}{\begin{tabular}[c]{ @ {} c @ {}}Initial\\ Disparity\end{tabular}}
&  \multicolumn{3}{c |}{LP - Advanced} 
& \multicolumn{3}{c |}{Continuous Greedy}
&  \multicolumn{3}{c |}{IRFA}
& \multicolumn{3}{c |}{SpGreedy}
& \multicolumn{3}{c |}{ACR - FoF}
\\ \cline{4 - 18}
& & & \multicolumn{1}{c |}{Lift} & \multicolumn{1}{c |}{Disparity  } & Time
& \multicolumn{1}{c |}{Lift} & \multicolumn{1}{c |}{Disparity} & Time
& \multicolumn{1}{c |}{Lift} & \multicolumn{1}{c |}{Disparity} & Time
& \multicolumn{1}{c |}{Lift} & \multicolumn{1}{c |}{Disparity} & Time
& \multicolumn{1}{c |}{Lift} & \multicolumn{1}{c |}{Disparity} & Time
\\ \hline
\multirow{4}{*}{FoF}
& 2 & 34.2$\pm$0.34 & \multicolumn{1}{c |}{\textbf{66.08$\pm$10.76}} & \multicolumn{1}{c |}{\textbf{0.14$\pm$0.51}} & 9.9$\pm$3.09  & \multicolumn{1}{c |}{45.71$\pm$9.04} & \multicolumn{1}{c |}{25.71$\pm$9.34} & 28.83$\pm$8.47  & \multicolumn{1}{c |}{
15.15$\pm$6.4} & \multicolumn{1}{c |}{28.9$\pm$6.81} & 2.62$\pm$1.07  & \multicolumn{1}{c |}{11.73$\pm$3.29} & \multicolumn{1}{c |}{24.11$\pm$5.15} & 29.88$\pm$8.38  & \multicolumn{1}{c |}{19.65$\pm$8.78} & \multicolumn{1}{c |}{28.14$\pm$6.53} & \textbf{1.57$\pm$0.81}
\\ \cline{2 - 18}
& 3 & 34.2$\pm$0.34 & \multicolumn{1}{c |}{\textbf{84.04$\pm$13.64}} & \multicolumn{1}{c |}{\textbf{0.08$\pm$0.21}} & 13.94$\pm$15.24  & \multicolumn{1}{c |}{51.68$\pm$8.54} & \multicolumn{1}{c |}{25.2$\pm$6.36} & 29.01$\pm$8.13  & \multicolumn{1}{c |
}{23.03$\pm$7.4} & \multicolumn{1}{c |}{29.92$\pm$9.21} & 3.88$\pm$1.3  & \multicolumn{1}{c |}{20.11$\pm$9.03} & \multicolumn{1}{c |}{21.08$\pm$6.1} & 40.73$\pm$9.68  & \multicolumn{1}{c |}{27.06$\pm$9.0} & \multicolumn{1}{c |}{24.98$\pm$7.13} & \textbf{1.77$\pm$0.91}
\\ \cline{2 - 18}

& 4 & 34.2$\pm$0.34 & \multicolumn{1}{c |}{\textbf{94.23$\pm$15.54}} & \multicolumn{1}{c |}{\textbf{0.33$\pm$0.9}} & 11.36$\pm$4.08  & \multicolumn{1}{c |}{56.41$\pm$8.87} & \multicolumn{1}{c |}{22.81$\pm$5.83} & 32.37$\pm$8.15  & \multicolumn{1}{c |}
{32.2$\pm$10.31} & \multicolumn{1}{c |}{27.21$\pm$11.12} & 4.9$\pm$2.55  & \multicolumn{1}{c |}{25.82$\pm$9.66} & \multicolumn{1}{c |}{19.41$\pm$6.91} & 53.07$\pm$9.45  & \multicolumn{1}{c |}{33.59$\pm$9.47} & \multicolumn{1}{c |}{23.5
2$\pm$8.13} & \textbf{1.86$\pm$1.02}
\\ \cline{2 - 18}
& 5 & 34.2$\pm$0.34 & \multicolumn{1}{c |}{\textbf{105.28$\pm$18.91}} & \multicolumn{1}{c |}{\textbf{1.68$\pm$3.66}} & 13.28$\pm$7.07  & \multicolumn{1}{c |}{60.12$\pm$7.92} & \multicolumn{1}{c |}{24.26$\pm$7.27} & 34.2$\pm$10.36  & \multicolumn{1}{c 
|}{37.96$\pm$10.26} & \multicolumn{1}{c |}{25.07$\pm$11.33} & 6.81$\pm$2.65  & \multicolumn{1}{c |}{32.17$\pm$11.14} & \multicolumn{1}{c |}{17.13$\pm$7.8} & 56.92$\pm$13.45  & \multicolumn{1}{c |}{40.53$\pm$10.65} & \multicolumn{1}{c |
}{22.54$\pm$7.69} & \textbf{2.05$\pm$1.16}
\\ \hline
\multirow{4}{*}{IGC}
& 2 & 34.2$\pm$0.34 & \multicolumn{1}{c |}{\textbf{78.69$\pm$17.11}} & \multicolumn{1}{c |}{\textbf{0.77$\pm$2.04}} & 14.23$\pm$4.64  & \multicolumn{1}{c |}{54.1$\pm$13.78} & \multicolumn{1}{c |}{22.17$\pm$10.03} & 31.4$\pm$11.12  & \multicolumn{1}{c 
|}{19.51$\pm$7.3} & \multicolumn{1}{c |}{30.1$\pm$6.8} & 2.1$\pm$0.93  & \multicolumn{1}{c |}{16.95$\pm$10.81} & \multicolumn{1}{c |}{23.11$\pm$7.53} & 35.37$\pm$9.83  & \multicolumn{1}{c |}{13.91$\pm$4.66} & \multicolumn{1}{c |}{29.54
$\pm$5.4} & \textbf{2.04$\pm$0.98}
\\ \cline{2-18}
& 3 & 34.2$\pm$0.34 & \multicolumn{1}{c |}{\textbf{94.26$\pm$21.55}} & \multicolumn{1}{c |}{\textbf{0.55$\pm$1.46}} & 15.65$\pm$5.94  & \multicolumn{1}{c |}{60.92$\pm$13.1} & \multicolumn{1}{c |}{23.51$\pm$11.15} & 32.7$\pm$11.51  & \multicolumn{1}{c 
|}{27.23$\pm$10.76} & \multicolumn{1}{c |}{28.21$\pm$8.44} & 2.83$\pm$1.07  & \multicolumn{1}{c |}{23.62$\pm$11.58} & \multicolumn{1}{c |}{19.89$\pm$7.09} & 50.76$\pm$13.59  & \multicolumn{1}{c |}{20.04$\pm$6.29} & \multicolumn{1}{c |}
{26.28$\pm$7.18} & \textbf{2.5$\pm$1.35}
\\ \cline{2-18}
& 4 & 34.2$\pm$0.34 & \multicolumn{1}{c |}{\textbf{108.58$\pm$22.96}} & \multicolumn{1}{c |}{\textbf{0.52$\pm$1.04}} & 18.23$\pm$6.74  & \multicolumn{1}{c |}{67.73$\pm$14.55} & \multicolumn{1}{c |}{21.2$\pm$9.8} & 35.37$\pm$10.47  & \multicolumn{1}{c 
|}{27.07$\pm$8.12} & \multicolumn{1}{c |}{26.34$\pm$9.41} & 3.91$\pm$1.11  & \multicolumn{1}{c |}{31.76$\pm$15.67} & \multicolumn{1}{c |}{17.71$\pm$7.79} & 62.24$\pm$8.7  & \multicolumn{1}{c |}{27.55$\pm$9.3} & \multicolumn{1}{c |}{25.
27$\pm$8.49} & \textbf{2.17$\pm$0.62}
\\ \cline{2-18}
& 5 & 34.2$\pm$0.34 & \multicolumn{1}{c |}{\textbf{116.99$\pm$23.29}} & \multicolumn{1}{c |}{\textbf{0.81$\pm$1.75}} & 16.06$\pm$3.83  & \multicolumn{1}{c |}{72.8$\pm$14.04} & \multicolumn{1}{c |}{20.84$\pm$12.02} & 33.46$\pm$9.66  & \multicolumn{1}{c
 |}{34.83$\pm$12.1} & \multicolumn{1}{c |}{24.85$\pm$8.42} & 5.73$\pm$1.27  & \multicolumn{1}{c |}{37.0$\pm$13.79} & \multicolumn{1}{c |}{16.19$\pm$6.47} & 87.31$\pm$17.68  & \multicolumn{1}{c |}{33.3$\pm$10.32} & \multicolumn{1}{c |}{
22.91$\pm$7.73} & \textbf{2.3$\pm$0.95}
\\ \hline
\end{tabular}}
\end{table*}
\begin{table*}[]
\caption{Comparing Lift (in \%), Disparity (in \%), and Runtime (in seconds), on 1000 node Pokec samples, of LP-Advanced and all baselines, varying $k$. 20 trials are used at each setting. Both Friend of Friend and Intersecting Group Count are reported. $p = 0.5$, $|V_s|=3$}
\label{table:varying-k-1000}
\vspace{-4mm}
\resizebox{\textwidth}{!}{ %
\begin{tabular}{ | c | c | c | ccc | ccc | ccc | ccc | ccc |}\hline
\multirow{2}{*}
{\begin{tabular}[c]{ @ {}c @ {}}Candidate\\ Edges\end{tabular}} &
\multirow{2}{*}{k}
&  \multirow{2}{*}{\begin{tabular}[c]{ @ {} c @ {}}Initial\\ Disparity\end{tabular}}
&  \multicolumn{3}{c |}{LP - Advanced}
& \multicolumn{3}{c |}{Continuous Greedy}
&  \multicolumn{3}{c |}{IRFA}
& \multicolumn{3}{c |}{SpGreedy}
& \multicolumn{3}{c |}{ACR - FoF}
\\ \cline{4 - 18}
& & & \multicolumn{1}{c |}{Lift} & \multicolumn{1}{c |}{Disparity  } & Time
& \multicolumn{1}{c |}{Lift} & \multicolumn{1}{c |}{Disparity} & Time
& \multicolumn{1}{c |}{Lift} & \multicolumn{1}{c |}{Disparity} & Time
& \multicolumn{1}{c |}{Lift} & \multicolumn{1}{c |}{Disparity} & Time
& \multicolumn{1}{c |}{Lift} & \multicolumn{1}{c |}{Disparity} & Time
\\ \hline
\multirow{4}{*}{FoF}
& 2 & 21.52$\pm$1.2 & \multicolumn{1}{c |}{\textbf{84.6$\pm$31.66}} & \multicolumn{1}{c |}{\textbf{0.4$\pm$1.3}} & 71.0$\pm$15.51  & \multicolumn{1}{c |}{46.81$\pm$16.56} & \multicolumn{1}{c |}{14.81$\pm$3.8} & 56.53$\pm$16.14  & \multicolumn{1}{c |}{
23.4$\pm$13.69} & \multicolumn{1}{c |}{17.16$\pm$3.66} & \textbf{6.02$\pm$1.95}  & \multicolumn{1}{c |}{21.05$\pm$10.86} & \multicolumn{1}{c |}{12.81$\pm$3.17} & 210.27$\pm$63.42  & \multicolumn{1}{c |}{27.41$\pm$23.96} & \multicolumn{1}{c |}{1
7.31$\pm$4.17} & 21.09$\pm$18.5
\\ \cline{2 - 18}

& 3 & 21.52$\pm$1.2 & \multicolumn{1}{c |}{\textbf{100.4$\pm$40.76}} & \multicolumn{1}{c |}{\textbf{0.23$\pm$0.65}} & 74.85$\pm$22.83  & \multicolumn{1}{c |}{51.97$\pm$16.02} & \multicolumn{1}{c |}{14.7$\pm$3.18} & 52.77$\pm$12.37  & \multicolumn{1}{c
 |}{33.18$\pm$20.1} & \multicolumn{1}{c |}{17.05$\pm$4.31} & \textbf{8.08$\pm$3.0}  & \multicolumn{1}{c |}{30.78$\pm$12.77} & \multicolumn{1}{c |}{11.03$\pm$2.93} & 311.1$\pm$90.71  & \multicolumn{1}{c |}{36.79$\pm$29.53} & \multicolumn{1}{c |}
{16.19$\pm$3.87} & 20.97$\pm$17.07
\\ \cline{2 - 18}

& 4 & 21.52$\pm$1.2 & \multicolumn{1}{c |}{\textbf{109.93$\pm$38.03}} & \multicolumn{1}{c |}{\textbf{0.54$\pm$1.24}} & 72.04$\pm$21.03  & \multicolumn{1}{c |}{56.23$\pm$15.54} & \multicolumn{1}{c |}{15.37$\pm$3.5} & 56.59$\pm$16.73  & \multicolumn{1}{
c |}{39.83$\pm$21.99} & \multicolumn{1}{c |}{16.29$\pm$4.23} & \textbf{13.31$\pm$6.23}  & \multicolumn{1}{c |}{37.67$\pm$20.37} & \multicolumn{1}{c |}{10.28$\pm$2.78} & 400.14$\pm$121.0  & \multicolumn{1}{c |}{43.93$\pm$32.43} & \multicolumn{1}
{c |}{15.29$\pm$3.65} & 21.41$\pm$14.44
\\ \cline{2 - 18}

& 5 & 21.52$\pm$1.2 & \multicolumn{1}{c |}{\textbf{123.06$\pm$44.94}} & \multicolumn{1}{c |}{\textbf{0.49$\pm$1.07}} & 66.04$\pm$17.58  & \multicolumn{1}{c |}{58.71$\pm$15.22} & \multicolumn{1}{c |}{14.53$\pm$3.18} & 52.47$\pm$11.13  & \multicolumn{1}
{c |}{47.65$\pm$23.39} & \multicolumn{1}{c |}{15.26$\pm$2.88} & \textbf{14.09$\pm$5.63}  & \multicolumn{1}{c |}{49.91$\pm$28.26} & \multicolumn{1}{c |}{9.51$\pm$2.34} & 489.08$\pm$166.75  & \multicolumn{1}{c |}{53.84$\pm$41.19} & \multicolumn{1
}{c |}{14.6$\pm$3.88} & 21.64$\pm$19.15
\\ \hline
\multirow{4}{*}{IGC}
& 2 & 21.52$\pm$1.2 & \multicolumn{1}{c |}{\textbf{94.23$\pm$57.62}} & \multicolumn{1}{c |}{\textbf{2.52$\pm$5.01}} & 95.87$\pm$41.96  & \multicolumn{1}{c |}{61.88$\pm$36.96} & \multicolumn{1}{c |}{12.93$\pm$5.29} & 49.06$\pm$18.41  & \multicolumn{1}{
c |}{24.96$\pm$16.3} & \multicolumn{1}{c |}{17.98$\pm$3.39} & \textbf{5.28$\pm$2.33}  & \multicolumn{1}{c |}{22.6$\pm$19.52} & \multicolumn{1}{c |}{13.74$\pm$4.73} & 198.21$\pm$91.44  & \multicolumn{1}{c |}{19.69$\pm$18.18} & \multicolumn{1}{c 
|}{18.2$\pm$4.6} & 33.63$\pm$31.67
\\ \cline{2-18}
& 3 & 21.52$\pm$1.2 & \multicolumn{1}{c |}{\textbf{109.88$\pm$60.69}} & \multicolumn{1}{c |}{\textbf{1.14$\pm$2.19}} & 141.6$\pm$61.67  & \multicolumn{1}{c |}{66.32$\pm$41.25} & \multicolumn{1}{c |}{13.16$\pm$3.58} & 67.38$\pm$23.32  & \multicolumn{1}
{c |}{30.32$\pm$17.49} & \multicolumn{1}{c |}{17.04$\pm$4.23} & \textbf{8.64$\pm$3.57}  & \multicolumn{1}{c |}{38.56$\pm$37.75} & \multicolumn{1}{c |}{11.75$\pm$4.9} & 377.39$\pm$113.57  & \multicolumn{1}{c |}{23.65$\pm$16.25} & \multicolumn{1}
{c |}{17.55$\pm$3.44} & 40.32$\pm$33.28
\\ \cline{2-18}
& 4 & 21.52$\pm$1.2 & \multicolumn{1}{c |}{\textbf{121.99$\pm$64.27}} & \multicolumn{1}{c |}{\textbf{1.35$\pm$2.35}} & 137.22$\pm$76.26  & \multicolumn{1}{c |}{73.2$\pm$45.57} & \multicolumn{1}{c |}{13.01$\pm$3.68} & 66.66$\pm$20.64  & \multicolumn{1}
{c |}{36.55$\pm$19.95} & \multicolumn{1}{c |}{14.92$\pm$4.3} & \textbf{14.37$\pm$5.74}  & \multicolumn{1}{c |}{48.81$\pm$46.7} & \multicolumn{1}{c |}{11.18$\pm$4.59} & 490.52$\pm$134.02  & \multicolumn{1}{c |}{32.0$\pm$19.68} & \multicolumn{1}{
c |}{15.81$\pm$4.2} & 37.7$\pm$32.0
\\ \cline{2-18}
& 5 & 21.52$\pm$1.2 & \multicolumn{1}{c |}{\textbf{133.82$\pm$72.74}} & \multicolumn{1}{c |}{\textbf{0.88$\pm$1.91}} & 133.89$\pm$71.11  & \multicolumn{1}{c |}{72.9$\pm$44.3} & \multicolumn{1}{c |}{13.21$\pm$3.77} & 66.84$\pm$20.57  & \multicolumn{1}{
c |}{43.76$\pm$28.86} & \multicolumn{1}{c |}{16.17$\pm$4.68} & \textbf{15.43$\pm$4.31}  & \multicolumn{1}{c |}{56.26$\pm$47.34} & \multicolumn{1}{c |}{10.09$\pm$5.01} & 589.39$\pm$175.17  & \multicolumn{1}{c |}{40.94$\pm$29.59} & \multicolumn{1
}{c |}{15.57$\pm$3.91} & 44.88$\pm$40.27
\\ \hline
\end{tabular}}
\end{table*}
\begin{table*}[]
\caption{Comparing Lift (in \%), Disparity (in \%), and Runtime (in seconds), on 500 node Antelope Valley data sets, of LP-Advanced and all baselines, varying $p$. 20 trials are used at each setting. Both Friend of Friend and Intersecting Group Count are reported. $k = 3$, $|V_s|=3$}
\vspace{-4mm}
\label{table:varying-p-500}
\resizebox{\textwidth}{!}{ %
\begin{tabular}{ | c | c | c | ccc | ccc | ccc | ccc | ccc |}\hline
\multirow{2}{*}
{\begin{tabular}[c]{ @ {}c @ {}}Candidate\\ Edges\end{tabular}} &
\multirow{2}{*}{p} 
&  \multirow{2}{*}{\begin{tabular}[c]{ @ {} c @ {}}Initial\\ Disparity\end{tabular}}
&  \multicolumn{3}{c |}{LP - Advanced}
& \multicolumn{3}{c |}{Continuous Greedy}
&  \multicolumn{3}{c |}{IRFA}
& \multicolumn{3}{c |}{SpGreedy}
& \multicolumn{3}{c |}{ACR - FoF}
\\ \cline{4 - 18}
& & & \multicolumn{1}{c |}{Lift} & \multicolumn{1}{c |}{Disparity  } & Time
& \multicolumn{1}{c |}{Lift} & \multicolumn{1}{c |}{Disparity} & Time
& \multicolumn{1}{c |}{Lift} & \multicolumn{1}{c |}{Disparity} & Time
& \multicolumn{1}{c |}{Lift} & \multicolumn{1}{c |}{Disparity} & Time
& \multicolumn{1}{c |}{Lift} & \multicolumn{1}{c |}{Disparity} & Time
\\ \hline
\multirow{4}{*}{FoF}
& 0.3 & 71.83$\pm$23.64 & \multicolumn{1}{c |}{\textbf{87.46$\pm$14.25}} & \multicolumn{1}{c |}{\textbf{1.81$\pm$4.97}} & 12.71$\pm$5.56  & \multicolumn{1}{c |}{60.41$\pm$10.65} & \multicolumn{1}{c |}{46.4$\pm$16.75} & 31.68$\pm$9.88  & \multicolumn{1}{
c |}{28.87$\pm$7.05} & \multicolumn{1}{c |}{56.29$\pm$21.74} & 3.96$\pm$1.35  & \multicolumn{1}{c |}{21.47$\pm$6.74} & \multicolumn{1}{c |}{44.72$\pm$13.83} & 43.04$\pm$10.62  & \multicolumn{1}{c |}{25.95$\pm$9.29} & \multicolumn{1}{c 
|}{54.46$\pm$22.14} & \textbf{1.88$\pm$0.91}
\\ \cline{2 - 18}
& 0.5 & 34.2$\pm$0.34 & \multicolumn{1}{c |}{\textbf{80.25$\pm$12.1}} & \multicolumn{1}{c |}{\textbf{0.42$\pm$1.0}} & 15.57$\pm$17.33  & \multicolumn{1}{c |}{51.94$\pm$8.82} & \multicolumn{1}{c |}{25.58$\pm$4.99} & 29.44$\pm$7.96  & \multicolumn{1}{c |}
{25.48$\pm$10.53} & \multicolumn{1}{c |}{26.92$\pm$8.62} & 3.55$\pm$1.59  & \multicolumn{1}{c |}{20.11$\pm$9.03} & \multicolumn{1}{c |}{21.08$\pm$6.1} & 42.01$\pm$9.7  & \multicolumn{1}{c |}{27.06$\pm$9.0} & \multicolumn{1}{c |}{24.98$\pm$7.13} & \textbf{2.18$\pm$0.95}
\\ \cline{2 - 18}
& 0.7 & 14.62$\pm$3.93 & \multicolumn{1}{c |}{\textbf{50.37$\pm$9.14}} & \multicolumn{1}{c |}{\textbf{0.08$\pm$0.21}} & 9.95$\pm$4.9  & \multicolumn{1}{c |}{29.13$\pm$4.37} & \multicolumn{1}{c |}{11.52$\pm$4.13} & 32.32$\pm$9.24  & \multicolumn{1}{c |}{
16.27$\pm$5.24} & \multicolumn{1}{c |}{13.11$\pm$6.05} & 3.84$\pm$1.53  & \multicolumn{1}{c |}{12.26$\pm$5.87} & \multicolumn{1}{c |}{9.26$\pm$4.06} & 36.4$\pm$8.88  & \multicolumn{1}{c |}{19.64$\pm$5.85} & \multicolumn{1}{c |}{10.45$\pm$4.09} & \textbf{1.51$\pm$0.69}
\\ \cline{2 - 18}
& 0.9 & 4.07$\pm$2.09 & \multicolumn{1}{c |}{\textbf{15.1$\pm$2.76}} & \multicolumn{1}{c |}{\textbf{0.48$\pm$0.85}} & 14.84$\pm$16.96  & \multicolumn{1}{c |}{8.7$\pm$1.08} & \multicolumn{1}{c |}{3.66$\pm$2.1} & 25.53$\pm$11.39  & \multicolumn{1}{c |}{6.
11$\pm$2.13} & \multicolumn{1}{c |}{3.41$\pm$2.18} & 3.39$\pm$1.47  & \multicolumn{1}{c |}{4.01$\pm$1.68} & \multicolumn{1}{c |}{2.84$\pm$1.84} & 28.07$\pm$9.99  & \multicolumn{1}{c |}{8.34$\pm$2.21} & \multicolumn{1}{c |}{2.82$\pm$1.1
9} & \textbf{1.09$\pm$0.82}
\\ \hline
\multirow{4}{*}{IGC}
& 0.3 & 71.83$\pm$23.64 & \multicolumn{1}{c |}{\textbf{86.45$\pm$17.63}} & \multicolumn{1}{c |}{\textbf{2.18$\pm$4.88}} & 14.17$\pm$3.49  & \multicolumn{1}{c |}{68.84$\pm$13.63} & \multicolumn{1}{c |}{41.86$\pm$16.89} & 37.06$\pm$10.2  & \multicolumn{1}
{c |}{31.18$\pm$13.43} & \multicolumn{1}{c |}{54.57$\pm$24.55} & 2.99$\pm$1.01  & \multicolumn{1}{c |}{26.65$\pm$11.2} & \multicolumn{1}{c |}{41.46$\pm$14.73} & 54.23$\pm$11.65  & \multicolumn{1}{c |}{20.83$\pm$7.33} & \multicolumn{1}{
c |}{57.12$\pm$21.12} & \textbf{2.57$\pm$0.99}
\\ \cline{2-18}
& 0.5 & 34.2$\pm$0.34 & \multicolumn{1}{c |}{\textbf{83.69$\pm$13.9}} & \multicolumn{1}{c |}{\textbf{0.39$\pm$1.0}} & 16.69$\pm$20.23  & \multicolumn{1}{c |}{60.19$\pm$10.61} & \multicolumn{1}{c |}{20.44$\pm$8.15} & 29.03$\pm$7.98  & \multicolumn{1}{c |
}{26.15$\pm$10.93} & \multicolumn{1}{c |}{30.29$\pm$8.95} & 2.91$\pm$0.7  & \multicolumn{1}{c |}{23.62$\pm$11.58} & \multicolumn{1}{c |}{19.89$\pm$7.09} & 52.18$\pm$10.99  & \multicolumn{1}{c |}{20.04$\pm$6.29} & \multicolumn{1}{c |}{2
6.28$\pm$7.18} & \textbf{2.25$\pm$0.67}
\\ \cline{2-18}
& 0.7 & 14.62$\pm$3.93 & \multicolumn{1}{c |}{\textbf{49.24$\pm$7.84}} & \multicolumn{1}{c |}{\textbf{0.56$\pm$1.5}} & 10.77$\pm$2.95  & \multicolumn{1}{c |}{35.79$\pm$7.26} & \multicolumn{1}{c |}{10.4$\pm$5.15} & 35.62$\pm$9.39  & \multicolumn{1}{c |}{
13.58$\pm$3.14} & \multicolumn{1}{c |}{12.53$\pm$4.43} & 3.91$\pm$1.01  & \multicolumn{1}{c |}{15.33$\pm$8.89} & \multicolumn{1}{c |}{8.38$\pm$3.9} & 51.63$\pm$9.15  & \multicolumn{1}{c |}{13.84$\pm$4.33} & \multicolumn{1}{c |}{10.95$\pm$4.06} & \textbf{1.98$\pm$0.69}
\\ \cline{2-18}
& 0.9 & 4.07$\pm$2.09 & \multicolumn{1}{c |}{\textbf{15.21$\pm$2.56}} & \multicolumn{1}{c |}{\textbf{0.34$\pm$0.77}} & 24.08$\pm$29.55  & \multicolumn{1}{c |}{11.07$\pm$2.1} & \multicolumn{1}{c |}{3.33$\pm$2.05} & 29.7$\pm$8.09  & \multicolumn{1}{c |}{4
.76$\pm$2.15} & \multicolumn{1}{c |}{3.48$\pm$2.06} & 2.81$\pm$1.01  & \multicolumn{1}{c |}{4.57$\pm$2.0} & \multicolumn{1}{c |}{3.01$\pm$2.03} & 47.51$\pm$9.84  & \multicolumn{1}{c |}{5.72$\pm$2.04} & \multicolumn{1}{c |}{3.01$\pm$1.6
6} & \textbf{2.38$\pm$0.99}
\\ \hline
\end{tabular}}
\end{table*}
\begin{table*}[]
\caption{Comparing Lift (in \%), Disparity (in \%), and Runtime (in seconds), on 1000 node Pokec samples, of LP-Advanced and all baselines, varying $p$. 20 trials are used at each setting. Both Friend of Friend and Intersecting Group Count are reported. $k = 3$, $|V_s|=3$}
\label{table:varying-p-1000}
\vspace{-4mm}
\resizebox{\textwidth}{!}{ %
\begin{tabular}{ | c | c | c | ccc | ccc | ccc | ccc | ccc |}\hline
\multirow{2}{*}
{\begin{tabular}[c]{ @ {}c @ {}}Candidate\\ Edges\end{tabular}} &
\multirow{2}{*}{p} & 
\multirow{2}{*}{\begin{tabular}[c]{ @ {} c @ {}}Initial\\ Disparity\end{tabular}}
&  \multicolumn{3}{c |}{LP - Advanced}
& \multicolumn{3}{c |}{Continuous Greedy}
&  \multicolumn{3}{c |}{IRFA}
& \multicolumn{3}{c |}{SpGreedy}
& \multicolumn{3}{c |}{ACR - FoF}
\\ \cline{4 - 18}
& & & \multicolumn{1}{c |}{Lift} & \multicolumn{1}{c |}{Disparity  } & Time
& \multicolumn{1}{c |}{Lift} & \multicolumn{1}{c |}{Disparity} & Time
& \multicolumn{1}{c |}{Lift} & \multicolumn{1}{c |}{Disparity} & Time
& \multicolumn{1}{c |}{Lift} & \multicolumn{1}{c |}{Disparity} & Time
& \multicolumn{1}{c |}{Lift} & \multicolumn{1}{c |}{Disparity} & Time
\\ \hline
\multirow{4}{*}{FoF}
& 0.3 & 62.92$\pm$25.82 & \multicolumn{1}{c |}{\textbf{93.55$\pm$36.44}} & \multicolumn{1}{c |}{\textbf{9.62$\pm$23.26}} & 70.64$\pm$22.78  & \multicolumn{1}{c |}{61.65$\pm$17.82} & \multicolumn{1}{c |}{37.91$\pm$17.06} & 48.45$\pm$9.44  & \multicolumn{
1}{c |}{40.53$\pm$18.3} & \multicolumn{1}{c |}{41.84$\pm$17.59} & \textbf{8.44$\pm$3.57}  & \multicolumn{1}{c |}{31.88$\pm$13.2} & \multicolumn{1}{c |}{35.46$\pm$20.56} & 283.91$\pm$112.87  & \multicolumn{1}{c |}{33.8$\pm$25.25} & \multicolumn{
1}{c |}{46.94$\pm$15.32} & 21.58$\pm$19.6
\\ \cline{2 - 18}

& 0.5 & 21.52$\pm$1.2 & \multicolumn{1}{c |}{\textbf{97.67$\pm$34.85}} & \multicolumn{1}{c |}{\textbf{0.13$\pm$0.53}} & 68.58$\pm$20.51  & \multicolumn{1}{c |}{53.58$\pm$16.69} & \multicolumn{1}{c |}{14.86$\pm$3.84} & 54.17$\pm$9.07  & \multicolumn{1}{c
 |}{33.16$\pm$17.94} & \multicolumn{1}{c |}{18.09$\pm$4.0} & \textbf{9.46$\pm$4.31}  & \multicolumn{1}{c |}{30.78$\pm$12.77} & \multicolumn{1}{c |}{11.03$\pm$2.93} & 293.93$\pm$103.49  & \multicolumn{1}{c |}{36.79$\pm$29.53} & \multicolumn{1}{c
 |}{16.19$\pm$3.87} & 23.0$\pm$20.99
\\ \cline{2 - 18}

& 0.7 & 8.43$\pm$2.46 & \multicolumn{1}{c |}{\textbf{58.03$\pm$20.8}} & \multicolumn{1}{c |}{\textbf{0.08$\pm$0.34}} & 74.43$\pm$32.97  & \multicolumn{1}{c |}{29.24$\pm$7.95} & \multicolumn{1}{c |}{6.17$\pm$2.21} & 49.63$\pm$8.95  & \multicolumn{1}{c |}
{20.31$\pm$11.63} & \multicolumn{1}{c |}{6.59$\pm$2.21} & \textbf{9.55$\pm$3.46}  & \multicolumn{1}{c |}{21.07$\pm$13.59} & \multicolumn{1}{c |}{4.58$\pm$1.37} & 283.88$\pm$98.09  & \multicolumn{1}{c |}{23.08$\pm$17.19} & \multicolumn{1}{c |}{6
.46$\pm$2.79} & 18.93$\pm$15.25
\\ \cline{2 - 18}

& 0.9 & 2.16$\pm$1.03 & \multicolumn{1}{c |}{\textbf{16.51$\pm$4.75}} & \multicolumn{1}{c |}{\textbf{0.06$\pm$0.18}} & 47.52$\pm$15.84  & \multicolumn{1}{c |}{8.5$\pm$1.8} & \multicolumn{1}{c |}{1.71$\pm$0.85} & 35.08$\pm$10.8  & \multicolumn{1}{c |}{5.
25$\pm$2.31} & \multicolumn{1}{c |}{1.75$\pm$0.91} & \textbf{5.65$\pm$2.54}  & \multicolumn{1}{c |}{6.03$\pm$2.38} & \multicolumn{1}{c |}{1.26$\pm$0.56} & 185.19$\pm$59.9  & \multicolumn{1}{c |}{7.36$\pm$4.86} & \multicolumn{1}{c |}{1.69$\pm$0.
94} & 13.16$\pm$10.99
\\ \hline
\multirow{4}{*}{IGC}
& 0.3 & 62.92$\pm$25.82 & \multicolumn{1}{c |}{\textbf{96.05$\pm$37.42}} & \multicolumn{1}{c |}{\textbf{9.85$\pm$23.4}} & 70.01$\pm$19.33  & \multicolumn{1}{c |}{72.25$\pm$32.56} & \multicolumn{1}{c |}{33.4$\pm$12.95} & 62.88$\pm$18.65  & \multicolumn{1
}{c |}{38.81$\pm$19.8} & \multicolumn{1}{c |}{43.53$\pm$19.17} & \textbf{8.68$\pm$3.06}  & \multicolumn{1}{c |}{36.62$\pm$27.76} & \multicolumn{1}{c |}{32.47$\pm$16.73} & 351.33$\pm$132.05  & \multicolumn{1}{c |}{24.82$\pm$17.68} & \multicolumn
{1}{c |}{48.78$\pm$20.72} & 41.8$\pm$40.85
\\ \cline{2-18}
& 0.5 & 21.52$\pm$1.2 & \multicolumn{1}{c |}{\textbf{100.23$\pm$40.08}} & \multicolumn{1}{c |}{\textbf{0.13$\pm$0.55}} & 73.25$\pm$17.7  & \multicolumn{1}{c |}{65.76$\pm$41.34} & \multicolumn{1}{c |}{13.62$\pm$4.49} & 64.79$\pm$16.75  & \multicolumn{1}{
c |}{31.56$\pm$21.74} & \multicolumn{1}{c |}{16.58$\pm$3.62} & \textbf{9.89$\pm$3.73}  & \multicolumn{1}{c |}{38.56$\pm$37.75} & \multicolumn{1}{c |}{11.75$\pm$4.9} & 359.06$\pm$120.53  & \multicolumn{1}{c |}{23.65$\pm$16.25} & \multicolumn{1}{
c |}{17.55$\pm$3.44} & 39.87$\pm$35.98
\\ \cline{2-18}
& 0.7 & 8.43$\pm$2.46 & \multicolumn{1}{c |}{\textbf{57.62$\pm$20.63}} & \multicolumn{1}{c |}{\textbf{0.06$\pm$0.21}} & 74.58$\pm$28.19  & \multicolumn{1}{c |}{37.52$\pm$20.75} & \multicolumn{1}{c |}{5.57$\pm$2.37} & 58.01$\pm$17.58  & \multicolumn{1}{c
 |}{19.43$\pm$11.0} & \multicolumn{1}{c |}{6.69$\pm$2.64} & \textbf{9.3$\pm$3.25}  & \multicolumn{1}{c |}{22.75$\pm$21.01} & \multicolumn{1}{c |}{5.07$\pm$2.85} & 332.06$\pm$130.84  & \multicolumn{1}{c |}{14.55$\pm$9.57} & \multicolumn{1}{c |}{
7.17$\pm$2.63} & 39.41$\pm$32.5
\\ \cline{2-18}
& 0.9 & 2.16$\pm$1.03 & \multicolumn{1}{c |}{\textbf{16.59$\pm$4.79}} & \multicolumn{1}{c |}{\textbf{0.09$\pm$0.2}} & 42.45$\pm$12.17  & \multicolumn{1}{c |}{10.52$\pm$4.75} & \multicolumn{1}{c |}{1.66$\pm$0.84} & 37.92$\pm$14.07  & \multicolumn{1}{c |}
{6.03$\pm$2.53} & \multicolumn{1}{c |}{1.7$\pm$0.89} & \textbf{5.65}$\pm$2.62  & \multicolumn{1}{c |}{6.9$\pm$4.65} & \multicolumn{1}{c |}{1.45$\pm$0.88} & 197.04$\pm$79.45  & \multicolumn{1}{c |}{4.66$\pm$2.68} & \multicolumn{1}{c |}{1.89$\pm$
0.96} & 21.83$\pm$19.2
\\ \hline
\end{tabular}}
\end{table*}

\stitle{\scalable, choice of hyper-parameter} 
\new{In this experiment we examine the choice of Nodes per Iteration (NpI) on the \scalable algorithm.} In each experiment, we ran \scalable on 10,000 nodes, with 5 content nodes, \new{$p=0.5$, and observed over different values of $k$ from 2-5}. We vary NpI between 200, 400, 600, and 800. For each experiment we run 5 trials and report the average. The initial average disparity across these 5 trials was $\sim7.84\%$.

\new{As demonstrated in Figure \ref{fig:ff-disparity}, as the NpI increases, the trend is that \scalable achieves a lower disparity. Equivalently, as demonstrated in Figure \ref{fig:ff-lift}, when the NpI increases, \scalable also produces a higher lift. It is thus obvious that higher NPI results in a more correct output. To balance this, as shown in \ref{fig:ff-time}, as the NpI increases, so does the time required to produce a selection. This demonstrates that the increased lift and fairness the edge selection for higher NpI comes at a cost of increased runtime.}

\begin{figure*}[!tb] 
    \begin{minipage}[t]{0.23\linewidth}
        \centering
        \includegraphics[width=\textwidth]{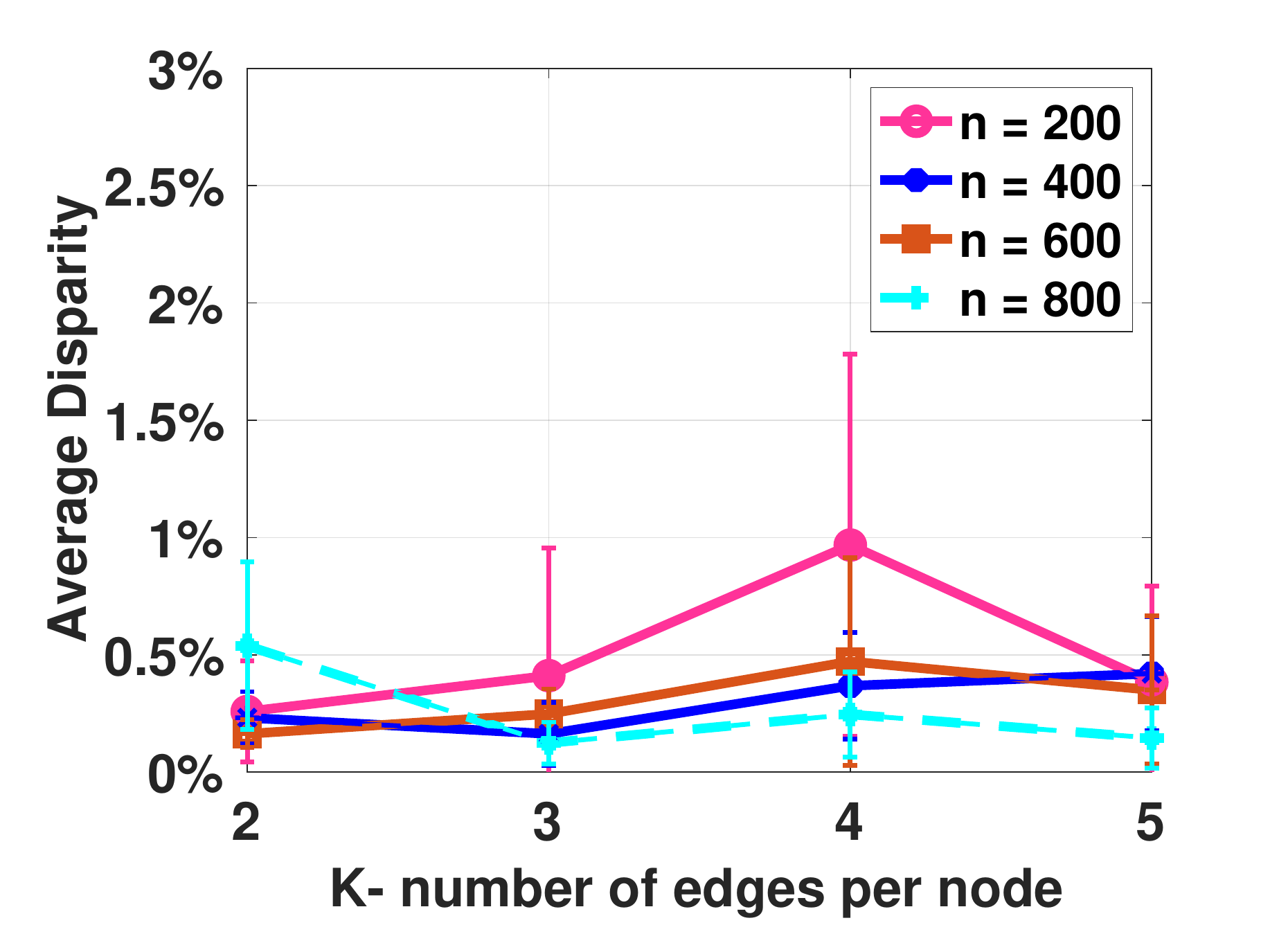}
        \caption{\small Disparity of \scalable for diff. \# nodes per iteration}
        \label{fig:ff-disparity}
    \end{minipage}
    \hfill
    \begin{minipage}[t]{0.23\linewidth}
        \centering
        \includegraphics[width=\textwidth]{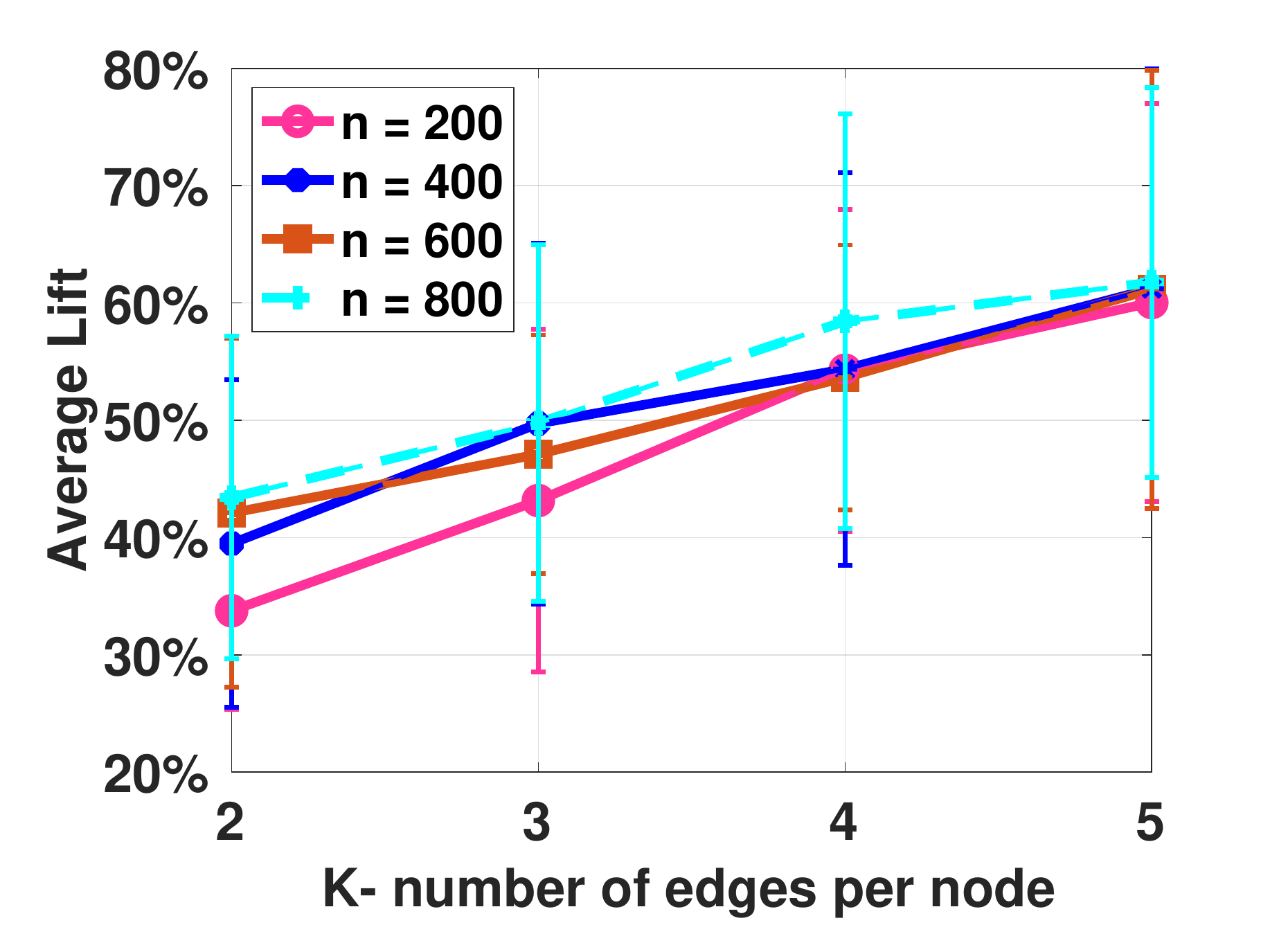}
        \caption{\small The lift of the \scalable for diff. \# nodes per iteration}
        \label{fig:ff-lift}
    \end{minipage}
    \hfill
    \begin{minipage}[t]{0.23\linewidth}
        \centering
        \includegraphics[width=\textwidth]{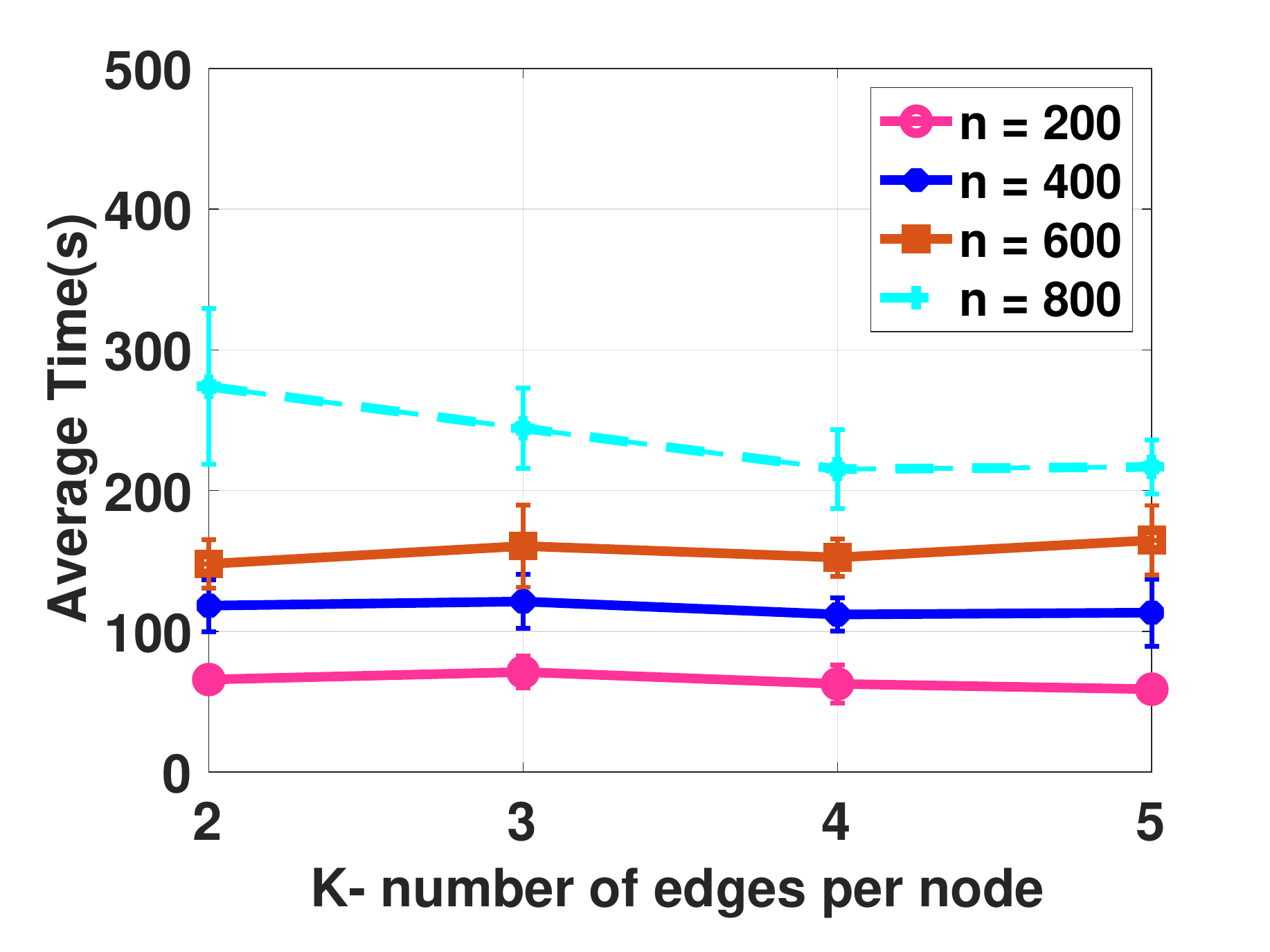}
        \caption{\small The runtime of the \scalable for diff. \# nodes per iteration}
        \label{fig:ff-time}
    \end{minipage}
\end{figure*}


\end{document}